\documentclass{llncs}
\usepackage{a4wide}

\usepackage{xcolor,pgf,tikz,pgflibraryarrows,pgffor,pgflibrarysnakes}
\usetikzlibrary{automata,shapes,snakes,backgrounds,decorations.pathreplacing}
\usepackage{bbold}
\usepackage{lscape}

\usepackage{amsmath}

\usepackage{amssymb,stmaryrd}
\usepackage{algorithm}
\usepackage[noend]{algorithmic}
\usepackage{comment}
\usepackage{times}
\usepackage{url}
\usepackage{wrapfig}
\usepackage{placeins}
\usepackage{enumerate}
\usepackage{algorithm}
\usepackage{algorithmic}
\usepackage{multirow}

\title{On the complexity of heterogeneous multidimensional \\quantitative games\thanks{This work has been partly supported by ERC Starting Grant (279499: inVEST)  and European project Cassting (FP7-ICT-601148).}}

\author{V\'eronique Bruy\`ere$^1$ \and Quentin Hautem$^1$\thanks{Author supported by FRIA fellowship} \and Jean-Fran\c{c}ois Raskin$^2$}	 
\institute{$^1$Universit\'e de Mons$\quad$ $^2$Universit\'e Libre de Bruxelles}

\newcommand{\N}{\mathbb{N}}
\newcommand{\Z}{\mathbb{Z}}
\newcommand{\Q}{\mathbb{Q}}

\newcommand{\ssetminus}{\! \setminus \!}

\newcommand{\G}{G}   
\newcommand{\Gdev}{(V_1,V_2, E, w)}   
\newcommand{\n}{n}   
\newcommand{\Plays}{\mathsf{Plays}}
\newcommand{\rhofactor}[2]{\rho_{[{#1},{#2}]}}

\newcommand{\Hist}{\mathsf{Hist}}
\newcommand{\Out}{\mathsf{Out}}
\newcommand{\Obj}{\Omega}  
\newcommand{\WinG}[3]{{\mathsf{Win}}_{#1}^{#2}({#3})}  
\newcommand{\Win}[2]{{\mathsf{Win}}_{#1}^{#2}}
\newcommand{\OutDev}{\mathsf{Out}(v_0,\sigma_1,\sigma_2)}

\newcommand{\Reach}{\mathsf{Reach}}
\newcommand{\Safe}{\mathsf{Safe}}
\newcommand{\Buchi}{\mathsf{Buchi}}
\newcommand{\CoBuchi}{\mathsf{CoBuchi}}
\newcommand{\GenReach}{\mathsf{GenReach}}
\newcommand{\GenBuchi}{\mathsf{GenBuchi}}
\newcommand{\WISL}{\{\DFWMP$, $\Inf$, $\Sup$, $\LimInf$, $\LimSup\}}
\newcommand{\WISLmath}{\{\DFWMP,\Inf,\Sup,\LimInf,\LimSup\}}

\newcommand{\ISL}{\{\Inf$, $\Sup$, $\LimInf$, $\LimSup\}}

\newcommand{\Reg}{\{\Reach$, $\Safe$, $\Buchi$, $\CoBuchi\}}

\newcommand{\GDReach}[2]{\mathsf{GDEnd}^{#1}({#2})}
\newcommand{\ICWReach}[2]{\mathsf{ICWEnd}^{#1}({#2})}

\newcommand{\Objf}{\mathsf{F^{\lambda}}}
\newcommand{\Objg}{\mathsf{End^{\lambda}}}
\newcommand{\prob}{synthesis problem}

\newcommand{\AlgoGDReach}{\mathsf{GDEnd}}
\newcommand{\AlgoICWReach}{\mathsf{ICWEnd}}

\newcommand{\SolveObjg}{\mathsf{End}}
\newcommand{\SolveObjf}{\mathsf{ObjF}}

\newcommand{\Sup}{\mathsf{Sup}}
\newcommand{\Inf}{\mathsf{Inf}}
\newcommand{\LimSup}{\mathsf{LimSup}}
\newcommand{\LimInf}{\mathsf{LimInf}}
\newcommand{\DFWMP}{\mathsf{WMP}}

\newcommand{\TP}{\mathsf{TP}}
\newcommand{\MP}{\mathsf{MP}}

\newcommand{\Attr}[3]{{\mathsf{Attr}}_{#1}({#2},{#3})}

\newcommand{\lwindow}[1]{$\lambda$-window at position~${#1}$}
\newcommand{\window}[1]{window at position~${#1}$}
\newcommand{\closed}[2]{closed in~${#2}$}
\newcommand{\iclosed}[2]{inductively-closed in~${#2}$}
\newcommand{\fclosed}[2]{first-closed in~${#2}$}
\newcommand{\ffclosed}[1]{first-closed}
\newcommand{\good}[1]{good}
\newcommand{\bad}[1]{bad}

\newcommand{\GoodDec}[1]{${#1}$-good decomposition}

\begin{document}
  \maketitle	

\begin{abstract} 
In this paper, we study two-player zero-sum turn-based games played on a finite multidimensional weighted graph. In recent papers all dimensions use the same measure, whereas here we allow to combine different measures.  Such heterogeneous multidimensional quantitative games provide a general and natural model for the study of reactive system synthesis. We focus on classical measures like the $\Inf$, $\Sup$, $\LimInf$, and $\LimSup$ of the weights seen along the play, as well as on the window mean-payoff ($\DFWMP$) measure recently introduced in~\cite{Chatterjee0RR15}. This new measure is a natural strengthening of the mean-payoff measure. We allow objectives defined as Boolean combinations of heterogeneous constraints. While multidimensional games with Boolean combinations of mean-payoff constraints are undecidable~\cite{Velner15}, we show that the problem becomes ${\sf EXPTIME}$-complete for DNF/CNF Boolean combinations of  heterogeneous measures taken among $\WISL$ and that exponential memory strategies are sufficient for both players to win. We provide a detailed study of the complexity and the memory requirements when the Boolean combination of the measures is replaced by an intersection. ${\sf EXPTIME}$-completeness and exponential memory strategies still hold for the intersection of measures in $\WISL$, and we get ${\sf PSPACE}$-completeness when $\DFWMP$ measure is no longer considered. To avoid ${\sf EXPTIME}$- or ${\sf PSPACE}$-hardness, we impose at most one occurrence of $\DFWMP$ measure and fix the number of $\Sup$ measures, and we propose several refinements (on the number of occurrences of the other measures) for which we get polynomial algorithms and lower memory requirements. For all the considered classes of games, we also study parameterized complexity.

\end{abstract}

\section{Introduction}

Two-player zero-sum turn-based games played on graphs are an adequate mathematical model to solve the {\em reactive synthesis problem}~\cite{PnueliR89}. To model systems with resource constraints, like embedded systems, games with quantitative objectives have been studied, e.g. mean-payoff~\cite{ZwickP96} and energy games~\cite{BouyerFLMS08}. In~\cite{ChatterjeeDHR10,VelnerR11,Chatterjee0RR15,VelnerC0HRR15}, multidimensional games with conjunctions of several quantitative objectives have been investigated, such that all dimensions use the {\em same} measure. Here, we study games played on multidimensional weighted graphs with objectives defined from different measures over the dimensions. Such {\em heterogeneous} quantitative games provide a general, natural, and expressive model for the reactive synthesis problem.

We study Boolean combinations of constraints over classical measures like the $\Inf$, $\Sup$, $\LimInf$, and $\LimSup$ of the weights seen along the play, as well as the {\em window mean-payoff} ($\DFWMP$) measure introduced in~\cite{Chatterjee0RR15}. For instance one can ask that along the play the $\LimInf$ of the weights on a first component is positive while either the $\Sup$ of the weights on a second component or their $\DFWMP$ is less than $\frac{2}{3}$. While the mean-payoff ($\MP$) measure considers the long-run average of the weights along the whole play, the $\DFWMP$ measure considers weights over a {\em local window of a given size} sliding along the play. A $\DFWMP$ objective asks now to ensure that the average weight satisfies a given constraint over every bounded window. This is a strengthening of the $\MP$ objective: winning for the $\DFWMP$ objective implies winning for the $\MP$ objective. Also, any finite-memory strategy that forces an $\MP$ measure larger than threshold $\nu + \epsilon$ (for any $\epsilon >0$), also forces the $\DFWMP$ measure to be larger than $\nu$ provided that the window size is taken large enough. Aside from their naturalness, $\DFWMP$ objectives are algorithmically more tractable than classical $\MP$ objectives, see~\cite{Chatterjee0RR15,HunterPR15}. 
First, unidimensional $\DFWMP$ games can be solved in polynomial time when working with polynomial windows~\cite{Chatterjee0RR15} while only pseudo-polynomial time algorithms are known for mean-payoff games~\cite{ZwickP96,BrimCDGR11}.
Second, multidimensional games with Boolean combinations of $\MP$ objectives are undecidable~\cite{Velner15}, whereas we show here that games with Boolean combinations of $\DFWMP$ objectives and other classical objectives are decidable.

\paragraph{{\bf Main contributions~}}

More precisely (see also Table~\ref{table:conclusion}), we show in this paper that the problem is $\sf EXPTIME$-complete for CNF and DNF Boolean combinations of heterogeneous measures taken among  $\WISL$.
We provide a detailed study of the complexity when the Boolean combination of the measures is replaced by an intersection, as it is often natural in practice to consider conjunction of constraints. $\sf EXPTIME$-completeness of the problem still holds for the intersection of measures in $\WISL$, and we get $\sf PSPACE$-completeness when $\DFWMP$ measure is not considered.
To avoid $\sf EXPTIME$-hardness, we consider fragments where there is at most one occurrence of a  $\DFWMP$ measure. In case of intersections of one $\DFWMP$ objective with any number of objectives of one kind among $\ISL$ (this number must be fixed in case of objectives $\Sup$), we get  $\sf P$-completeness when dealing with polynomial windows, a reasonable hypothesis in practical applications. In case of no occurrence of $\DFWMP$ measure, we propose several refinements (on the number of occurrences of the other measures) for which we get $\sf P$-completeness.
Some of our results are obtained by reductions to known qualitative games but most of them are obtained by new algorithms that require genuine ideas to handle in an optimal way one $\DFWMP$ objective together with qualitative objectives such as safety, reachability, B\"uchi and coB\"uchi objectives.
Finally, in all our results, we provide a careful analysis of the memory requirements of winning strategies for both players, and we study parameterized complexity.


\begin{table}
\scriptsize
\begin{center}
\begin{tabular}{|c|c|c|c|}
\hline
Objectives & Complexity class  & Player~$1$ memory & Player~$2$ memory \\
\hline
\hline
\rule{0cm}{0.35cm} (CNF/DNF) Boolean combination of $\underline\MP$, $\overline\MP$  \cite{Velner15}  & Undecidable                       &  infinite                  & infinite  \\
\hline
\rule{0cm}{0.25cm} (CNF/DNF) Boolean combinaison of   &  \multirow{4}{*}{$\sf EXPTIME$-complete}  & \multicolumn{2}{c|}{} \\
\rule{0cm}{0.25cm} $\DFWMP$, $\Inf$, $\Sup$, $\LimInf$, $\LimSup$ (*) &  & \multicolumn{2}{c|}{\multirow{3}{*}{exponential}} \\
\cline{0-0}
\rule{0cm}{0.3cm} Intersection of $\DFWMP$, $\Inf$, $\Sup$, $\LimInf$, $\LimSup$ (*) &    &   \multicolumn{2}{c|}{} \\
\cline{0-0}
Intersection of $\DFWMP$ \cite{Chatterjee0RR15} & & \multicolumn{2}{c|}{} \\
\cline{0-1}
\rule{0cm}{0.3cm} Intersection of $\Inf$, $\Sup$, $\LimInf$, $\LimSup$ (*) & $\sf PSPACE$-complete         &  \multicolumn{2}{c|}{}  \\
\cline{2-4}
and refinements (*) & \multicolumn{3}{c|}{See Table~\ref{table:ISL} } \\
\hline
\rule{0cm}{0.2cm} Intersection of $\underline\MP$ \cite{VelnerC0HRR15} & $\sf coNP$-complete      &  \multirow{2}{*}{infinite} & \multirow{3}{*}{memoryless}  \\
\cline{0-1}
\rule{0cm}{0.35cm} Intersection of $\overline\MP$ \cite{VelnerC0HRR15}  & \multirow{2}{*}{$\sf NP\cap coNP$}   &  &  \\
\cline{0-0} \cline{3-3}
Unidimensional $\sf MP$ \cite{ZwickP96,BrimCDGR11}  &    & memoryless &  \\
\hline 
\rule{0cm}{0.23cm} Unidimensional $\DFWMP$ \cite{Chatterjee0RR15} & $\sf P$-complete & \multicolumn{2}{c|}{\multirow{2}{*}{pseudo-polynomial}} \\
\rule{0cm}{0.22cm} $\DFWMP \cap \Obj$ with $\Obj \in \ISL$ (*) & (Polynomial windows) & \multicolumn{2}{c|}{} \\
\hline
\rule{0cm}{0.3cm} Unidimensional $\Inf$, $\Sup$, $\LimInf$, $\LimSup$ \cite{2001automata} & $\sf P$-complete & \multicolumn{2}{c|}{memoryless} \\
\hline
\end{tabular}
\end{center}
\caption{Overview - Our results are marked with {\scriptsize (*)}}
\label{table:conclusion}
\end{table}

\paragraph{{\bf Related work~}}
Multidimensional mean-payoff games have been studied in~\cite{VelnerC0HRR15}. Conjunction of lim inf mean-payoff ($\underline\MP$) objectives are $\sf coNP$-complete, conjunctions of lim sup $\MP$ ($\overline\MP$) objectives are in $\sf NP$ and in $\sf coNP$, and infinite-memory strategies are necessary for player~$1$ whereas player~$2$ can play memoryless. The general case of Boolean combinations of $\underline\MP$ and $\overline\MP$ have been shown undecidable in~\cite{Velner15}. Multidimensional energy games were studied in~\cite{ChatterjeeDHR10,ChatterjeeRR14} for unfixed initial credit and in~\cite{JurdzinskiLS15} for fixed initial credit. For unfixed initial credit, they are $\sf coNP$-complete and exponential memory strategies are necessary and sufficient to play optimally. For fixed initial credit, they are $\sf 2EXPTIME$-complete and doubly exponential memory is necessary and sufficient to play optimally. Energy games and mean-payoff games with imperfect information, that generalize multidimensional energy and mean-payoff games, have been studied in~\cite{DegorreDGRT10} and shown undecidable.

The $\DFWMP$ measure was first introduced in~\cite{Chatterjee0RR15}. Unidimensional $\DFWMP$ games can be solved in polynomial time, and multidimensional $\DFWMP$ games are $\sf EXPTIME$-complete. In~\cite{Chatterjee0RR15}, the $\DFWMP$ measure is considered on all the dimensions, with no conjunction with other measures like $\Inf$, $\Sup$, $\LimInf$, and $\LimSup$, and the case of Boolean combinations of $\DFWMP$ objectives is not investigated.

Games with objectives expressed in fragments of $\sf LTL$ have been studied in~\cite{AlurTM03}. Our result that games with intersection of objectives in $\ISL$ are in $\sf PSPACE$ can be obtained by reduction to some of these fragments. But we here propose a simple proof adapted to our context, that allows to identify several polynomial fragments. Our other results cannot be obtained in this way and require genuine techniques and new algorithmic ideas.

\paragraph{{\bf Structure of the paper~}}
In Section~\ref{sec:preliminaries}, we fix definitions and notations, and recall known results. 
In Section~\ref{sec:general}, we study the complexity of multidimensional games with an intersection of objectives taken among $\WISL$.
In Section~\ref{sec:OneWindow}, we identify a polynomial fragment when only one $\DFWMP$ objective is allowed in the intersection.
In Section~\ref{sec:regular}, we study intersections that do not use $\DFWMP$ objectives, and we identify some additional polynomial fragments.
In Section~\ref{sec:beyond}, we establish the complexity of general Boolean combinations (instead of intersections) of objectives in $\WISL$.
In Section~\ref{sec:para}, we study the parameterized complexity of our results of Sections~\ref{sec:general}, \ref{sec:regular} and \ref{sec:beyond}. Finally, in Section~\ref{sec:conclu}, we give a conclusion.

\section{Preliminaries} \label{sec:preliminaries}

We consider turn-based games played by two players on a finite multidimensional weighted directed graph.

\begin{definition}
A \emph{multi-weighted two-player game structure} (or simply game structure) is a tuple $\G = \Gdev$ where
\begin{itemize}
\item $(V,E)$ is a finite directed graph, with $V$ the set of vertices and $E \subseteq V \times V$ the set of edges such that for each $v \in V$, there exists $(v,v') \in E$ for some $v' \in V$ (no deadlock),
\item $(V_1,V_2)$ forms a partition of $V$ such that $V_p$ is the set of vertices controlled by player $p \in \{1,2\}$,
\item $w : E \rightarrow \Z^{\n}$ is the $\n$-dimensional weight function that associates a vector of $\n$ weights to each edge, for some $\n \geq 1$.
\end{itemize}
\end{definition}

We also say that $\G$ is an $n$-weighted game structure. The opponent of player $p \in \{1,2\}$ is denoted by $\overline{p}$. A \emph{play} of $\G$ is an infinite sequence $\rho = \rho_0 \rho_1 \ldots \in V^{\omega}$ such that $(\rho_k,\rho_{k+1}) \in E$ for all $k \in \N$. We denote by $\Plays(\G)$ the set of plays in $\G$. \emph{Histories} of $\G$ are finite sequences $\rho = \rho_0 \ldots \rho_i \in V^+$ defined in the same way. The set of all histories is denoted by $\Hist(\G)$. Given a play $\rho = \rho_0 \rho_1 \ldots$, the history $\rho_k \ldots \rho_{k+i}$ is also denoted by $\rhofactor{k}{k+i}$. We denote by $w_m$ the projection of the weight function $w$ on the $m^{th}$ dimension, and by $W$ the maximum weight in absolute value on all dimensions. 

\paragraph{\bf Strategies} A \emph{strategy} $\sigma$ for player~$p \in \{1,2\}$ is a function $\sigma : V^*V_p \rightarrow V$ assigning to each history $hv \in V^*V_p$ a vertex $v' = \sigma(hv)$ such that $(v,v') \in E$. It is \emph{memoryless} if $\sigma(hv) = \sigma(h'v)$ for all histories $hv,h'v$ ending with the same vertex $v$, that is, $\sigma$ is a function $\sigma : V_p \rightarrow V$. It is \emph{finite-memory} if it can be encoded by a deterministic \emph{Moore machine} $(M, m_0, \alpha_u, \alpha_n)$ where $M$ is a finite set of states (the memory of the strategy), $m_0 \in M$ is the initial memory state, $\alpha_u : M \times V \rightarrow M$ is an update function, and $\alpha_n : M \times V_p \rightarrow V$ is the next-action function. The Moore machine defines a strategy $\sigma$ such that $\sigma(hv) = \alpha_n(\widehat{\alpha}_u(m_0,h),v)$ for all histories $hv \in V^*V_p$, where $\widehat{\alpha}_u$ extends $\alpha_u$ to sequences of vertices as expected. Note that such a strategy is memoryless if $|M| = 1$.

Given a strategy $\sigma$ of player $p \in \{1,2\}$, we say that a play $\rho$ of $\G$ is \emph{consistent} with $\sigma$ if $\rho_{k+1} = \sigma(\rho_0 \ldots \rho_k)$ for all $k \in \N$ such that $\rho_k \in V_p$. A history consistent with a strategy is defined similarly. Given an initial vertex $v_0$, and a strategy $\sigma_p$ of each player~$p$, we have a unique play that is consistent with both strategies. This play is called the \emph{outcome} of $(\sigma_1,\sigma_2)$ from $v_0$ and is denoted by $\Out(v_0,\sigma_1,\sigma_2)$. 

\paragraph{\bf Closed sets and attractors} 
Let $\G = \Gdev$ be a game structure, and let $p$ be a player in $\{1,2\}$. A set $U \subseteq V$ is \emph{$p$-closed} if for all $v \in U \cap V_p$ and all $(v,v') \in E$, we have $v' \in U$ (player $p$ cannot leave $U$), and for all $v \in U \cap V_{\overline{p}}$, there exists $(v,v') \in E$ such that $v' \in U$ (player $\overline{p}$ can ensure to stay in $U$). We thus get a subgame structure induced by $U$ that is denoted by $\G[U]$. 

Given $U \subseteq V$, the \emph{$p$-attractor} $\Attr{p}{U}{\G}$ is the set of vertices from which player $p$ has a strategy to reach $U$ against any strategy of player $\overline{p}$. 
It is contructed by induction as follows: 
$\Attr{p}{U}{\G} =  \cup_{k \geq 0} X_k$ such that
\begin{eqnarray*} \label{eq:attr}
X_0 &=& U, \\
X_{k+1} &=& X_k \cup \{v \in V_p \mid \exists (v,v') \in E, v' \in X_k \} \cup \{v \in V_{\overline{p}} \mid \forall (v,v') \in E, v' \in X_k \}.
\end{eqnarray*}
Notice that $U \subseteq \Attr{p}{U}{\G}$. The next properties are classical.

\begin{theorem} \label{thm:attr} 
Let $\G = \Gdev$ be a game structure, let $p$ be a player in $\{1,2\}$, and let $U \subseteq V$ be a set of vertices. Then
\begin{itemize}
\item The attractor $\Attr{p}{U}{\G}$ can be computed in $O(|V| + |E|)$ time {\rm \cite{Beeri80,Immerman81}}. 
\item The set $V \ssetminus \Attr{p}{U}{\G}$ is $p$-closed {\rm \cite{Zielonka98}}.
\item If $U$ is $\overline{p}$-closed, then $\Attr{p}{U}{\G}$ is $\overline{p}$-closed {\rm \cite{Zielonka98}}.
\end{itemize}
\end {theorem} 

\paragraph{\bf Objectives and winning sets} Let $\G = \Gdev$ be a game structure. An \emph{objective for player $p$} is a  set of plays $\Obj \subseteq \Plays(\G)$. An objective can be \emph{qualitative} (it only depends on the graph $(V,E)$), or \emph{quantitative} (it also depends on the weight function $w$). A play $\rho$ is \emph{winning} for player $p$ if $\rho \in \Obj$, and losing otherwise (i.e. winning for player $\overline{p}$). We thus consider zero-sum games such that the objective of player $\overline{p}$ is $\overline{\Obj} = \Plays(\G) \ssetminus \Obj$, that is, opposite to the objective $\Obj$ of player $p$. In the following, we always take the point of view of player~$1$ by supposing that $\Obj$ is his objective, and we denote by $(\G,\Obj)$ the corresponding \emph{game}. Given an initial state $v_0$, a strategy $\sigma_p$ for player $p$ is \emph{winning} from $v_0$ if $\Out(v_0,\sigma_p,\sigma_{\overline{p}}) \in \Obj$ for all strategies $\sigma_{\overline{p}}$ of player $\overline{p}$. Vertex $v_0$ is also called \emph{winning} for player $p$ and the \emph{winning set} $\Win{p}{\Obj}$ is the set of all his winning vertices. Similarly the winning vertices of player $\overline{p}$ are those from which player $\overline{p}$ can ensure to satisfy his objective $\overline{\Obj}$ against all strategies of player $p$, and $\Win{\overline{p}}{\overline{\Obj}}$ is his winning set. When $\Win{p}{\Obj} \cup \Win{\overline{p}}{\overline{\Obj}} = V$, we say that the game is \emph{determined}. It is known that every turn-based game with Borel objectives is determined~\cite{Martin75}. 
When we need to refer to a specific game structure $\G$, we use notation $\WinG{p}{\Obj}{\G}$.

\paragraph{\bf Qualitative objectives} Let $\G = (V_1,V_2,E)$ be an \emph{unweighted} game structure. Given a set $U \subseteq V$, classical qualitative objectives are the following ones. 
\begin{itemize}
\item A \emph{reachability} objective $\Reach(U)$ asks to visit a vertex of $U$ at least once. 
\item A \emph{safety} objective $\Safe(U)$ asks to visit no vertex of $V \setminus U$, that is, to avoid $V \setminus U$. A safety objective is thus the opposite of a reachability objective, that is $\Safe(U) = \overline{\Reach(V \setminus U)}$.
\item A \emph{B\"uchi} objective $\Buchi(U)$ asks to visit a vertex of $U$ infinitely often. 
\item A \emph{co-B\"uchi} objective $\CoBuchi(U)$ asks to visit no vertex of $V \setminus U$ infinitely often.  A co-B\"uchi objective is thus the opposite of a B\"uchi objective, that is $\CoBuchi(U) = \overline{\Buchi(V \setminus U)}$.
\end{itemize}

These objectives can be mixed by taking their intersection. Let $U_1, \ldots, U_i$ be a family of subsets of $V$: 
\begin{itemize}
\item A \emph{generalized reachability} objective $\GenReach(U_1, \ldots, U_i) = \cap_{k=1}^{i} \Reach(U_k)$ asks to visit a vertex of $U_k$ at least once, for each $k \in \{1, \ldots, i\}$.
\item A \emph{generalized B\"uchi} objective $\GenBuchi(U_1, \ldots, U_i) = \cap_{k=1}^{i} \Buchi(U_k)$ asks to visit a vertex of $U_k$ infinitely often, for each $k \in \{1, \ldots, i\}$.
\item A \emph{generalized B\"uchi $\cap$ co-B\"uchi} objective $\GenBuchi(U_1, \ldots, U_{i-1}) \cap \CoBuchi(U_i)$ asks to visit a vertex of $U_k$ infinitely often, for each $k \in \{1, \ldots, i-1\}$, and no vertex of $V \setminus U_i$ infinitely often.
\end{itemize}
Notice that the intersection of safety objectives $\cap_{k=1}^{i} \Safe(U_k)$ is again a safety objective that is equal to $\Safe(\cap_{k=1}^{i}U_k)$. Similarly the intersection of co-B\"uchi objectives $\cap_{k=1}^{i} \CoBuchi(U_k)$ is the co-B\"uchi objective $\CoBuchi(\cap_{k=1}^{i}U_k)$. Notice also that in an intersection of objectives, a safety objective $\Safe(U)$ can be replaced by the co-B\"uchi objective $\CoBuchi(U)$ by making each vertex of $V \setminus U$ absorbing, that is, with a unique outgoing edge being a self loop.

A game with an objective $\Obj$ is just called \emph{$\Obj$ game}. Hence a game with a reachability objective is called \emph{reachability game}, aso. In the sequel, we sometimes say by abuse of notation that $\Obj \in \Reg$ without mentioning a subset $U$ of vertices.

As all the previous qualitative objectives $\Obj$ are $\omega$-regular and thus Borel objectives, the corresponding games $(\G,\Obj)$ are determined. 

\paragraph{\bf Quantitative objectives} For a \emph{1-weighted} game structure $\G = \Gdev$ (with dimension $\n = 1$), let us now introduce quantitative objectives defined by some classical \emph{measure functions} $f : \Plays(\G) \rightarrow \Q$. Such a function $f$ associates a rational number to each play $\rho = \rho_1\rho_2 \ldots$ according to the weights $w(\rho_k,\rho_{k+1})$, $k \geq 0$, and can be one among the next functions. Let $\rho \in \Plays(\G)$:
\begin{itemize}
\item $\Inf(\rho) = \inf_{k \geq 0} ( w(\rho_k,\rho_{k+1}) )$: the $\Inf$ measure defines the minimum weight seen along the play.
\item $\Sup(\rho) = \sup_{k \geq 0} ( w(\rho_k,\rho_{k+1}) )$: the $\Sup$ measure defines the maximum weight seen along the play.
\item $\LimInf(\rho) = \liminf\limits_{k \to \infty} ( w(\rho_k,\rho_{k+1}) )$: the $\LimInf$ measure defines the mininum weight seen infinitely often along the play. 
\item $\LimSup(\rho) = \limsup\limits_{k \to \infty} ( w(\rho_k,\rho_{k+1}) )$:  the $\LimSup$ measure defines the maximum weight seen infinitely often along the play.  
\end{itemize}
Given such a measure function $f \in \{\Inf,\Sup,\LimInf,\LimSup\}$, a bound $\nu \in \Q$, and a relation $\sim \,\, \in$ $\{>,$ $\geq$, $<, \leq\}$, we define the objective $\Obj = f(\sim \nu)$ such that 
\begin{eqnarray} \label{eq:objclassical}
f(\sim \nu) &=&  \{ \rho \in \Plays(\G) \mid f(\rho) \sim \nu \}.
\end{eqnarray}

We are also interested in the next two measure functions defined on histories instead of plays. Let $\rho = \rho_0 \ldots \rho_i \in \Hist(\G)$:
\begin{itemize}
\item $\TP(\rho) = \Sigma_{k=0}^{i-1} w(\rho_k,\rho_{k+1})$: the \emph{total-payoff} measure $\TP$ defines the sum of the weights seen along the history.
\item $\MP(\rho) = \frac{1}{i}\TP(\rho)$: the \emph{mean-payoff} measure $\MP$ defines the mean of the weights seen along the history.
\end{itemize}

The second measure can be extended to plays $\rho$ as either $\underline{\MP}(\rho) = \liminf_{k\geq 0} \MP(\rhofactor{0}{k})$ or $\overline{\MP}(\rho) = \limsup_{k\geq 0} \MP(\rhofactor{0}{k})$. The $\MP$ measure on histories allows to define the \emph{window mean-payoff objective}, a new $\omega$-regular objective introduced in~\cite{Chatterjee0RR15}: given a bound $\nu \in \Q$, a relation $\sim \,\, \in \{>, \geq, <, \leq\}$, and a window size $\lambda \in \N \ssetminus\{0\}$, the objective $\DFWMP(\lambda, \sim \nu)$\footnote{This objective is called ``direct fixed window mean-payoff'' in \cite{Chatterjee0RR15} among several other variants.} is equal to 
\begin{eqnarray} \label{eq:objwindow}
\DFWMP(\lambda, \sim \nu) = \{ \rho \in \Plays(\G) \mid \forall k \geq 0, \exists l \in \{1, \ldots, \lambda\}, \MP(\rhofactor{k}{k+l}) \sim \nu \}. 
\end{eqnarray}
The window mean-payoff objective asks that the average weight becomes $\sim \nu$ inside a local bounded window for all positions of this window sliding along the play, instead of the classical mean-payoff objective asking that the long run-average $\underline{\MP}(\rho)$ (resp. $\overline{\MP}(\rho)$) is $\sim \nu$. This objective is a strengthening of the mean-payoff objective.
 
Given a multi-weighted game structure $\G = \Gdev$ (with dimension $\n \geq 1$), we can mix objectives of (\ref{eq:objclassical}) and (\ref{eq:objwindow}) by fixing one such objective $\Obj_m$ for each dimension $m \in \{1, \ldots, \n\}$, and then taking the intersection $\cap_{m=1}^{\n} \Obj_m$. More precisely, given a vector $(\sim_1 \nu_1, \ldots, \sim_n \nu_{\n})$, each objective $\Obj_m$ uses a measure function based on the weight function $w_m$; $\Obj_m$ is either of the form $f(\sim_m \nu_m)$ with $f \in \{\Inf,\Sup,\LimInf,\LimSup\}$, or of the form $\DFWMP(\lambda, \sim_m \nu_m)$ for some window size~$\lambda$ (this size can change with $m$).

As done for qualitative objectives, we use the shortcut \emph{$\Obj$ game} for a game with quantitative objective $\Obj$. For instance an \emph{$\Inf(\sim \nu)$ game} is a 1-weighted game with objective $\Inf(\sim \nu)$, a \emph{$\LimSup(\sim_1 \nu_1) \cap \DFWMP(\lambda,\sim_2 \nu_2) \cap \Inf(\sim_3 \nu_3)$ game} is a 3-weighted game with the intersection of a $\LimSup(\sim_1 \nu_1)$ objective on the first dimension, a $\DFWMP(\lambda,\sim_2 \nu_2)$ objective on the second dimension, and an $\Inf(\sim_3 \nu_3)$ objective on the third dimension.

In the sequel, we sometimes abusively say that $\Obj = \cap_{m=1}^{\n} \Obj_m$ with $\Obj_m \in \WISL$ without mentioning the used relations, bounds and window sizes. It is implicitly supposed that $\Obj_m$ deals with the $m^{th}$ component of the weight function.

In this paper, we want to study the next problem. We will go beyond this problem in Section~\ref{sec:beyond} by considering Boolean combinations (instead of conjunctions) of objectives.  
\begin{problem} \label{prob:thres}
Let $\G = \Gdev$ be a multi-weighted game structure with dimension $\n \geq 1$, and $\Obj =  \cap_{m=1}^{\n} \Obj_m$ be a quantitative objective such that each $\Obj_m \in \WISL$. Can we compute the winning sets $\Win{1}{\Obj}$ and $\Win{2}{\overline{\Obj}}$? If yes what is the complexity of computing these sets and how (memoryless, finite-memory, general) are the winning strategies of both players? Given such a game $(\G,\Obj)$ and an initial vertex $v_0$, the \emph{\prob} asks to decide whether player~$1$ has a winning strategy for $\Obj$ from $v_0$ and to build such a strategy when it exists.
\end{problem}

\begin{remark} \label{rem:threshold0}
$(i)$ In this problem, we can assume that the bounds used in the vector $(\sim_1 \nu_1,\ldots, \sim_\n\nu_{\n})$ are such that $(\nu_1,\ldots, \nu_{\n}) = (0, \ldots, 0)$. Indeed, suppose that $\nu_m = \frac{a}{b}$ with $a \in \Z$ and $b \in \N \ssetminus \{0\}$, then replace the $m^{th}$ component $w_m$ of the weight function $w$ by $b \cdot w_m - a$. $(ii)$ Moreover notice that if $\nu_m = 0$ and $\Obj_m = \DFWMP$, then $ \MP(\rhofactor{k}{k+l})$ can be replaced by $\TP(\rhofactor{k}{k+l})$ in (\ref{eq:objwindow}). 
$(iii)$ Finally, the vector $(\sim_1 0,\ldots, \sim_\n 0)$ can be supposed to be equal to $(\geq 0,\ldots, \geq 0)$. Indeed strict inequality $> 0$ (resp. $<0$) can be replaced by inequality $\geq 1$ (resp. $\leq -1$), and inequality $\leq 0$ can be replaced by $\geq 0$ by replacing the weight function by its negation and the measure $\Inf$ (resp. $\Sup$, $\LimInf$, $\LimSup$, $\TP$) by $\Sup$ (resp. $\Inf$, $\LimSup$, $\LimInf$, $\TP$).

\noindent From now on, we only work with vectors $(\geq 0,\ldots, \geq 0)$ and we no longer mention symbol $\geq$. Hence, as an example, $\Inf(\geq 0)$ and $\DFWMP(2,\geq 0)$ are replaced by $\Inf(0)$ and $\DFWMP(2,0)$; and when the context is clear, we only mention $\Inf$ and $\DFWMP$.
\end{remark}


All the objectives considered in Problem~\ref{prob:thres} are $\omega$-regular, the corresponding games $(G,\Obj)$ are thus determined. In the proofs of this paper, we will often use the property that $\Win{p}{\Obj} \cup \Win{\overline{p}}{\overline{\Obj}} = V$.

Let us illustrate a mixing of quantitative objectives on the following example.

\begin{example} \label{ex:game}
Consider the 3-weighted game structure depicted on Figure~\ref{Ex:Firstexample}. In all examples in this paper, we assume that circle (resp. square) vertices are controlled by player~$1$ (resp. player~$2$). Let $\Obj = \DFWMP(3,0) \cap \Sup(0) \cap \LimSup(0)$ be the objective of player~$1$. Let us recall that by definition of $\Obj$, we look at the $\DFWMP$ (resp. $\Sup$, $\LimSup$) objective on the first (resp. second, third) dimension. Let us show that $v_0$ is a winning vertex for player~$1$. Let $\sigma_1$ be the following strategy of player~$1$ from $v_0$: go to $v_1$, take the self loop once, go back to $v_0$ and then always go to $v_2$. More precisely, we have $\sigma_1(v_0) = v_1$, $\sigma_1(v_0v_1) = v_1$, $\sigma_1(v_0v_1v_1) = v_0$, and $\sigma_1(v_0v_1v_1hv_0) = v_2$ for all $hv_0 \in V^*V_1$.  Let $\sigma_2$ be any strategy of player~$2$. We are going to prove that $\rho = \Out(v_0,\sigma_1,\sigma_2) \in \Obj$. Notice that  $\rho \in  v_0v_1v_1v_0\{v_2,v_0\}^{\omega}$. As player~$1$ forces $\rho$ to begin with $v_0v_1v_1$, he ensures that $\Sup(\rho) \geq 0$ on the second component. Moreover as $\rho$ visits infinitely often edge $(v_2,v_2)$ or $(v_2,v_0)$, player~$1$ also ensures to have  $\LimSup(\rho) \geq 0$ on the third component. 
Finally, we have to check that $\rho \in \DFWMP(3,0)$ with respect to the first component, that is (by Remark~\ref{rem:threshold0}), for all $k$, there exists $l \in \{1,2,3\}$ such that $\TP(\rhofactor{k}{k+l}) \geq 0$. For $k = 0$ (resp. $k=1$, $k=2$, $k = 3$) and $l=3$ (resp. $l = 2$, $l=1$, $l=1$), we have $\TP(\rhofactor{k}{k+l}) \geq 0$. Now, from position $k = 4$, the sum of weights is non-negative in at most $2$ steps. Indeed, either player~$2$ takes the self loop $(v_2,v_2)$ or he goes to $v_0$ where player~$1$ goes back to $v_2$. Therefore, for each $k \geq 4$, either $\rho_k = v_0$ and $\TP(\rhofactor{k}{k+l}) \geq 0$ with $l=1$, or $\rho_k = v_2$ and $\TP(\rhofactor{k}{k+l}) \geq 0$ with $l=1$ if $\rho_{k+1} = v_2$, and with $l=2$ otherwise. It follows that $v_0 \in \Win{1}{\Obj}$. The strategy $\sigma_1$ needs memory: indeed, player~$1$ needs to remember if he has already visited the edge $(v_1,v_1)$ as this is the only edge visiting a weight greater than or equal to $0$ for the $\Sup$ objective. Finally, one can show that $\Win{1}{\Obj} = \{v_0,v_1\}$ and $\Win{2}{\overline{\Obj}} = \{v_2\}$.
\end{example}
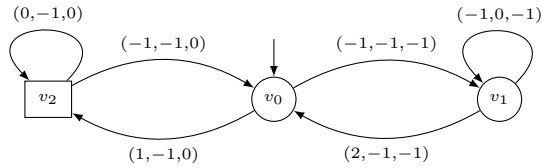
\begin{figure}[h]
\centering
  \begin{tikzpicture}[scale=4]
    \everymath{\scriptstyle}
    \draw (0,0) node [circle, draw] (A) {$v_0$};
    \draw (0.75,0) node [circle, draw] (B) {$v_1$};
    \draw (-0.75,0) node [rectangle, inner sep=5pt, draw] (C) {$v_2$};
    
    \draw[->,>=latex] (A) to[bend left] node[above,midway] {$(-1,-1,-1)$} (B);
    \draw[->,>=latex] (B) to[bend left] node[below,midway] {$(2,-1,-1)$} (A);
    
    \draw[->,>=latex] (A) to[bend left] node[below,midway] {$(1,-1,0)$} (C);
    \draw[->,>=latex] (C) to[bend left] node[above,midway] {$(-1,-1,0)$} (A);
    
    \draw[->,>=latex] (C) .. controls +(45:0.4cm) and +(135:0.4cm) .. (C) node[above,midway] {$(0,-1,0)$};
    \draw[->,>=latex] (B) .. controls +(45:0.4cm) and +(135:0.4cm) .. (B) node[above,midway] {$(-1,0,-1)$};
	\path (0,0.2) edge [->,>=latex] (A);    
    
    \end{tikzpicture}
\caption{Example of a multi-weighted two-player game}
\label{Ex:Firstexample}
\end{figure}

\begin{remark} \label{rem:reductions}
In Problem~\ref{prob:thres}, the vector $(\sim_1 \nu_1,\ldots, \sim_\n\nu_{\n})$ can be assumed equal to $(\geq 0,\ldots, \geq 0)$ by Remark~\ref{rem:threshold0}. It follows that an $\Inf$ (resp. $\Sup$, $\LimInf$, $\LimSup$) objective is nothing else than a safety (resp. reachability, co-B\"uchi, B\"uchi) objective, and conversely. These two game reductions, together with a third reduction, will be described and proved in Section~\ref{sec:general}. They all will be used throughout this paper.
\end{remark}

\paragraph{\bf Some well-known properties} Let us recall some well-known properties about the objectives mentioned above. The complexity of the algorithms is expressed in terms of the size $|V|$ and $|E|$ of the game structure $\G$, the maximum weight $W$ (in absolute value) and the dimension $n$ of the weight function when $\G$ is weighted, the number $i$ of objectives in an intersection of objectives\footnote{Notice that $i = n$ for the objectives considered in Problem~\ref{prob:thres}.}, and the window size $\lambda$. We begin with the qualitative objectives in an unweighted game structure (see also Table~\ref{table:reg}).

\begin{theorem} \label{thm:reg}
\begin{itemize}
\item For reachability or safety games,
deciding the winner is $\mathsf P$-complete (with an algorithm in $O(|V|+|E|)$ time) and both players have memoryless winning strategies {\rm\cite{Beeri80,2001automata,Immerman81}}.
\item For B\"uchi or co-B\"uchi games, 
deciding the winner is $\mathsf P$-complete (with an algorithm in $O(|V|^2)$ time) and both players have memoryless winning strategies  {\rm\cite{ChatterjeeH14,EmersonJ91,Immerman81}}. 
\item For generalized reachability games, deciding the winner is $\mathsf{PSPACE}$-complete (with an algorithm in $O(2^i \cdot (|V|+|E|))$ time) and exponential memory strategies are necessary and sufficient for both players  {\rm\cite{FijalkowH13}}.
\item For generalized B\"uchi games, deciding the winner is $\mathsf P$-complete (with an algorithm in $O(i \cdot |V| \cdot |E|)$ time) and polynomial memory (resp. memoryless) strategies are necessary and sufficient for player~$1$ (resp. player~$2$)  {\rm\cite{DziembowskiJW97}}.
\item For B\"uchi $\cap$ co-B\"uchi games, deciding the winner is $\mathsf P$-complete (with an algorithm in $O(|V| \cdot |E|)$ time)  and memoryless strategies are necessary and sufficient for both players {\rm \cite{AlfaroF07}}.
\item For generalized B\"uchi $\cap$ co-B\"uchi games, deciding the winner is $\mathsf P$-complete (with an algorithm in tile $O(i^2 \cdot |V| \cdot |E|)$) and polynomial memory (resp. memoryless) strategies are necessary and sufficient for player~$1$ (resp. player~$2$).\footnote{This result is obtained thanks to a classical reduction of generalized B\"uchi $\cap$ co-B\"uchi games to B\"uchi $\cap$ co-B\"uchi games (see also Lemma~\ref{lem:GenBtoB}).}
\end{itemize}
\end{theorem}
In case of reachability games with objective $\Reach(U)$ for player~$1$, notice that his winning set is equal to $\Attr{1}{U}{\G}$ (see Theorem~\ref{thm:attr}). 

\begin{table}
\begin{center}
\begin{tabular}{|c|c|c|c|c|c|c|c|}
\hline
Objective           	 & Complexity class &  Algorithmic complexity                   & Player~$1$ memory & Player~$2$ memory \\
\hline
$\Reach/\Safe$         & $\mathsf P$-complete       & $O(|V| + |E|)$        & memoryless         & memoryless \\
$\Buchi/\CoBuchi$    & $\mathsf P$-complete        &  $O(|V|^2)$     & memoryless         & memoryless \\
$\GenReach$           & $\mathsf{PSPACE}$-complete& $O(2^i \cdot (|V|+|E|))$ & exponential memory      & exponential memory \\
$\GenBuchi$            & $\mathsf P$-complete            & $O(i \cdot |V| \cdot |E|)$  & polynomial memory      & memoryless \\
$\Buchi \cap \CoBuchi$ & $\mathsf P$-complete    & $O(|V| \cdot |E|)$ &  memoryless       & memoryless \\
$\GenBuchi \cap \CoBuchi$ & $\mathsf P$-complete & $O(i^2 \cdot |V| \cdot |E|)$ &  polynomial memory       & memoryless \\
\hline
\end{tabular}
\end{center}
\caption{Overview of some known results for qualitative objectives ($i$ is the number of objectives in the intersection of reachability/B\"uchi objectives)}
\label{table:reg}
\end{table}

The next corollary will be useful in the proofs.

\begin{corollary}{\rm \cite{2001automata}} \label{cor:reg}
\begin{itemize}
\item For reachability or generalized reachability games, $\Win{2}{\overline{\Obj}}$ is 1-closed. If the set(s) to reach is (are) 2-closed, then $\Win{1}{\Obj}$ is 2-closed.
\item For safety games, $\Win{1}{\Obj}$ is 2-closed. If the set to avoid is 1-closed, then $\Win{2}{\overline{\Obj}}$ is 1-closed.
\item For B\"uchi, co-B\"uchi, generalized B\"uchi games, and generalized B\"uchi $\cap$ co-B\"uchi games, $\Win{1}{\Obj}$ is 2-closed and $\Win{2}{\overline{\Obj}}$ is 1-closed.
\end{itemize}
\end{corollary}

We conclude this section with the known results about the window mean-payoff objective.

\begin{theorem} {\rm\cite{Chatterjee0RR15}}\footnote{When $n=1$, the time complexity and the memory requirements have been here correctly stated.} \label{thm:WMP}
Let $(G,\Obj)$ be an n-weighted game such that $\Obj = \cap_{m=1}^n \Obj_m$ with $\Obj_m = \DFWMP$ for all $m$. Then the \prob\ is ${\sf EXPTIME}$-complete (with an algorithm in  $O(\lambda^{4n} \cdot |V|^2  \cdot W^{2n})$ time), exponential memory strategies are sufficient and necessary for both players. This problem is already ${\sf EXPTIME}$-hard when $n = 2$. 

When $n=1$, the \prob\ is decidable in $O(\lambda \cdot |V| \cdot (|V| + |E|) \cdot \lceil \log_2(\lambda \cdot W) \rceil)$ time, both players require finite-memory strategies, and memory in $O(\lambda^2 \cdot W)$ (resp. in $O(\lambda^2 \cdot W \cdot |V|)$) is sufficient for player~$1$ (resp. player~$2$). Moreover, if $\lambda$ is polynomial in the size of the game, then the \prob\ is $\sf P$-complete.
\end{theorem}
\section{Intersection of objectives in $\WISL$} \label{sec:general}

In this section, we solve Problem~\ref{prob:thres}: the \prob\ is $\mathsf{EXPTIME}$-complete, and exponential memory strategies are both sufficient and necessary for both players (Theorem~\ref{thm:general}). Before proving this theorem, we first need to introduce some terminology about the window mean-payoff objective. We also need to establish several reductions between different classes of games. These preliminaries will be useful in the proof of various results of this paper, included Theorem~\ref{thm:general}.

\begin{theorem} \label{thm:general}
Let $(\G, \Obj)$ be an $\n$-weighted game such that $\Obj = \cap_{m = 1}^n \Obj_m$ with $\Obj_m \in \WISL$ for all $m$. Then, the \prob\ is $\mathsf{EXPTIME}$-complete (with an algorithm in $O(|V| \cdot |E|  \cdot (\lambda^2 \cdot W)^{2n})$ time), and exponential memory strategies are both necessary and sufficient for both players.
\end{theorem}

\subsection{Properties of windows}

Let us come back to the $\DFWMP(\lambda, 0)$ objective (with $\TP$ in (\ref{eq:objwindow})) and let us introduce some useful terminology. Let $\rho = \rho_0 \rho_1 \ldots$ be a play. A \emph{\lwindow{k}} is a window of size $\lambda$ placed along $\rho$ from $k$ to $k + \lambda$. If there exists $l \in \{1, \ldots, \lambda\}$ such that $\TP(\rhofactor{k}{k+l}) \geq 0$, such a \lwindow{k} is called \emph{\good{\lambda}} or \emph{\closed{\lambda}{k+l}} (to specify index $l$), otherwise it is called \emph{\bad{\lambda}}. Moreover if $l$ is the smallest index such that $\TP(\rhofactor{k}{k+l}) \geq 0$, we say it is \emph{\fclosed{\lambda}{k+l}}. 

An interesting property is the following one: a \lwindow{k} is \emph{\iclosed{\lambda}{k+l}} if it is \closed{\lambda}{k+l} and for all $k' \in \{k+1, \ldots, k+l-1\}$, the \lwindow{k'} is also \closed{\lambda}{k+l} (see Figure~\ref{fig:iclosed}).
\begin{figure}
\centering
\begin{tikzpicture}[scale=1]

    \draw (-1.6,0) node [circle] (A) {$\ldots$};
    \draw (0,0) node [circle, draw,inner sep=6pt] (B) {$\rho_k$};
    \draw (2.4,0) node [circle, draw,,inner sep=3.8pt] (C) {$\rho_{k+1}$};
    \draw (4,0) node [circle] (D) {$\ldots$};
    \draw (5.6,0) node [circle, draw,inner sep=2.3pt] (E) {$\rho$\tiny$_{k+l-1}$};
    \draw (8,0) node [circle, draw,inner sep=3.8pt] (F) {$\rho_{k+l}$};
    \draw (10.4,0) node [circle] (G) {$\ldots$};
    
    \draw (A) -- (B);
	\draw (B) -- (C);
	\draw (C) -- (D);
	\draw (D) -- (E);
	\draw (E) -- (F);
	\draw (F) -- (G);
    
    \draw (B) to[bend left] node[below,midway] {$\geq 0$} (F);
    \draw (C) to[bend left] node[below,midway] {$\geq 0$} (F);
    \draw (E) to[bend left] node[below,midway] {$\geq 0$} (F);

\end{tikzpicture}
\caption{A \lwindow{k} that is \iclosed{\lambda}{k+l}}
\label{fig:iclosed}
\end{figure}
One easily checks the next property:

\begin{lemma} \label{lem:factorize}
If a \lwindow{k} is \fclosed{\lambda}{k+l}, then it is \iclosed{\lambda}{k+l}. Conversely, if a \lwindow{k} is \iclosed{\lambda}{k+l}, then either it is \fclosed{\lambda}{k+l}, or there exists $l' \in \{1, \ldots, l-1\}$ such that the \lwindow{k} is \fclosed{\lambda}{k+l'} and the \lwindow{k+l'} is \iclosed{\lambda}{k+l}.
\end{lemma}

The next lemma will be useful in several proofs of this paper:

\begin{lemma} \label{lem:iclosed}
A play $\rho$ is winning for $\DFWMP(\lambda, 0)$ iff there exists a sequence $(k_i)_{i \geq 0}$ with $k_0 = 0$ such that for each $i$, we have $k_{i+1} - k_i \in \{1,\ldots, \lambda\}$, and the \lwindow{k_i} is \iclosed{\lambda}{k_{i+1}}.
\end{lemma}

\noindent
When such a sequence $(k_i)_{i \geq 0}$ exists for a play $\rho$, we say that it is a \emph{\GoodDec{\lambda}} of $\rho$. We extend this notion to histories $\rho = \rho_0\rho_1 \ldots \rho_k$ as follows. A finite sequence $(k_i)_{i = 0}^{j}$ is a \GoodDec{\lambda} of $\rho$ if $k_0 = 0$, $k_j = k$, and for each $i \in \{0, \ldots j-1\}$, we have $k_{i+1} - k_i \in \{1,\ldots, \lambda\}$, and the \lwindow{k_i} is \iclosed{\lambda}{k_{i+1}}. The \emph{size} of the decomposition is $j$. In particular the history $\rho_0$ (take $k=0$) has always a \GoodDec{\lambda} of size $j = 0$.

Notice that if a sequence $(k_i)_{i \geq 0}$ is a \GoodDec{\lambda} of a play $\rho$, then by Lemma~\ref{lem:factorize}, there exists a unique sequence $(k'_i)_{i \geq 0}$ of which $(k_i)_{i \geq 0}$ is a subsequence that is also a \GoodDec{\lambda} of $\rho$ and such that for each $i$ the \lwindow{k'_i} is \fclosed{\lambda}{k'_{i+1}} (instead of being only \iclosed{\lambda}{k'_{i+1}}). See also Figure~\ref{fig:factorize}. We have a similar property for finite good decompositions. We call such a sequence a \emph{maximal} \GoodDec{\lambda}.

\begin{figure}
\centering
\begin{tikzpicture}[scale=3]

	\draw (0,0) -- (4.3,0); 
	\draw (4.4,0) node (A) {$\ldots$};
	
	\draw (0,0) to[bend left=50] (1,0);	
	\draw (1,0) to[bend left] (2.6,0);
	\draw (2.6,0) to[bend left] (4,0);
	
	\draw[dashed] (0,0) to[bend right=60] (0.5,0);
	\draw[dashed] (0.5,0) to[bend right=60] (1,0);
	\draw[dashed] (1,0) to[bend right=80] (1.4,0);
	\draw[dashed] (1.4,0) to[bend right=45] (2.1,0);
	\draw[dashed] (2.1,0) to[bend right=70] (2.6,0);
	\draw[dashed] (2.6,0) to[bend right] (3.6,0);
	\draw[dashed] (3.6,0) to[bend right=85] (4,0);

	\draw (0,0.15) node [] (B) {$k_0$};
	\draw (1,0.15) node [] (C) {$k_1$};
	\draw (2.6,0.15) node [] (D) {$k_2$};
	\draw (4,0.15) node [] (M) {$k_3$};
	\draw (0,-0.15) node [] (E) {$k'_0$};
	\draw (0.5,-0.15) node [] (F) {$k'_1$};
	\draw (1,-0.15) node [] (G) {$k'_2$};
	\draw (1.4,-0.15) node [] (H) {$k'_3$};
	\draw (2.1,-0.15) node [] (I) {$k'_4$};
	\draw (2.6,-0.15) node [] (J) {$k'_5$};
	\draw (3.6,-0.15) node [] (K) {$k'_6$};
	\draw (4,-0.15) node [] (L) {$k'_7$};

\end{tikzpicture}
\caption{The sequence $(k'_i)_{i \geq 0}$ is a good decompostion with all its $\lambda$-windows being first-closed.}
\label{fig:factorize}
\end{figure}
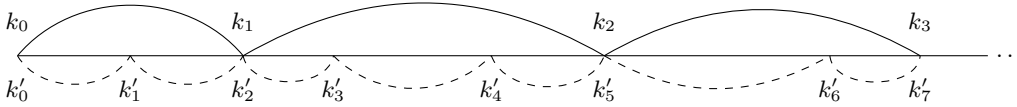

\begin{proof}[of Lemma~\ref{lem:iclosed}]
Suppose that $\rho$ is winning for $\DFWMP(\lambda, 0)$. Then there exists $k_1$ such that \lwindow{k_0 = 0} is \fclosed{\lambda}{k_{1}} (and thus \iclosed{\lambda}{k_{1}}), there exists $k_2$ such that \lwindow{k_1} is \fclosed{\lambda}{k_{2}} (and thus \fclosed{\lambda}{k_{2}}), aso. This leads to a \GoodDec{\lambda} $(k_i)_{i \geq 0}$ of $\rho$, that is its maximal \GoodDec{\lambda}. 

Conversely, if there exists a \GoodDec{\lambda} $(k_i)_{i \geq 0}$ of $\rho$, then for each $k_i$, the \lwindow{k_i} is \good{\lambda} as well as all $\lambda$-windows at position $k' \in \{k_i +1, \ldots, k_{i+1}-1\}$. It follows that $\rho$ is winning for $\DFWMP(\lambda, 0)$.
\qed\end{proof}

Let us come back to Example~\ref{ex:game} to illustrate those properties of windows.

\addtocounter{example}{-1}
\begin{example}[continued]
We consider again the 3-weighted game structure depicted on Figure~\ref{Ex:Firstexample} and the objective $\Obj = \DFWMP(3,0) \cap \Sup(0) \cap \LimSup(0)$ for player~$1$. We keep the same strategy $\sigma_1$ for player~$1$ from $v_0$. Against any strategy of player~$2$, we have an outcome $\rho \in  v_0v_1v_1v_0\{v_2,v_0\}^{\omega}$. Let us explain again that $\rho \in \DFWMP(3,0)$ with respect to the first component. To this end, we use Lemma~\ref{lem:iclosed} and show that $\rho$ has a \GoodDec{3} $(k_i)_{i \geq 0}$ (the sequence that we will propose is the maximal \GoodDec{3}  of $\rho$). We let $k_0 = 0$, $k_1 = 3$ as the $3$-\window{0} is \fclosed{3}{3}, $k_2 = 4$ as the $3$-\window{3} is \fclosed{3}{4}.
Now, from position $4$ the sum of weights is non-negative in at most $2$ steps. Indeed, either player~$2$ takes the self loop $(v_2,v_2)$ or he goes to $v_0$ where player~$1$ goes back to $v_2$. Then, for $j \geq 3$, we take $k_j = k_{j-1} + 1$ or $k_j = k_{j-1} + 2$, depending on  the choice of player~$2$. This shows that $\rho \in \DFWMP(3,0)$ by Lemma~\ref{lem:iclosed}. 
\end{example}

From now on, a $\lambda$-window will be simply called a window.

\subsection{Game reductions}

Let us turn to the reductions. We begin by a polynomial reduction of the class of games $(\G,\Obj = \cap_{m=1}^{\n} \Obj_m)$ with each $\Obj_m \in \WISL$ to the class of games $(\G',\Obj' = \cap_{m=1}^{\n} \Obj'_m)$ with each $\Obj'_m \in  \{ \DFWMP,\Safe, \Reach, \CoBuchi, \Buchi \}$. The objective $\Obj'$ is thus the intersection of qualitative and quantitative objectives. 

\begin{proposition} \label{prop:WISLtoWreg}
Each $\n$-weighted game $(\G,\Obj)$ with $\G = \Gdev$, and $\Obj = \cap_{m=1}^{\n} \Obj_m$ such that for all $m$, $$\Obj_m \in \WISLmath$$ can be polynomially reduced to a game $(\G',\Obj')$ with $|V|+|E|$ vertices and $2 \cdot |E|$ edges, and $\Obj' = \cap_{m=1}^{\n} \Obj'_m$ such that for all $m$, $$\Obj'_m \in \{ \DFWMP,\Safe, \Reach, \CoBuchi, \Buchi \}$$
Moreover, if $\Obj_m$ is a $\DFWMP$ (resp. $\Inf$, $\Sup$, $\LimInf$, $\LimSup$) objective, then $\Obj'_m$ is a $\DFWMP$ (resp. $\Safe$, $\Reach$, $\CoBuchi$, $\Buchi$) objective.

\noindent A memoryless (resp. (polynomial, exponential) memory) strategy in $G'$ transfers to a memoryless (resp. (polynomial, exponential) memory) strategy in $G$.
\end{proposition}


\begin{proof}
Let $(\G,\Obj)$ with $\Obj = \cap_{m=1}^{\n} \Obj_m$ such that $\Obj_m \in \WISL$ for all $m$. By Remark~\ref{rem:threshold0}, we suppose that the threshold $\nu$ is equal $(0, \ldots, 0)$ and that $\MP(\rhofactor{k}{k+l})$ is replaced by $\TP(\rhofactor{k}{k+l})$ in (\ref{eq:objwindow}).

From $\G$, we construct a new game structure $\G' = (V'_1,V'_2,E')$ in a way to have weights depending on vertices instead of edges, except for objectives $\Obj_m = \DFWMP$: each edge $e = (v,v') \in E$ is split into two consecutive edges, and the new intermediate vertex belongs to player~$1$\footnote{It could belong to player~$2$ since there is exactly one outgoing edge.}. The resulting game structure has thus $|V| + |E|$ vertices and $2 \cdot |E|$ edges, that is, a polynomial size. To explain how to proceed with the weights, we first suppose that $\G$ is 1-weighted and $\Obj = \Obj_1$. Then we will explain how to extend the construction to $\n$-weighted game structures and $\Obj = \cap_{m=1}^{\n} \Obj_m$.

Let $\n = 1$. We proceed as follows in the case $\Obj = \Inf(0)$. For each edge $e \in E$ split into two consecutive edges, the new intermediate vertex is decorated with $0$ if $w(e) \geq 0$ and with $-1$ otherwise, and the vertices of $V$ are decorated with $0$. Then, we define the set $U' \subseteq V'$ composed of all vertices decorated with $0$. Clearly, player~$1$ has a winning strategy  from vertex $v_0$ in the initial game iff player~$1$ has a winning strategy from $v_0$ in the constructed game with the safety objective $\Obj' = \Safe(U')$. If $\Obj = \Sup(0)$ (resp. $\LimInf(0)$, $\LimSup(0)$), the construction is the same except that the vertices of $V$ are decorated with $-1$ (resp. $0$, $-1$) and the objective $\Obj'$ is now $\Reach(U')$ (resp. $\CoBuchi(U')$, $\Buchi(U')$). 

When $n = 1$, it remains to consider the case $\Obj = \DFWMP(\lambda,0)$ with $\lambda \in \N \ssetminus \{0\}$. For each edge $e \in E$ split into two consecutive edges, the weight $w(e)$ is replaced by two weights respectively equal to $-1$ on the first new edge and $1 + w(e)$ on the second one. Let us show that player~$1$ has a winning strategy  from vertex $v_0$ in the initial game iff player~$1$ has a winning strategy from $v_0$ in the constructed game with the $\DFWMP(2 \lambda, 0)$ objective. There is a one-to-one correspondance between the plays of $\G$ and the plays of $\G'$. Let $\rho \in \Plays(\G)$ and its corresponding $\rho' \in \Plays(\G')$. If a \window{k} is \closed{\lambda}{k+l} for $\rho$, then the window of size $2 \lambda$ at position $2k$ (resp. $2k+1$) is \closed{2 \lambda}{2k+2l} for $\rho'$. It follows that  if $\rho$ is winning for the $\DFWMP(\lambda, 0)$ objective, then $\rho'$ is winning for the $\DFWMP(2 \lambda, 0)$ objective. The converse is also true since if a \window{2k} is \closed{2 \lambda}{2k+l} for $\rho'$, then we can suppose that $l = 2l'$, and thus we have a \window{k} that is \closed{\lambda}{k+l'} for $\rho$.


Let us now suppose that $\n \geq 2$. Let $\n^*$ be the number of objectives $\Obj_m \in \ISL$ and $\n - \n^*$ be the number of objectives $\Obj_m = \DFWMP$. We proceed as in the case $\n = 1$ except that the vertices of the new game structure are decorated by vectors of weights in $\{-1,0\}^{\n^*}$ and its edges are labeled by vectors of weights in $\Z^{\n - \n^*}$. The decoration and labeling are done component-wise as explained above. The game structure $\G'$ has still a size polynomial in the initial game structure $\G$.

Finally, it is clear that any strategy in $G'$ that is memoryless (resp. (polynomial, exponential) memory) transfers to a strategy in $G$ that is again memoryless (resp. (polynomial, exponential) memory).
\qed\end{proof}

We now turn to an exponential reduction of the class of games $(\G,\Obj = \cap_{m=1}^{\n} \Obj_m)$ with each $\Obj_m \in \WISL$ to the class of games $(\G',\Obj' = \cap_{m=1}^{\n} \Obj'_m)$ with each $\Obj'_m$ being an objective in  $\{\Buchi,\CoBuchi\}$.

\begin{proposition} \label{prop:WISLtoBC}
Each $\n$-weighted game $(\G,\Obj)$ with $\Obj = \cap_{m=1}^{\n} \Obj_m$ such that for all $m$, $$\Obj_m \in \WISLmath$$ can be exponentially reduced to a game $(\G',\Obj')$ with $O(|V| \cdot (\lambda^2 \cdot W)^n)$ vertices and $O(|E| \cdot (\lambda^2 \cdot W)^n)$ edges, and $\Obj' = \cap_{m=1}^{\n} \Obj'_m$ such that for all $m$, $$\Obj'_m \in \{\Buchi,\CoBuchi\}$$
Moreover, if $\Obj_m$ is a $\DFWMP$ (resp. $\Inf$, $\Sup$, $\LimInf$, $\LimSup$) objective, then $\Obj'_m$ is a $\CoBuchi$ (resp. $\CoBuchi$, $\Buchi$, $\CoBuchi$, $\Buchi$) objective.

\noindent A memoryless (resp.  finite-memory) strategy in $G'$ transfers to a finite-memory strategy in $G$.
\end{proposition}

\begin{proof}
Let $(\G,\Obj)$ with $\Obj = \cap_{m=1}^{\n} \Obj_m$ such that $\Obj_m \in \WISL$ for all $m$. We proceed like in the previous proof by first considering 1-weighted games.

Let $\n = 1$ and $\Obj = \Inf(0)$. From $\G$, we construct a new unweighted game structure $\G' = (V'_1,V'_2,E')$ in a way to keep in the current vertex of $\G'$ the information whether or not we have seen a strictly negative weight along the current history in $\Hist(\G)$. More precisely, we define $V' = V \times \{-1,0\}$ such that $-1$ (resp. $0$) expresses the presence (resp. absence) of strictly negative weights, and $E'$ is composed of the edges $((v,a),(v',a'))$ such that $e = (v,v') \in E$ and $a' = -1$ if $w(e) < 0$ or $a = -1$, and $a' = 0$ otherwise. Then, we define the set $U' \subseteq V'$ composed of all vertices in $V \times \{0\}$. It follows that player~$1$ has a winning strategy from $v_0$ in $(\G,\Obj)$ iff player~$1$ has a winning strategy from $(v_0,0)$ in $(\G',\Obj')$ with $\Obj' = \CoBuchi(U')$. 

The case $\Obj = \Sup(0)$ is solved similarly with the current vertex of $\G'$ now keeping the information whether or not we have seen a positive weight along the current history in $\Hist(\G)$. We define $V' = V \times \{-1,0\}$ such that $0$ (resp. $-1$) expresses the presence (resp. absence) of positive weights; $E'$ is composed of the edges $((v,a),(v',a'))$ such that $e = (v,v') \in E$ and $a' = 0$ if $w(e) \geq 0$ or $a = 0$, and $a' = -1$ otherwise; and $U' = V \times \{0\}$. Then player~$1$ has a winning strategy from $v_0$ in $(\G,\Obj)$ iff player~$1$ has a winning strategy from $(v_0,-1)$ in $(\G',\Obj')$ with $\Obj' = \Buchi(U')$.

The cases $\Obj = \LimInf(0)$ and $\Obj = \LimSup(0)$ are solved more easily. Again $V' = V \times \{-1,0\}$, but now $((v,a),(v',a')) \in E'$ iff $e = (v,v') \in E$ and $a' = 0$ if $w(e) \geq 0$, $a' = -1$ otherwise (we just remember the sign of the current weight in $\G$). The $\LimInf(0)$ (resp. $\LimSup$) objective is replaced by the $\CoBuchi(U')$ (resp. $\Buchi(U')$) objective with $U' = V \times \{0\}$.

It remains to study the case $\Obj = \DFWMP(\lambda,0)$ when $\n = 1$. We proceed like in~\cite{Chatterjee0RR15}. By Lemma~\ref{lem:iclosed}, a play $\rho$ is winning iff $\rho$ has a \GoodDec{\lambda}. This good decomposition can be supposed maximal by Lemma~\ref{lem:factorize}. Hence $\rho$ is winning iff there exists a sequence $(k_i)_{i \geq 0}$ with $k_0 = 0$ such that for each $i$, we have $k_{i+1} - k_i \in \{1,\ldots, \lambda\}$, and the \window{k_i} is \fclosed{\lambda}{k_{i+1}}. We thus keep in the current vertex of $\G'$ the sum of the weights of the current history in $\Hist(\G)$, and as soon as it turns non-negative (in at most $\lambda$ steps), we reset the sum counter to zero. The number of performed steps is also stored in the current vertex as a second counter. More precisely, we define $V'  = V \times \{-\lambda \cdot W, \ldots, 0\} \times \{0, \ldots, \lambda-1 \} \cup \{\beta\}$, where the additional vertex $\beta$ denotes a special vertex for the detection of a \bad{\lambda} window (the current sum is negative and the number of steps exceeds $\lambda$). Let $u = (v,s,l) \in V'$ be such that $v \in V$, $s$ is the current sum, and $l$ is the current number of steps. Given $e = (v,v') \in E$, we construct the edge $(u,u') \in E'$ outgoing $u$ such that
$$u' =
\begin{cases}
(v', s + w(e), l+1) &  \mbox{ if } s + w(e) < 0 \mbox{ and } l < \lambda -1 \\
\beta &                     \mbox{ if } s + w(e) < 0 \mbox{ and } l = \lambda -1 \\
(v',0,0)  &  \mbox{ if } s + w(e) \geq 0 
\end{cases}$$
We also add the edge $(\beta,\beta)$ to $E'$, and define $U' = V' \ssetminus \{\beta\}$. Clearly, player~$1$ has a winning strategy from $v_0$ in $(\G,\Obj)$ iff player~$1$ has a winning strategy from $(v_0,0,0)$ in $(\G',\Obj')$ with $\Obj' = \CoBuchi(U')$.

Let us turn to $\n$-weighted games $\G$ with $\n \geq 2$. Let $\n^*$ be the number of objectives $\Obj_m \in \ISL$ and $\n - \n^*$ be the number of objectives $\Obj_m = \DFWMP$. We proceed as in the case $\n = 1$ by adding some information in the vertices for each component $m \in \{1, \ldots, \n\}$: 
\begin{eqnarray} \label{eq:exp}
V' &=& V \times \{-1,0\}^{\n^*} \times \Big(\{-\lambda \cdot W, \ldots, 0\} \times \{0, \ldots, \lambda - 1\} ~ \cup ~ \{\beta\}\Big)^{n - n^*}.
\end{eqnarray}
The edges of $E'$ are defined component-wise as explained above\footnote{The same vertex $\beta$ could be used for each objective $\Obj_m = \DFWMP$, instead of one such vertex by component.}. The game structure $\G'$ has a size exponential in the initial game structure $\G$, that is, with $O(|V'|) = O(|V| \cdot (\lambda^2 \cdot W)^n)$ vertices and $O(|E'|) = O(|E| \cdot (\lambda^2 \cdot W)^n)$ edges.

Finally, any strategy in $G'$ that is memoryless or finite-memory transfers to a strategy in $G$ that has to be finite-memory in both cases.
\qed\end{proof}

The last reduction is concerned with a polynomial reduction of the class of games $(\G,\Obj = \cap_{m=1}^{\n} \Obj_m)$ with each $\Obj_m \in \Reg$ to the class of games $(\G',\Obj' = \cap_{m=1}^{\n} \Obj'_m)$ with each $\Obj'_m \in  \{\Sup, \Inf, \LimSup, \LimInf \}$.

\begin{proposition}\label{prop:WRegtoWISL}
Each unweighted game $(G,\Obj)$ with $\Obj = \cap_{m=1}^n \Obj_m$ such that for all $m$,
$$ \Obj_m \in \{\Reach,\Safe,\Buchi,\CoBuchi\}$$
can be polynomially reduced to an $n$-weighted game $(G',\Obj')$ with $|V|$ vertices and $|E|$ edges, and $\Obj' = \cap_{m=1}^n \Obj'_m$, such that for all $m$,
$$ \Obj'_m \in \{\Sup, \Inf, \LimSup, \LimInf \}$$
Moreover, if $\Obj_m$ is a $\Reach$ (resp. $\Safe, \Buchi, \CoBuchi$) objective, then $\Obj'_m$ is a $\Sup$ (resp. $\Inf$, $\LimSup$, $\LimInf$) objective.

\noindent A memoryless (resp. (polynomial, exponential) memory) strategy in $G'$ transfers to a memoryless (resp. (polynomial, exponential) memory) strategy in $G$.
\end{proposition}

\begin{proof}
Let $(\G,\Obj)$ with $\Obj = \cap_{m=1}^{\n} \Obj_m$ such that $\Obj_m \in \Reg$ for all $m$. We proceed like in the previous proofs by first considering 1-weighted games.

Let $n = 1$ and $\Obj = \Reach(U)$ (resp. $\Safe(U)$, $\Buchi(U)$, $\CoBuchi(U)$). We construct the new game $(G',\Obj')$ as follows. The game structure $G'$ is the game structure $G$ with a weight function $w$ such that, for all $e = (v,v') \in E$, $w(e) = 0$ if $v \in U$, $-1$ otherwise. We let $\Obj' = \Sup(0)$ (resp. $\Inf(0)$, $\LimSup(0)$, $\LimInf(0)$). Thus, the game $G'$ is built in a way to necessarily see a non-negative weight iff one visits a vertex of $U$ in $\G$. Clearly, player~$1$ has a winning strategy from $v_0$ in the $(G,\Obj)$ iff he has a winning strategy from $v_0$ in $(G',\Obj')$. 

Let us now suppose that $n \geq 2$. We proceed as in the case $n =1$ except that the weight function is $n$-dimensional, and is defined component-wise as explained above. Let us note that the game structure $G'$ has still the same size as $G$.

Finally, any strategy in $G'$ that is memoryless (resp. (polynomial, exponential) memory) transfers to a strategy in $G$ that is again memoryless (resp. (polynomial, exponential) memory).
\qed\end{proof}

\subsection{Proof of Theorem~\ref{thm:general}}

We here provide the proof of Theorem~\ref{thm:general} which uses the reduction of Proposition~\ref{prop:WISLtoBC}.

\begin{proof}[of Theorem~\ref{thm:general}]
First, we show that the \prob\ is in $\mathsf{EXPTIME}$. By Proposition~\ref{prop:WISLtoBC}, we can construct a game $(G',\Obj')$ with $\Obj' = \cap_{m=1}^n \Obj'_m$, such that for all $m$, $\Obj'_m \in \{\Buchi, \CoBuchi\}$, and the size of the game is exponential with $O(|V| \cdot (\lambda^2 \cdot W)^n)$ vertices and $O(|E| \cdot (\lambda^2 \cdot W)^n)$ edges. Recall that the intersection of co-B\"uchi objectives is a co-B\"uchi objective. It follows that $\G'$ is a generalized B\"uchi $\cap$ co-B\"uchi game. Then, by Theorem~\ref{thm:reg} (last item), we can compute the winning sets of both players in $\G$ in $O(n^2 \cdot |V'| \cdot |E'|) = O(|V| \cdot |E|  \cdot (\lambda^2 \cdot W)^{2n})$ time. Moreover, since polynomial memory strategies are sufficient in $\G'$, exponential memory strategies are sufficient in $\G$. Finally, the $\mathsf{EXPTIME}$-hardness and the necessity of exponentiel memory follow from Theorem~\ref{thm:WMP}. This establishes the result.
\qed\end{proof}
\section{Efficient fragment with one $\DFWMP$ objective} \label{sec:OneWindow}


In the previous section, we considered  games $(G, \Obj)$ with $\Obj = \cap_{m=1}^{\n} \Obj_m$ being any intersection of objectives in $\WISL$. In this section, we focus on a particular class of games in a way to achieve a lower complexity for the \prob. We do not consider the case where at least two $\Obj_m$ are $\DFWMP$ objectives because the \prob\ is already ${\sf EXPTIME}$-hard in this case (by Theorem~\ref{thm:WMP}). We thus focus on the intersections of exactly one\footnote{the case with no $\DFWMP$ objective will be treated in the next section.} objective $\DFWMP$ and any number of objectives of one kind in $\ISL$. Notice that this number must be fixed in the case of objectives $\Sup$ to avoid ${\sf PSPACE}$-hardness in this case (see Remark~\ref{rem:severalsup} below). For the considered fragment, we show that the \prob\ is $\sf P$-complete when $\lambda$ is polynomial in the size of the game structure. The latter assumption is reasonable in practical applications where one expects a positive mean-payoff in any ``short'' window sliding along the play.

\begin{theorem} \label{thm:2obj}
Let $(G, \Obj)$ be an $n$-weighted game with objective $\Obj = \Obj_1 \cap \Gamma$ for player~$1$ such that $\Obj_1 = \DFWMP$ and $\Gamma = \cap_{m=2}^{\n} \Obj_m$ such that $\forall m\, \Obj_m = \Inf$ (resp. $\forall m\, \Obj_m = \LimInf$, $\forall m\, \Obj_m = \LimSup$, $\{\forall m\, \Obj_m = \Sup$ and $\n$ is fixed$\}$).
Then the \prob\ is decidable (in time polynomial in the size of the game, $\lambda$ and $\lceil \log(W)\rceil$). In general, both players require finite-memory strategies, and pseudo-polynomial memory is sufficient for both players. Moreover, when $\lambda$ is polynomial in the size of the game then the \prob\ is $\sf P$-complete.
\end{theorem}

To prove this theorem, we use the polynomial reduction of Proposition~\ref{prop:WISLtoWreg} to obtain a game $(G', \Obj'_1 \cap \Gamma')$ (with $\Gamma' = \cap_{m=2}^{\n} \Obj'_m$) such that $\Obj'_1 = \DFWMP$ and $\forall m, \Obj'_m \in \Reg$. Recall that the intersection of safety (resp. co-B\"uchi) objectives is a safety (resp. co-B\"uchi) objective. We thus have to study $\DFWMP \cap \Safe$ (resp. $\DFWMP \cap \Reach$, $\DFWMP \cap \GenReach$ (with $n$ fixed), $\DFWMP \cap \Buchi$, $\DFWMP \cap \GenBuchi$, $\DFWMP \cap \CoBuchi$) games. Each of these six possible cases for $\Obj'_2$ is studied in a separate section; the proof of Theorem~\ref{thm:2obj} will follow. 

Table~\ref{table:poly} gives an overview of the properties of the six studied games.
This table indicates time polynomial in the size of the game, $\lambda$ and $\lceil \log(W)\rceil$, and pseudo-polynomial memories for both players. When $\lambda$ is polynomial in size of the game, the complexity becomes polynomial. Moreover, having pseudo-polynomial memories is not really a problem since the proofs show that the strategies can be efficiently encoded by programs using two counters, as in the case of Theorem~\ref{thm:WMP}.

\begin{table}
\scriptsize
\begin{center}
\begin{tabular}{|c|c|c|c|c|c|c|c|}
\hline
Objective           	 &  Algorithmic complexity                   & Player~$1$ memory & Player~$2$ memory \\
\hline
$\DFWMP \cap \Safe$  & $O(\lambda \cdot |V| \cdot (|V| + |E|) \cdot \lceil \log(\lambda \cdot W)\rceil)$       & $O(\lambda^2 \cdot W)$ & $O(\lambda^2 \cdot W \cdot |V|)$  \\
$\DFWMP \cap \Reach$  &  $O(\lambda \cdot |V| \cdot (|V| + |E|) \cdot \lceil \log(\lambda \cdot W)\rceil)$     & $O(\lambda^2 \cdot W \cdot |V|)$    & $O(\lambda^2 \cdot W \cdot |V|)$ \\
$\DFWMP \cap \GenReach$ \tiny{($i$ fixed})  &  $O(\lambda \cdot 2^{2i} \cdot |V| \cdot (|V| + |E|) \cdot \lceil \log(\lambda \cdot W)\rceil)$     & $O(\lambda^2 \cdot 2^i \cdot W \cdot |V|)$    & $O(\lambda^2 \cdot 2^i \cdot W \cdot |V|)$ \\
$\DFWMP \cap \Buchi$   &    $O(\lambda \cdot |V|^2 \cdot (|V| + |E|) \cdot \lceil \log(\lambda \cdot W)\rceil)$  & $O(\lambda^2 \cdot W \cdot |V|)$     & $O(\lambda^2 \cdot W \cdot |V|^2)$  \\
$\DFWMP \cap \GenBuchi$      & $O(\lambda \cdot i^3 \cdot |V|^2 \cdot (|V| + |E|) \cdot \lceil \log(\lambda \cdot W)\rceil)$  & $O(\lambda^2 \cdot W \cdot i \cdot |V|)$     & $O(\lambda^2 \cdot W \cdot i^2 \cdot |V|^2)$ \\
$\DFWMP \cap \CoBuchi$    & $O(\lambda \cdot |V|^2 \cdot (|V| + |E|) \cdot \lceil \log(\lambda \cdot W)\rceil)$ &  $O(\lambda^2 \cdot W \cdot |V|^2)$      & $O(\lambda^2 \cdot W \cdot |V|)$ \\
\hline
\end{tabular}
\end{center}
\caption{Overview of the fragment ($i$ is the number of objectives in the intersection of reachability/B\"uchi objectives)}
\label{table:poly}
\end{table}

\subsection{Objective $\DFWMP(\lambda,0) \cap \Safe(U)$}

We begin with the easy case $\Obj'_2 = \Safe$. To solve this game, we first compute the winning set $X_1$ for the objective $\Safe(U)$ and then the winning set in the subgame $\G[X_1]$ for the objective $\DFWMP(\lambda,0)$.

\begin{proposition} \label{prop:WMPSafe}
Let $\G = \Gdev$ be a $1$-weighted game structure, and assume that $\Obj$ is the objective $\DFWMP(\lambda,0) \cap \Safe(U)$.
Then $\WinG{1}{\Obj}{\G}$ can be computed in $O(\lambda \cdot |V| \cdot (|V| + |E|) \cdot \lceil \log_2(\lambda \cdot W) \rceil)$ time and strategies with memory in $O(\lambda^2 \cdot W)$ (resp. in $O(\lambda^2 \cdot W \cdot |V|)$) are sufficient for player~$1$ (resp. player~$2$).
\end{proposition}

\begin{proof}
In a first step we compute the winning set $X_1 = \WinG{1}{\Safe(U)}{\G}$ and $X_2 = V \ssetminus X_1$ in the game $(\G,\Safe(U))$, and in a second step we compute the winning set $Y_1 = \WinG{1}{\DFWMP(\lambda,0)}{\G[X_1]}$ and $Y_2 = X_1 \ssetminus Y_1$ in the game $(\G[X_1],\DFWMP(\lambda,0))$. Notice that $X_1 \subseteq U$, and $\G[X_1]$ is a game structure since $X_1$ is 2-closed by Corollary~\ref{cor:reg}. Let us prove that $Y_1 = \WinG{1}{\Obj}{\G}$ by showing that $Y_1 \subseteq \WinG{1}{\Obj}{\G}$ and $X_2 \cup Y_2 \subseteq \WinG{2}{\overline{\Obj}}{\G}$. 

Let $v \in Y_1$, and $\sigma_1$ be a winning strategy of player~$1$ for the objective $\DFWMP(\lambda,0)$  in $\G[X_1]$. Since $X_1$ is 2-closed and $X_1 \subseteq U$, $\sigma_1$ is also winning in $\G$ for this objective while staying in $U$. This shows that $v \in  \WinG{1}{\Obj}{\G}$.

Clearly $X_2 \subseteq \WinG{2}{\overline{\Obj}}{\G}$ since player~$2$ can force to reach $V \setminus U$ from any initial vertex $v \in X_2$. Let $v \in Y_2$. We define a strategy $\sigma_2$ for player~$2$ from $v$ as follows: while staying in $X_1$, play a winning strategy for the objective $\overline{\DFWMP(\lambda,0)}$ in $\G[X_1]$, and as soon as player~$1$ decides to go to $X_2$, shift to a winning strategy to reach $V \setminus U$. This strategy $\sigma_2$ clearly shows that $v \in \WinG{2}{\overline{\Obj}}{\G}$.

Let us now study the time complexity and the memory requirements. We easily get from Theorems~\ref{thm:attr} and~\ref{thm:WMP} that $Y_1$ can be computed in $O(\lambda \cdot |V| \cdot (|V| + |E|) \cdot \lceil \log_2(\lambda \cdot W) \rceil)$ time, and that both players have finite-memory strategies with a memory in $O(\lambda^2 \cdot W)$ for player~$1$ (resp. in $O(\lambda^2 \cdot W \cdot |V|)$ for player~$2$). 
\qed\end{proof}

\subsection{Objectives $\ICWReach{\lambda}{U}$ and $\GDReach{\lambda}{U}$}

The case of games $(\G,\Obj_1 \cap \Obj_2)$ with $\Obj_1 = \DFWMP(\lambda,0)$ and $\Obj_2 \in \{\Reach, \GenReach, \Buchi,\GenBuchi$, $\CoBuchi\}$ (with $n$ fixed when $\Obj_2$ is a generalized reachability objective) is much more involved. To solve them, we need to develop some new tools generalizing the concept of $p$-attractor while dealing with good windows. To this end, let us introduce the next two new objectives (see Figure~\ref{fig:ICWReach}).

\begin{definition}
Let $\G = \Gdev$ be a $1$-weighted game structure, $U \subseteq V$ be a set of vertices, and $\lambda \in \N \setminus\{0\}$ be a window size. We consider the next two sets of plays:
\begin{itemize}
\item $\ICWReach{\lambda}{U} = \{ \rho \in \Plays(\G) \mid \exists l \in \{1,\ldots, \lambda\}, \rho_l \in U$, and the \window{0} is \iclosed{\lambda}{l}  $\}$, 
\item $\GDReach{\lambda}{U} = \{ \rho \in \Plays(\G) \mid \exists l \geq 0, \rho_l \in U \mbox{ and } \rhofactor{0}{l}, \mbox{ has a  \GoodDec{\lambda}} \}$.
\end{itemize}
\end{definition}

\begin{figure}
\centering
\begin{tikzpicture}[scale=3]

	\draw (-0.7,0.5) node[] (UU) {$\ICWReach{\lambda}{U}:$};
	\draw (0,0.5) node[circle,draw] (ZZ) {$\rho_0$};
	\draw (1.5,0.5) node[circle,draw] (YY) {$\rho_l$};
	\draw (1.6,0.65) node[] (XX) {$\in U$};
	\draw (1.9,0.5) node (AA) {$\ldots$};
	\draw (ZZ) to[bend left=20] (YY);
	\draw (ZZ) -- (YY);
	\draw (YY) -- (AA);

	\draw (-0.7,0) node[] (U) {$\GDReach{\lambda}{U} :$};
	\draw (0,0) node[circle,draw] (Z) {$\rho_0$};
	\draw (3,0) node[circle,draw] (Y) {$\rho_l$};
	\draw (3.1,0.15) node[] (X) {$\in U$};
	\draw (Z) -- (Y); 
	\draw (3.4,0) node (A) {$\ldots$};
	\draw (Y) -- (A);
	
	\draw (Z) to[bend left=25] (1,0);	
	\draw (1,0) to[bend left] (2.1,0);
	\draw (2.1,0) to[bend left] (Y);

\end{tikzpicture}

\caption{A play that belongs to $\ICWReach{\lambda}{U}$ and a play that belongs to $\GDReach{\lambda}{U}$}
\label{fig:ICWReach}
\end{figure}

\noindent Notice that the plays of $\ICWReach{\lambda}{U}$ are the particular plays of $\GDReach{\lambda}{U}$ such that the \GoodDec{\lambda} of $\rhofactor{0}{l}$ has size $1$. Hence $\ICWReach{\lambda}{U} \subseteq \GDReach{\lambda}{U}$. We propose two algorithms, Algorithms~\ref{algo:ICWReach} and~\ref{algo:GDReach}, one for computing the winning set of player~$1$ for the objective $\ICWReach{\lambda}{U}$, and the other one for the objective $\GDReach{\lambda}{U}$.

Algorithm~\ref{algo:ICWReach} uses the next operator $\oplus$. Let $a, b \in \Z \cup \{-\infty\}$,
$$a \oplus b = 
\begin{cases}
a + b & \mbox{ if } a+b \geq 0 \\
-\infty & \mbox{otherwise}
\end{cases}$$

\noindent
With this definition, either $a \oplus b \geq 0$ or $a \oplus b = -\infty$. One can check that if $b \geq b'$, then $a \oplus b \geq a \oplus b'$. Moreover if $a \oplus b \geq 0$, then $a \oplus b = a + b \geq 0$. Algorithm~\ref{algo:ICWReach} intuitively works as follows. Given a vertex $v$ and a number $i$ of steps, the value $C_i(v)$ is computed iteratively (from $C_{i-1}(v)$) and represents the best total payoff that player~$1$ can ensure in \emph{at most $i$ steps} while closing the window from $v$ in a vertex of $U$. The value $-\infty$ indicates that the window starting in $v$ cannot be inductively-closed in $U$. Therefore, the winning set of player~$1$ is the set of vertices $v$ for which $C_{\lambda}(v) \geq 0$. This algorithm is inspired from Algorithm~$\sf GoodWin$ in~\cite{Chatterjee0RR15} that computes $\WinG{1}{\ICWReach{\lambda}{U}}{\G}$ with $U = V$.\footnote{We have detected a flaw in this algorithm that has been corrected in Algorithm~\ref{algo:ICWReach}. The algorithm in~\cite{Chatterjee0RR15} wrongly computes the set of vertices from which player~$1$ can force to close the window in exactly $l$ steps (instead of at most $l$ steps) for some $l \in \{1,\ldots,\lambda\}$.}

\begin{algorithm}
\caption{$\AlgoICWReach$}
\label{algo:ICWReach}
\begin{algorithmic}[1]
\REQUIRE $1$-weighted game structure $\G = \Gdev$, set $U \subseteq V$, window size $\lambda \in \N \setminus\{0\}$ 
\ENSURE  $\WinG{1}{\ICWReach{\lambda}{U}}{\G}$
\FORALL {$v \in V$} 
\IF {$v \in U$}
\STATE $C_0(v) \leftarrow 0$
\ELSE
\STATE $C_0(v) \leftarrow -\infty$
\ENDIF 
\ENDFOR
\FORALL {$l \in \{1, \ldots, \lambda\}$}
\FORALL {$v \in V_1$}
\STATE $C_l(v) \leftarrow \max_{(v,v') \in E} \{ w(v,v') \oplus \max \{ C_0(v') , C_{l-1}(v') \}  \}$
\ENDFOR
\FORALL {$v \in V_2$}
\STATE $C_l(v) \leftarrow \min_{(v,v') \in E} \{ w(v,v') \oplus \max \{ C_0(v') , C_{l-1}(v') \}  \}$
\ENDFOR
\ENDFOR
\RETURN $\{ v \in V \mid C_{\lambda}(v)  \geq 0\}$
\end{algorithmic}
\end{algorithm}  

\begin{lemma} \label{lem:ICWReach}
Let $\G$ be a $1$-weighted game structure, $U$ be a subset of $V$, and $\lambda \in \N \setminus\{0\}$ be a window size. Then Algorithm $\AlgoICWReach$ computes the set $\WinG{1}{\ICWReach{\lambda}{U}}{\G}$ of winning vertices of player~$1$ for the objective $\ICWReach{\lambda}{U}$ in $O(\lambda \cdot (|V| + |E|) \cdot \lceil \log_2(\lambda \cdot W) \rceil)$ time, and finite-memory strategies with memory linear in $\lambda$ are sufficient for both players.
\end{lemma}

\begin{proof}
Let us prove that Algorithm~$\AlgoICWReach$ correctly computes the set $\WinG{1}{\ICWReach{\lambda}{U}}{\G}$. Notice that for all $v$ and $l$, either $C_l(v) \geq 0$ or $C_l(v) = -\infty$.

Let $v_0$ be an initial vertex. Consider the following finite-memory strategy $\sigma_1^*$ (resp. $\sigma_2^*$) of player~$1$ (resp. player~$2$): The memory set is $\{0, \ldots, \lambda\}$, the initial memory state is $\lambda$ at the initial vertex $v_0$. At vertex $v$ and memory state $l$, if $l \geq 1$ then player~$1$ (resp. player~$2$) chooses a vertex $v'$ that maximizes (resp. minimizes) $w(v,v') \oplus \max \{ C_0(v') , C_{l-1}(v') \}$, and if $l = 0$ then he chooses any edge. After each move, the memory state $l$ is updated to $l-1$ until always keeping value $0$. We will show that $(i)$ if $C_{\lambda}(v_0)  \geq 0$ then $\sigma_1^*$ is a winning strategy for $\ICWReach{\lambda}{U}$ and $(ii)$ if $C_{\lambda}(v_0) = -\infty$ then $\sigma_2^*$ is a winning strategy for $\overline{\ICWReach{\lambda}{U}}$. These winning strategies are finite-memory and have both a memory in $O(\lambda)$.

$(i)$ Take $\sigma_1^*$ and suppose that $C_{\lambda}(v_0)  \geq 0$. Let $\sigma_2$ be any strategy of player~$2$, and let $\rho = \rho_0\rho_1 \ldots = \Out(v_0,\sigma_1^*,\sigma_2)$. To show that $\rho \in \ICWReach{\lambda}{U}$, we are going to prove for all $k \in \{0,\ldots,\lambda-1 \}$ that

\begin{enumerate}
\item[$(*)$] if $C_{\lambda - k}(\rho_k) \geq 0$ then there exists $l \in \{k+1, \ldots, \lambda\}$ such that $\rho_l \in U$, the \window{k} is \iclosed{\lambda}{l}, and $\TP(\rhofactor{k}{l}) \geq C_{\lambda - k}(\rho_k)$.
\end{enumerate} 

\noindent In particular, when $k = 0$, we will get from $C_{\lambda}(v_0)  \geq 0$ the existence of $l \in \{1, \ldots, \lambda\}$ such that $\rho_l \in U$ and the \window{0} is \iclosed{\lambda}{l}, that is, $\rho \in \ICWReach{\lambda}{U}$.

Before proving $(*)$, we establish the following inequality~\ref{eq:C}.
Consider $C_{\lambda - k}(\rho_k) \geq 0$ with $k \in \{0,\ldots,\lambda-1 \}$, and the edge $(\rho_k,\rho_{k+1})$. We have
\begin{eqnarray} \label{eq:C}
C_{\lambda - k}(\rho_k) &\leq& w(\rho_k,\rho_{k+1}) \oplus \max \{ C_0(\rho_{k+1}) , C_{\lambda - k-1}(\rho_{k+1}) \}
\end{eqnarray}
Indeed if $\rho_k \in V_1$, by definition of $\sigma_1^*$, we have (\ref{eq:C}) with an equality (instead of an inequality). If $\rho_k \in V_2$, we have (\ref{eq:C}) by definition of $C_{\lambda - k}(\rho_k)$ and since player~$2$ does not necessarily play optimally. 

Let us now prove $(*)$. Suppose firstly that (\ref{eq:C}) is realized by
\begin{eqnarray} \label{eq:C1}
C_{\lambda - k}(\rho_k) \leq w(\rho_k,\rho_{k+1}) \oplus C_0(\rho_{k+1}).
\end{eqnarray}
As $C_{\lambda - k}(\rho_k) \geq 0$, it follows that $C_0(\rho_{k+1}) = 0$, that is, $\rho_{k+1} \in U$, and $w(\rho_k,\rho_{k+1}) \geq C_{\lambda - k}(\rho_k) \geq 0$. We thus get $(*)$ with $l = k+1$.

Suppose secondly that (\ref{eq:C}) is realized by 
\begin{eqnarray} \label{eq:C2}
C_{\lambda - k}(\rho_k) \leq w(\rho_k,\rho_{k+1}) \oplus C_{\lambda - k-1}(\rho_{k+1}).
\end{eqnarray}
In this case, we prove $(*)$ by induction on $k \in \{0,\ldots,\lambda -1\}$. We begin with $k = \lambda -1$ (basic case of the induction). Then (\ref{eq:C2}) is equal to (\ref{eq:C1}) already treated. We can thus proceed with $k < \lambda -1$. As $C_{\lambda - k}(\rho_k) \geq 0$, it follows from (\ref{eq:C2}) that $C_{\lambda - k-1}(\rho_{k+1}) \geq 0$, and 
\begin{eqnarray} \label{eq:C3}
w(\rho_k,\rho_{k+1}) + C_{\lambda - k-1}(\rho_{k+1}) \geq C_{\lambda - k}(\rho_k) \geq 0.
\end{eqnarray} 
By induction hypothesis (with $k+1$), from $(*)$ and $C_{\lambda - k-1}(\rho_{k+1}) \geq 0$, we get the existence of $l \in \{k+2, \ldots, \lambda\}$ such that $\rho_l \in U$, the \window{k+1} is \iclosed{\lambda}{l}, and $\TP(\rhofactor{k+1}{l}) \geq C_{\lambda - k - 1}(\rho_{k+1})$. To get $(*)$ for $k$, with this index $l$, it remains to proves that $\TP(\rhofactor{k}{l}) \geq C_{\lambda - k}(\rho_k) \geq 0$. This follows from $\TP(\rhofactor{k}{l}) = w(\rho_k,\rho_{k+1}) + \TP(\rhofactor{k+1}{l}) \geq w(\rho_k,\rho_{k+1}) + C_{\lambda - k-1}(\rho_{k+1})$ and (\ref{eq:C3}).

\smallskip
$(ii)$ Take $\sigma_2^*$ and suppose that $C_{\lambda}(v_0)  = - \infty$. We have to prove that $\sigma_2^*$ is a winning strategy for $\overline{\ICWReach{\lambda}{U}}$. Assume the contrary: then there exists a strategy $\sigma_1$ of player~$1$ such that the play $\rho = \Out(v_0,\sigma_1,\sigma_2^*)$ belongs to $\ICWReach{\lambda}{U}$. Therefore, there exists $l \in \{1, \ldots, \lambda \}$ such that $\rho_l \in U$ and the \window{0} is \iclosed{\lambda}{l}. To have a contradiction, we will show that $C_{\lambda}(v_0) \geq 0$. More generally, we are going to prove by induction on $k \in \{0, \ldots, l-1\}$ that
\begin{eqnarray} \label{eq:D}
C_{\lambda - k}(\rho_k) &\geq& \TP(\rhofactor{k}{l}).
\end{eqnarray}
In particular, with $k = 0$, we will get $C_{\lambda}(v_0) \geq \TP(\rhofactor{0}{l}) \geq 0$.

\noindent
Consider $C_{\lambda - k}(\rho_k)$ with $k \in \{0,\ldots,l-1 \}$, and the edge $(\rho_k,\rho_{k+1})$. We have
\begin{eqnarray} \label{eq:D1}
C_{\lambda - k}(\rho_k) &\geq& w(\rho_k,\rho_{k+1}) \oplus \max \{ C_0(\rho_{k+1}) , C_{\lambda - k-1}(\rho_{k+1}) \}.
\end{eqnarray}
Indeed if $\rho_k \in V_2$, by definition of $\sigma_2^*$, we have (\ref{eq:D1}) with an equality. If $\rho_k \in V_1$, we have (\ref{eq:D}) by definition of $C_{\lambda - k}(\rho_k)$ and since player~$1$ does not necessarily play optimally. 
Suppose that $k = l-1$ (basic case of the induction). By (\ref{eq:D1}) and since $\rho_l \in U$ and the \window{0} is \iclosed{\lambda}{l}, we get that $C_{\lambda - l + 1}(\rho_k) \geq w(\rho_{l-1},\rho_{l}) \oplus  C_0(\rho_{l}) = w(\rho_{l-1},\rho_{l})$, thus proving (\ref{eq:D}) for $k = l-1$. Suppose now that $k < l-1$. By (\ref{eq:D1}) and induction hypothesis (with $k+1$) , we have 
\begin{eqnarray} \label{eq:D2}
C_{\lambda - k}(\rho_k) \geq w(\rho_k,\rho_{k+1}) \oplus C_{\lambda - k-1}(\rho_{k+1}) \geq w(\rho_k,\rho_{k+1}) \oplus \TP(\rhofactor{k+1}{l}).
\end{eqnarray} 
As the \window{0} is \iclosed{\lambda}{l}, $\TP(\rhofactor{k}{l}) = w(\rho_k,\rho_{k+1}) + \TP(\rhofactor{k+1}{l}) \geq 0$, and then $\TP(\rhofactor{k}{l}) = w(\rho_k,\rho_{k+1}) \oplus \TP(\rhofactor{k+1}{l})$. Thus, we get with (\ref{eq:D2}) that $C_{\lambda - k}(\rho_k) \geq \TP(\rhofactor{k}{l})$ and then (\ref{eq:D}) holds for $k$.

\smallskip
It remains to discuss the complexity of the algorithm. Clearly, it takes a number of elementary arithmetic operations which is bounded by $O(\lambda \cdot (|V| + |E|))$. As each $C_l(v)$ is bounded by $\lambda \cdot W$, each elementary arithmetic operation takes time linear in the number of bits of the encoding of $\lambda \cdot W$. Hence, the time complexity of the algorithm is in  $O(\lambda \cdot (|V| + |E|) \cdot \log_2(\lambda \cdot W))$ time. 
\qed\end{proof}

We now turn to Algorithm~\ref{algo:GDReach} for computing the winning set of player~$1$ for the objective $\GDReach{\lambda}{U}$. It shares similarities with the algorithm computing the $p$-attractor $\Attr{p}{U}{G}$ while requiring to use previous Algorithm~$\ICWReach{\lambda}{U}$.

\begin{algorithm}
\caption{$\AlgoGDReach$}
\label{algo:GDReach}
\begin{algorithmic}[1]
\REQUIRE $1$-weighted game structure $\G = \Gdev$, subset $U \subseteq V$, window size $\lambda  \in \N \setminus\{0\}$ 
\ENSURE  $\WinG{1}{\GDReach{\lambda}{U}}{\G}$
\STATE $k \leftarrow 0$
\STATE $X_0 \leftarrow U$
\REPEAT
\STATE $X_{k+1} \leftarrow X_k \cup \mathsf{\AlgoICWReach}(\G,X_k,\lambda)$
\STATE $k \leftarrow k+1$
\UNTIL {$X_k = X_{k-1}$}
\RETURN $X_k$
\end{algorithmic}
\end{algorithm}

\begin{lemma} \label{lem:GDReach}
Let $\G$ be a $1$-weighted game structure, $U$ be a subset of $V$, and $\lambda \in \N \setminus\{0\}$ be a window size. Then Algorithm $\AlgoGDReach$ computes the set $\WinG{1}{\GDReach{\lambda}{U}}{\G}$ of winning vertices of player~$1$ for the objective $\GDReach{\lambda}{U}$ in $O(\lambda \cdot |V| \cdot (|V| + |E|) \cdot \lceil \log_2(\lambda \cdot W) \rceil)$ time, and finite-memory strategies with memory in $O(\lambda^2 \cdot W \cdot |V|)$ (resp. in $O(\lambda^2 \cdot W)$) are sufficient for player~$1$ (resp. player~$2$).
\end{lemma}

\begin{proof}
Let $X^* = \cup_{k\geq 0} X_k$ be the set computed by Algorithm $\AlgoGDReach$.

Suppose that $v_0 \in X^*$, and let us prove that $v_0 \in \WinG{1}{\GDReach{\lambda}{U}}{\G}$ by induction on $k$ such that $v_0 \in  X_k$. This is trivially true for $k = 0$ since $X_0 = U$. Let $k > 0$ be such that $v_0 \in X_k \setminus X_{k-1}$. Then $v_0 \in \WinG{1}{\ICWReach{\lambda}{X_{k-1}}}{\G}$ by Lemma~\ref{lem:ICWReach}. We propose the following strategy $\sigma_1$ of player~$1$ from $v_0$: play a winning strategy for the objective $\ICWReach{\lambda}{X_{k-1}}$, and when this objective is realized play a winning strategy for the objective $\GDReach{\lambda}{U}$ (the second strategy exists by induction hypothesis). Clearly, $\sigma_1$ is a winning strategy for $\GDReach{\lambda}{U}$ showing that $v_0 \in \WinG{1}{\GDReach{\lambda}{U}}{\G}$. Notice that this winning strategy, that consists in repeatedly playing a winning strategy for the objective $\ICWReach{\lambda}{X_{k}}$ for decreasing values of $k$ until reaching $U$, is finite-memory with a memory in $O(\lambda^2 \cdot W \cdot |V|)$. Indeed, for each $k$ (and there are at most $|V|$ such $k$) knowing that the objective $\ICWReach{\lambda}{X_{k}}$ is realized requires to memorize the number of steps (up to $\lambda$) and the sum of the weights inside the current window, and to check if $X_k$ is reached when this window becomes closed.

For the converse, we are going to prove that if $v_0 \in V \setminus X^*$ then $v_0 \in \WinG{2}{\overline{\GDReach{\lambda}{0}}}{G}$. Let $v_0 \in V \setminus X^*$, then player~$2$ has a winning strategy $\sigma_2^*$ for the objective $\overline{\ICWReach{\lambda}{X^*}}$. Notice that against any strategy $\sigma_1$ of player~$1$, for $\rho = \Out(v_0,\sigma_1,\sigma^*_2)$, if the \window{0} is \iclosed{\lambda}{\rho_l} for $l \in \{1,\ldots,\lambda\}$, then $\rho_l \in V \setminus X^*$. Consider the following strategy $\sigma_2$ of player~$2$ from $v_0$.
\begin{enumerate}
\item If the current vertex $v$ belongs to $V \setminus X^*$, play the winning strategy $\sigma^*_2$.
\item As soon as the window starting from $v$ is \fclosed{\lambda}{v'} in $l$ steps with $l \in \{1,\ldots,\lambda\}$, as $v' \in V \setminus X^*$, go back to 1.
\item As soon as the window starting from $v$ is \bad{\lambda}, play whatever.
\end{enumerate}
Let us show that $\sigma_2$ is a winning strategy of player~$2$ for the objective $\overline{\GDReach{\lambda}{U}}$. Let $\sigma_1$ be any strategy of player~$1$ and consider $\rho = \Out(v_0,\sigma_1,\sigma_2)$. Assume by contradiction that there exists $l \geq 0$ such that $\rho_l \in U$ and $\rhofactor{0}{l}$ has a \GoodDec{\lambda} $(k_i)_{i = 0}^{j}$. By Lemma~\ref{lem:factorize}, we can suppose that this good decomposition is maximal, that is, for each $i$ the \window{k_i} is \fclosed{\lambda}{k_{i+1}}. Therefore by construction of $\sigma_2$, $\rho_{k_i} \in V \setminus X^*$ for all $i$. In particular, $\rho_l \in V \setminus X^*$ in contradiction with $\rho_l \in U$ (recall that $U \subseteq X^*$). Notice that this winning strategy is finite-memory with a memory in $O(\lambda^2 \cdot W)$. Indeed, when playing $\sigma^*_2$, knowing that the current window is \ffclosed{\lambda} or \bad{\lambda} requires to memorize the number of steps (up to $\lambda$) and the sum of the weights inside this window.


It remains to discuss the complexity of the algorithm. The loop of Algorithm $\AlgoGDReach$ is executed at most $|V|$ times, and each step in this loop is in $O(\lambda \cdot (|V| + |E|) \cdot \log_2(\lambda \cdot W))$ time by Lemma~\ref{lem:ICWReach}. Then, the algorithm runs in $O(\lambda \cdot |V| \cdot (|V| + |E|) \cdot \log_2(\lambda \cdot W))$ time. 
\qed\end{proof}

Let us illustrate Algorithms~\ref{algo:ICWReach}~and~\ref{algo:GDReach} on the following example.

\begin{example}
Consider the game $(G,\Obj)$ depicted on Figure~\ref{ex:windowICWEnd}, where the objective $\Obj = \GDReach{2}{U}$ with $U = \{v_1,v_3,v_4\}$. Let us execute Algorithm~\ref{algo:GDReach}: $X_0 = U$, $X_1 = \WinG{1}{\ICWReach{2}{X_0}}{G} = \{v_1,v_2,v_3,v_4\}$ and then, $X_2 = \WinG{1}{\ICWReach{2}{X_1}}{G} = V$. Thus all vertices in $G$ are winning for player~$1$ for the objective $\Obj$. A winning strategy for player~$1$ consists in looping once in $v_2$ and then going to $v_3$. Indeed for any strategy of player~$2$, the outcome is either $v_0v_1^{\omega}$ or $v_0v_2v_2v_3v_4^{\omega}$, and both outcomes admit a prefix which has a \GoodDec{2} and ends with a vertex of $U$.
\end{example}

\begin{figure}[h]
\begin{minipage}[c]{.50\linewidth}
  \centering
  \begin{tikzpicture}[scale=3]
    \everymath{\scriptstyle}
    \draw (0,0) node [rectangle, inner sep = 5pt, draw] (A) {$v_0$};
    \draw (0.5,0.25) node [circle, draw] (B) {$v_2$};
    \draw (0.5,-0.25) node [circle, draw,fill=lightgray] (C) {$v_1$};
    \draw (1,0) node [circle, draw,fill=lightgray] (D) {$v_3$};
    \draw (1.5,0) node [circle, draw,fill=lightgray] (E) {$v_4$};
    
    \draw[->,>=latex] (A) to node[above,midway,sloped] {$-1$} (B);
    \draw[->,>=latex] (A) to node[below,midway,sloped] {$0$} (C);
     \draw[->,>=latex] (B) to node[above,midway,sloped] {$-1$} (D);
    \draw[->,>=latex] (D) to node[above,midway,sloped] {$1$} (E);
    
    \draw[->,>=latex] (C) .. controls +(35:0.4cm) and +(325:0.4cm) .. (C) node[right,midway] {$0$};
    \draw[->,>=latex] (E) .. controls +(305:0.4cm) and +(235:0.4cm) .. (E) node[below,midway] {$0$};
    \draw[->,>=latex] (B) .. controls +(305:0.3cm) and +(235:0.3cm) .. (B) node[below,midway] {$1$};
	\path (0,0.25) edge [->,>=latex] (A);
	    \end{tikzpicture}
\caption{Objective $\GDReach{2}{U}$}
\label{ex:windowICWEnd}
\end{minipage}\hfill
\begin{minipage}{.50\linewidth}
    \centering
  \begin{tikzpicture}[scale=3]
    \everymath{\scriptstyle}
    \draw (0,0) node [rectangle, inner sep = 5pt, draw] (A) {$v_0$};
    \draw (0.5,0.25) node [circle, draw] (B) {$v_2$};
    \draw (0.5,-0.25) node [circle, draw,double] (C) {$v_1$};
    \draw (1,0) node [circle, draw,double] (D) {$v_3$};
    \draw (1.5,0) node [circle, draw] (E) {$v_4$};
    
    \draw[->,>=latex] (A) to node[above,midway,sloped] {$-1$} (B);
    \draw[->,>=latex] (A) to node[below,midway,sloped] {$0$} (C);
     \draw[->,>=latex] (B) to node[above,midway,sloped] {$-1$} (D);
    \draw[->,>=latex] (D) to node[above,midway,sloped] {$1$} (E);
    
    \draw[->,>=latex] (C) .. controls +(35:0.4cm) and +(325:0.4cm) .. (C) node[right,midway] {$0$};
    \draw[->,>=latex] (E) .. controls +(305:0.4cm) and +(235:0.4cm) .. (E) node[below,midway] {$0$};
    \draw[->,>=latex] (B) .. controls +(305:0.3cm) and +(235:0.3cm) .. (B) node[below,midway] {$1$};
	\path (0,0.25) edge [->,>=latex] (A);

  \end{tikzpicture}
\caption{Objective $\DFWMP(2,0) \cap \Reach(U)$}
\label{ex:windowreach}
\end{minipage}
\end{figure}

\subsection{Objective $\DFWMP(\lambda,0) \cap \Reach(U)$}

In the previous section, we have proposed an algorithm 
for computing the winning set of player~$1$ for the game $(\G,\GDReach{\lambda}{U})$. Thanks to this algorithm, we can compute the winning set of player~$1$ for the game $(\G,\DFWMP(\lambda,0) \cap \Reach(U))$ (in time polynomial in the size of the game, $\lambda$ and $\lceil \log_2(W) \rceil$).
First notice that the objectives $\GDReach{\lambda}{U}$ and $\DFWMP(\lambda,0) \cap \Reach(U)$ are close to each other: a play $\rho$ belongs to $\GDReach{\lambda}{U}$ if it has a prefix which has a \GoodDec{\lambda} and ends with a vertex in $U$, while $\rho$ belongs to $\DFWMP(\lambda,0) \cap \Reach(U)$ if it has a \GoodDec{\lambda} and one of its vertices belongs to $U$. Therefore, to solve the game $(\G,\DFWMP(\lambda,0) \cap \Reach(U))$, we first add a bit to each vertex indicating whether or not a vertex of $U$ has been seen along the current history, and we then compute the winning set for the objective $\GDReach{\lambda}{U'}$ for the set $U'$ of vertices that indicate a visit of $U$ and are winning for the $\DFWMP$ objective.

\begin{proposition} \label{prop:WMPReach}
Let $\G = \Gdev$ be a $1$-weighted game structure, and assume that $\Obj$ is the objective $\DFWMP(\lambda,0) \cap \Reach(U)$.
Then $\WinG{1}{\Obj}{\G}$ can be computed in $O(\lambda \cdot |V| \cdot (|V| + |E|) \cdot \lceil \log_2(\lambda \cdot W) \rceil)$ time, and finite-memory strategies with memory in $O(\lambda^2 \cdot W \cdot |V|)$  are sufficient for both  players.
\end{proposition}

\begin{proof}
From $\G$ and $U$, we construct a new game structure $\G' = (V'_1,V'_2,E',w')$ with a bit added to the vertices of $V$ that indicates whether $U$ has been already visited (the bit equals $1$) or not (the bit equals $0$). More precisely, we have $V'_1 = V_1 \times \{0,1\}$, $V'_2 = V \times \{0,1\}$; given $(v,v') \in E$, an edge $((v,b),(v',b'))$ belongs to $E'$ such that $b'=1$ if ($b=1$ or $v' \in U$), and $b' = 0$ otherwise, and $w'((v,b),(v',b')) = w(v,v')$. We also define $U' = V \times \{1\}.$

For each $v \in V$, we denote by $(v,b_v)$ the vertex of $V'$ such that $b_v = 1$ if $v \in U$, and $b_v = 0$ otherwise. Notice that to each play $\rho$ of $\G$ from vertex $v$ corresponds a unique play $\rho'$ of $\G'$ from vertex $(v,b_v)$. The converse is also true. Moreover, if $\rhofactor{0}{l}$ is a prefix of $\rho$, with $l \geq 0$, such that $\rhofactor{0}{l}$ has a \GoodDec{\lambda} and  $\rho_k \in U$ for some $k \leq l$, then $\rho'$ belongs to $\GDReach{\lambda}{U'}$. The converse is also true. Finally, a strategy on $\G$ can be easily translated to a strategy on $\G'$ and conversely. In the sequel of the proof, we will often shift from one game to the other game.

The winning set of player~$1$ for the objective $\Obj$ is computed as follows. In $\G'$, we first compute the set $X' = U' \cap \WinG{1}{\DFWMP(\lambda,0)}{\G'}$, and then the set $Y' = \WinG{1}{\GDReach{\lambda}{X'}}{\G'}$. Let $Y = \{v \in V \mid (v,b_v) \in Y'\}$. We want to show that $Y = \WinG{1}{\Obj}{G}$. Notice that this set can be computed in $O(\lambda \cdot |V| \cdot (|V| + |E|) \cdot \lceil \log_2(\lambda \cdot W) \rceil)$ time by Theorem~\ref{thm:WMP} and Lemma~\ref{lem:GDReach}.

Let $v \in Y$, that is, $(v,b_v) \in Y'$. We define the next strategy $\sigma'_1$ on $G'$: begin with a winning strategy for the objective $\GDReach{\lambda}{X'}$, and as soon as this objective is realized, shift to a winning strategy for the objective $\DFWMP(\lambda,0)$. This strategy $\sigma'_1$ will force a play $\rho'$ to have a \GoodDec{\lambda} $(k_i)_{i \geq 0}$ such that $\rho'_{k_i} \in U'$ for some $k_i$. Therefore, the strategy $\sigma_1$ on $\G$ corresponding to $\sigma'_1$ will force a play $\rho$ to have a \GoodDec{\lambda} such that $\rho_{l} \in U$ for some $l$. Hence $v \in \WinG{1}{\Obj}{G}$. By Theorem~\ref{thm:WMP} and Lemma~\ref{lem:GDReach} (using arguments as in the proof of this lemma), one checks that the proposed winning strategy $\sigma_1$ is finite-memory and requires a memory in $O(\lambda^2 \cdot W \cdot |V|)$.


Let $v \not\in Y$, that is, $(v,b_v) \not\in Y'$. We define the next strategy $\sigma'_2$ on $\G'$: begin with a winning strategy for the objective $\overline{\GDReach{\lambda}{X'}}$, and if the objective $\GDReach{\lambda}{U'}$ is realized, immediately shift to a winning strategy for the objective $\overline{\DFWMP(\lambda,0)}$ (such a strategy exists by definition of $X'$ and $Y'$). Let $\sigma_2$ be the corresponding strategy in $\G$. Let us show that $\sigma_2$ is a winning strategy of player~$2$ for the objective $\overline{\Obj}$. Assume the contrary: there exists a strategy $\sigma_1$ of player~$1$ such that the outcome $\rho = \Out(v,\sigma_1,\sigma_2)$ is winning for $\Obj$, that is, $\rho$ has a \GoodDec{\lambda} $(k_i)_{i \geq 0}$, that we suppose maximal, with $\rho_l \in U$ for some $l$. Take the smallest index $k_i$ such that $k_i \geq l$;  then the corresponding play $\rho'$ in $\G'$ has the same \GoodDec{\lambda} and $\rho'_{k_i} \in U'$. From this index $k_i$, $\sigma'_2$ will force the play to have no \GoodDec{\lambda}, in contradiction with the \GoodDec{\lambda} $(k_i)_{i \geq 0}$. Finally, as for the strategy proposed above from $v \in Y$, one checks that $\sigma_2$ is finite-memory and requires a memory in $O(\lambda^2 \cdot W \cdot |V|)$.

\qed\end{proof}

The following example illustrates the proof of Proposition~\ref{prop:WMPReach}.

\addtocounter{example}{-1}
\begin{example}[continued]
Consider the game $(G,\Obj)$ depicted on Figure~\ref{ex:windowreach}, where $\Obj = \DFWMP(2,0) \cap \Reach(U)$ with $U = \{v_1,v_3\}$. To compute the winning set $\WinG{1}{\Obj}{G}$, we need to modify $G$ in a new game structure $G'$ where we add a bit to each vertex indicating whether $U$ has been visited (bit equals $1$) or not (bit equals 0). This game structure is the one of previous Figure~\ref{ex:windowICWEnd}, where the set $U' = \{v_1,v_3,v_4\}$ of gray nodes are those with bit~$1$. We already know that every vertex of $G'$ is winning for objective $\GDReach{2}{U'}$, and we note that they are also all winning for objective $\DFWMP(2,0)$. Therefore, in $G'$, player~$1$ can ensure a finite \GoodDec{2} ending in a vertex of $U'$ from which he can ensure an infinite \GoodDec{2}. Thanks to the added bit and coming back to $G$, we get that all vertices of $G$ are winning for~$\Obj$.
\end{example}

\subsection*{Objective $\DFWMP(\lambda,0) \cap \GenReach(U_1,\ldots,U_i)$}

Thanks to the algorithm for the objective $\DFWMP \cap \Reach$ of Proposition~\ref{prop:WMPReach}, we here propose an algorithm 
for the objective $\DFWMP \cap \GenReach$ when $i$ is fixed.

\begin{proposition} \label{prop:WMPGenReach}
Let $\G = \Gdev$ be a $1$-weighted game structure, and $\Obj = \DFWMP(\lambda,0) \cap \GenReach$ $(U_1,\ldots,U_i)$ such that $i$ is fixed.
Then $\WinG{1}{\Obj}{\G}$ can be computed in $O(\lambda \cdot 2^{2i} \cdot |V| \cdot (|V| + |E|) \cdot \lceil \log_2(\lambda \cdot W) \rceil)$ time and finite-memory strategies with memory in $O(\lambda^2 \cdot W \cdot 2^i \cdot |V|)$ are sufficient for both players.
\end{proposition}

The proof of this proposition uses a classical polynomial reduction (since $i$ is fixed) of generalized reachability games to reachability games that we recall for the sake of completeness.

\begin{lemma}\label{lem:GenRtoR}
Each unweighted game $(G,\Obj)$ with $\Obj = \GenReach(U_1,\ldots,U_i)$ for a family $U_1,\ldots, U_i$ of subsets of $V$ with $i$ fixed can be polynomially reduced to an unweighted game $(G',\Obj')$ with $2^i \cdot |V|$ vertices and $2^i \cdot |E|$ edges, and $\Obj' = \Reach(U')$ for some $U' \subseteq V'$. 

\noindent A memoryless (resp. finite-memory) strategy in $G'$ transfers to a finite-memory strategy in $G$. 
\end{lemma}

\begin{proof}
From $G$, we construct a new game $G'$ in a way to keep in a vertex $(v,b_1,\ldots,b_i)$ of $G'$ the current vertex $v$ of $G$ and $i$ bits. For $k \in \{1,\ldots,i\}$, the $k^{th}$ bit indicates whether the set $U_k$ has been visited (bit equals $1$) or not (bit equals $0$). More formally, we define $V' = V \times \{0,1\}^i$ and we have that $((v,b_1,\ldots,b_i),(v',b'_1,\ldots,b'_i)) \in E' \text{ iff }$
\[
(v,v') \in E \text{ and for each } k \in \{1,\ldots,i\},\ 
b'_k =\begin{cases}
1 & \text{if } b_k = 1 \text{ or } v' \in U_k \\
0 &  \text{otherwise}
\end{cases}.\]

The new game $G'$ has $2^i \cdot |V|$ vertices and $2^i \cdot |E|$ edges. We also define $U' = V \times \{1,\ldots,1\}$ and $\Obj' = \Reach(U')$. Let $v_0$ be an initial vertex in $G$ and $b$ the $i$-tuple of bits such that the $k^{th}$ bit is equal to $1$ if $v_0 \in U_k$ and $0$ otherwise. One can check that player~$1$ has a winning strategy from $v_0$ in $(G,\Obj)$ if and only if he has a winning strategy from $(v_0,b)$ in $(G',\Obj')$. Note that the size of $G'$ is polynomial in the size of the original game since $i$ is fixed.

Finally, due to the way $G'$ is constructed from $G$, any strategy in $G'$ that is memoryless (resp. (finite-memory) transfers to a strategy in $G$ that is finite-memory.
\qed\end{proof}

\begin{proof}[of Proposition~\ref{prop:WMPGenReach}]
Let $(\G,\Obj)$ be a 1-weighted game with $\Obj = \DFWMP \cap \GenReach(U_1,\ldots,U_i)$. It is easy to adapt the polynomial reduction of Lemma~\ref{lem:GenRtoR} in a way to get a 1-weighted game $(\G',\Obj')$ with $\Obj' = \DFWMP \cap \Reach(U')$: one just needs to define the weight function $w'$ such that 
$$w'((v,b_1,\ldots,b_i),(v',b'_1,\ldots,b'_i)) = w(v,v')$$ 
(see proof of Lemma~\ref{lem:GenRtoR}). By Proposition~\ref{prop:WMPReach}, the set $\WinG{1}{\Obj'}{\G'}$ can be computed in $O(\lambda \cdot |V'| \cdot (|V'| + |E'|) \cdot \lceil \log_2(\lambda \cdot W) \rceil)$ time and finite-memory strategies with memory in $O(\lambda^2 \cdot W \cdot |V'|)$ are sufficient for both players in $\G'$. Coming back to $\G$, it follows that $\WinG{1}{\Obj}{\G}$ can be computed in $O(\lambda \cdot 2^{2i} \cdot |V| \cdot (|V| + |E|) \cdot \lceil \log_2(\lambda \cdot W) \rceil)$ time and finite-memory strategies with memory in $O(\lambda^2 \cdot W \cdot 2^i \cdot |V|)$ are sufficient for both players.
\qed\end{proof}

\begin{remark} \label{rem:severalsup}
The assumption about $i$ being fixed is necessary. Indeed, by Table~\ref{table:reg} (third row), we know that generalized reachability games are already ${\sf PSPACE}$-hard.
\end{remark}

\subsection{Objective $\DFWMP(\lambda,0) \cap \Buchi(U)$}

We propose in this section an algorithm 
for computing the winning set of player~$1$ for the objective $\DFWMP \cap \Buchi$. This algorithm works as follows: it repeatedly removes the set of vertices that are losing for the objective $\DFWMP \cap \Reach$, as well as its $2$-attractor, until reaching a fixed point. This fixed point is the winning set for the objective $\DFWMP \cap \Buchi$.

\begin{proposition} \label{prop:WMPBuchi}
Let $\G = \Gdev$ be a $1$-weighted game structure, and assume that $\Obj$ is the objective $\DFWMP(\lambda,0) \cap \Buchi(U)$.
Then $\WinG{1}{\Obj}{\G}$ can be computed in $O(\lambda \cdot |V|^2 \cdot (|V| + |E|) \cdot \lceil \log_2(\lambda \cdot W) \rceil)$ time and finite-memory strategies with memory in $O(\lambda^2 \cdot W \cdot |V|)$ (resp. in $O(\lambda^2 \cdot W \cdot |V|^2)$) are sufficient for player~$1$ (resp. player~$2$).
\end{proposition}

\begin{proof}
We propose an algorithm to compute the set of winning vertices $\WinG{1}{\Obj}{G}$ for $\Obj = \DFWMP(\lambda,0) \cap \Buchi(U)$. This algorithm roughly works as follows:
\begin{enumerate}
\item Compute the winning set $X$ for player~$1$ for the objective $\DFWMP(\lambda,0) \cap \Reach(U)$ in $\G$.
\item Compute the 2-attractor $Y$ of the set $V \ssetminus X$ in $\G$.
\item Repeat step 1. in the game $\G[V \ssetminus Y]$  and step 2. in the game $\G$ until $X$ is empty or $X$ is the set of all vertices in $G[V \ssetminus Y]$.
\item The final set $X$ is the set of winning vertices.
\end{enumerate}
This algorithm is detailed in Algorithm~\ref{algo:WMPinterBuchi} (see also Figure~\ref{fig:WMPBuchi}). It makes several calls to two algorithms: Algorithm~${\sf WMPReach}(G,U,\lambda)$ that computes the set of winning vertices of player~$1$ for the objective $\DFWMP(\lambda,0) \cap \Reach(U)$ and Algorithm~${\sf Attractor}(G,2,X)$ that computes the $2$-attractor $\Attr{2}{X}{G}$. To avoid heavy notation, we denote $V \ssetminus X$ by $\overline X$.

\begin{algorithm}
\caption{$\mathsf{WMPBuchi}$}
\label{algo:WMPinterBuchi}
\begin{algorithmic}[1]
\REQUIRE 1-weigthed game structure $G = (V_1,V_2,E,w)$, subset $U \subseteq V$, window size $\lambda \in \N \setminus \{0\}$
\ENSURE $\WinG{1}{\Obj}{G}$ for $\Obj = \DFWMP(\lambda,0) \cap \Buchi(U)$
\STATE $k \leftarrow 0$
\STATE $Y_0 \leftarrow \emptyset$
\STATE $X_0 \leftarrow {\sf WMPReach}(G,U,\lambda)$
\WHILE{$X_k \neq \overline{Y}_k$ and $X_k \neq \emptyset$}
\STATE $Y_{k+1} \leftarrow {\sf Attractor}(\G,2,\overline{X}_k)$
\STATE $X_{k+1} \leftarrow {\sf WMPReach}(G[\overline{Y}_{k+1}],U,\lambda)$
\STATE $k \leftarrow k+1$
\ENDWHILE
\RETURN $X_k$
\end{algorithmic}
\end{algorithm}

Let us prove that Algorithm $\mathsf{WMPBuchi}$ correctly computes the set $\WinG{1}{\Obj}{G}$. Let $k^*$ be the first index such that $X_{k^*} = \emptyset$ or $X_{k^*} = \overline{Y}_{k^*}$. We will show that $(i)$ if $v_0 \in X_{k^*}$ then $v_0 \in \WinG{1}{\Obj}{G}$, and $(ii)$ if $v_0 \not\in X_{k^*}$ 
then $v_0 \in \WinG{2}{\overline{\Obj}}{G}$. We first notice that 
\begin{description}
\item [$(*)$] for each vertex $v$ in $\overline{Y}_{k^*}$, player~$2$ cannot force a play to leave $\overline{Y}_{k^*}$ from $v$ ($Y_{k^*}$ is $2$-closed by Theorem~\ref{thm:attr}). Nevertheless player~$1$ could decide to go to $\overline{Y}_{k^*}$.
\end{description}

$(i)$ Let $v_0 \in X_{k^*}$. This means that $X_{k^*} = \overline{Y}_{k^*} \neq \emptyset$.
Consider the following strategy $\sigma_1$ of player~$1$: Play a winning strategy from $v_0$ for the objective $\DFWMP(\lambda,0) \cap \Reach(U)$ in the game $G[\overline{Y}_{k^*}]$. As soon as a vertex of $U$ has been visited, and the current history has a \GoodDec{\lambda} ending in a vertex $v \in X_{k^*}$ ($v$ belongs to $X_{k^*}$ by $(*)$), play again a winning strategy from $v$ for the objective $\DFWMP(\lambda,0) \cap \Reach(U)$. Continue ad infinitum. 

We claim that $\sigma_1$ is winning for the objective $\DFWMP(\lambda,0) \cap \Buchi(U)$. Indeed, consider any strategy $\sigma_2$ of player $2$ and let $\rho = \OutDev$. By construction of $\sigma_1$, $\rho$ visits infinitely often $U$ and has a \GoodDec{\lambda} since player~$1$ always plays in $X_{k^*}$ and $\sigma_2$ cannot force to leave $X_{k^*}$ by $(*)$.

$(ii)$ Let $v_0 \not\in X_{k^*}$. Then according to $X_{k^*} = \emptyset$ or $X_{k^*} = \overline{Y}_{k^*}$, we have that $v_0 \in V$ or $v_0 \in Y_{k^*}$. 

\noindent $(a)$ Let us begin with case $X_{k^*} = \overline{Y}_{k^*}$. Suppose that $v_0 \in Y_k$ for some $k \in \{1,\ldots,k^*\}$. We construct a strategy $\sigma_2$ from $v_0$ for player~$2$ as follows: 
\begin{enumerate}
\item Play a memoryless strategy from $v_0$ to reach some $v \in \overline{X}_{k-1}$. 
\item If $v \in \overline{X}_{k-1} \cap \overline{Y}_{k-1}$ and as long as the play stays in $G[\overline{Y}_{k-1}]$, play a winning strategy for the objective $\overline{\DFWMP(\lambda,0)\cap\Reach(U)}$ in $G[\overline{Y}_{k-1}]$.
\item Otherwise the plays visits some $v' \in Y_{k-1}$ from which one goes back step 1. with $k$ replaced $k-1$.
\end{enumerate}

Let us show that $\sigma_2$ is winning for player~$2$ for the objective $\overline \Obj$. Let $\sigma_1$ be any strategy of player~$1$ and $\rho = \OutDev$. By construction of $\sigma_2$, there exists $k' \leq k$,  $k \in \{1, \ldots , k^*\}$, such that $\rho$ eventually stays in $G[\overline{Y}_{k'-1}]$, that is, there exists a smallest $i\geq 0$ such that $\rho_i \rho_{i+1} \ldots \in (\overline{Y}_{k'-1})^{\omega}$. Moreover player~$2$ plays a winning strategy from $\rho_i$ for the objective $\overline{\DFWMP(\lambda,0)\cap \Reach(U)}$ which ensures that $\rho_i \rho_{i+1} \ldots$ has a \bad{\lambda} window or never visits $U$. It follows that $\rho$ is losing for the objective $\DFWMP(\lambda,0) \cap \Buchi(U)$ and that $v_0 \in \WinG{2}{\overline{\Obj}}{G}$.

$(b)$ We now consider the case $X_{k^*} = \emptyset$. If $v_0 \in Y_{k^*}$, then we proceed as above. Otherwise, let $v_0 \in  \overline{Y}_{k^*}$. Recall that player~$2$ is winning from each vertex of $\overline{Y}_{k^*}$ for the objective $\overline{\DFWMP(\lambda,0)\cap \Reach(U)}$ in $\G[\overline{Y}_{k^*}]$. Therefore we define a strategy $\sigma_2$ for player~$2$ as above, except that we start at step 2. with $k-1 = k^*$. This strategy is thus also winning and $v_0 \in \WinG{2}{\overline{\Obj}}{G}$.

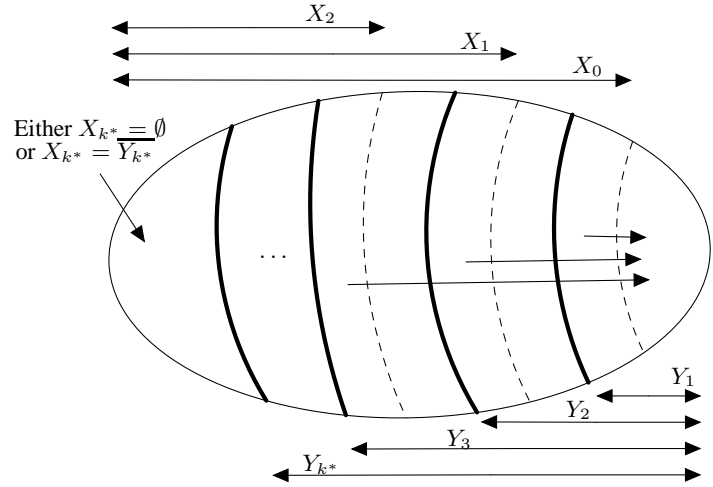
\begin{figure}
\centering
\begin{tikzpicture}[line cap=round,line join=round,>=triangle 45,x=1.0cm,y=1.0cm,scale = 0.7]
\clip(2.5,-4.3) rectangle (16.5,5.);
\draw [rotate around={-178.5442262016258:(10.456896896521387,0.18436530345677968)}] (10.456896896521387,0.18436530345677968) ellipse (5.711608754913006cm and 3.099966598853598cm);
\draw [shift={(19.592544266324655,0.5819420279931938)},dash pattern=on 3pt off 3pt]  plot[domain=2.7972148833184094:3.591189347084698,variable=\t]({1.*5.203790955471708*cos(\t r)+0.*5.203790955471708*sin(\t r)},{0.*5.203790955471708*cos(\t r)+1.*5.203790955471708*sin(\t r)});
\draw [shift={(20.084490005990855,0.6829027534014489)},line width=1.6pt]  plot[domain=2.8218903569695484:3.581169810018563,variable=\t]({1.*6.884511235986853*cos(\t r)+0.*6.884511235986853*sin(\t r)},{0.*6.884511235986853*cos(\t r)+1.*6.884511235986853*sin(\t r)});
\draw [shift={(13.09966611268006,0.7294795709972621)},line width=1.6pt]  plot[domain=2.836178259607482:3.6959668815988214,variable=\t]({1.*6.30727299641972*cos(\t r)+0.*6.30727299641972*sin(\t r)},{0.*6.30727299641972*cos(\t r)+1.*6.30727299641972*sin(\t r)});
\draw [shift={(18.838943526042332,0.48836782687979663)},dash pattern=on 3pt off 3pt]  plot[domain=2.7479035503913716:3.6126480427421845,variable=\t]({1.*6.83985455214304*cos(\t r)+0.*6.83985455214304*sin(\t r)},{0.*6.83985455214304*cos(\t r)+1.*6.83985455214304*sin(\t r)});
\draw [shift={(17.394449298772106,0.6267655780384979)},line width=1.6pt]  plot[domain=2.731723785206424:3.68819510335062,variable=\t]({1.*6.619826538486716*cos(\t r)+0.*6.619826538486716*sin(\t r)},{0.*6.619826538486716*cos(\t r)+1.*6.619826538486716*sin(\t r)});
\draw (15.258012837160935,-1.800124464780219) node[anchor=north west] {$Y_1$};
\draw (13.266771766048581,-2.4254729003361657) node[anchor=north west] {$Y_2$};
\draw [->,line width=0.4pt] (15.,-2.5) -- (16.,-2.5);
\draw [->,line width=0.4pt] (15.,-2.5) -- (14.,-2.5);
\draw [->,line width=0.4pt] (14.,-3.) -- (16.,-3.);
\draw [->,line width=0.4pt] (14.,-3.) -- (11.806438790604469,-2.994106951499533);
\draw (7.441157392711614,0.3392254463322302) node[anchor=north west] {$\ldots$};
\draw [shift={(18.33641405771037,0.7765203895501177)},dash pattern=on 3pt off 3pt]  plot[domain=2.8537309731514444:3.5769717486756067,variable=\t]({1.*8.758064722296412*cos(\t r)+0.*8.758064722296412*sin(\t r)},{0.*8.758064722296412*cos(\t r)+1.*8.758064722296412*sin(\t r)});
\draw [shift={(20.849386520415887,1.1603594317010615)},line width=1.6pt]  plot[domain=2.9823257824930773:3.4758030336046786,variable=\t]({1.*12.280027647729701*cos(\t r)+0.*12.280027647729701*sin(\t r)},{0.*12.280027647729701*cos(\t r)+1.*12.280027647729701*sin(\t r)});
\draw [->,line width=0.4pt] (13.77617508659392,0.5389929511773414) -- (14.962726319005553,0.5129149021133494);
\draw [->,line width=0.4pt] (11.530045353844821,0.0469275322072881) -- (14.868035634035788,0.09908363033527182);
\draw [->,line width=0.4pt] (9.287333134341518,-0.38336027734857764) -- (15.02450392841974,-0.2920871056246061);
\draw [->,line width=0.4pt] (12.5,-3.5) -- (16.,-3.5);
\draw [->,line width=0.4pt] (12.5,-3.5) -- (9.335678428154736,-3.509179863840087);
\draw [->,line width=0.4pt] (12.,-4.) -- (7.8466237858832235,-4.00879688197066);
\draw [->,line width=0.4pt] (12.,-4.) -- (16.,-4.);
\draw (10.995769552713831,-2.98499518478096) node[anchor=north west] {$Y_3$};
\draw (8.263984281601015,-3.462234780336814) node[anchor=north west] {$Y_{k^*}$};
\draw [->,line width=0.4pt] (10.,3.5) -- (14.681179987931106,3.4984260986114544);
\draw [->,line width=0.4pt] (10.,3.5) -- (4.8158005546946345,3.506952787836378);
\draw [->,line width=0.4pt] (9.,4.) -- (12.5,4.);
\draw [->,line width=0.4pt] (9.,4.) -- (4.807273865469711,4.001500762881941);
\draw [->,line width=0.4pt] (7.5,4.5) -- (10.,4.5);
\draw [->,line width=0.4pt] (7.5,4.5) -- (4.790220487019864,4.496048737927504);
\draw [->,line width=0.4pt] (4.580805220245229,1.7195872550026) -- (5.433474142737577,0.4150038035893042);
\draw (2.7839572015975964,2.9393584152227454) node[anchor=north west] {Either $X_{k^*} = \emptyset$};
\draw (2.8168702771531726,2.5608580463336197) node[anchor=north west] {or $X_{k^*} = \overline{Y_{k^*}}$};
\draw (13.381967530493098,4.1) node[anchor=north west] {$X_0$};
\draw (11.259074157158441,4.568555655223764) node[anchor=north west] {$X_1$};
\draw (8.31335389493438,5.1) node[anchor=north west] {$X_2$};
\end{tikzpicture}
\caption{Illustration of Algorithm~\ref{algo:WMPinterBuchi}}
 \label{fig:WMPBuchi}
\end{figure}

It remains to discuss the complexity of the algorithm and the memory of the winning strategies. The while loop of Algorithm $\mathsf{WMPBuchi}$ is executed at most $|V|$ times, and each step in this loop is in $O(\lambda \cdot |V| \cdot (|V| + |E|) \cdot \lceil \log_2(\lambda \cdot W) \rceil)$ time by Theorem~\ref{thm:attr} and Proposition~\ref{prop:WMPReach}. Then, the algorithm runs in  $O(\lambda \cdot |V|^2 \cdot (|V| + |E|) \cdot \lceil \log_2(\lambda \cdot W) \rceil)$ time. As explained in the proof, a winning strategy of player~$1$ consists in repeatedly playing a winning strategy for the objective $\DFWMP(\lambda,0) \cap \Reach(U)$. This strategy is thus finite-memory and one checks that it needs a memory in $O(\lambda^2 \cdot W \cdot |V|)$ (see Proposition~\ref{prop:WMPReach}). Finally, a winning strategy of player~$2$ consists in repeatedly (for decreasing values of $k$) playing a memoryless strategy in $\G$ and then a winning strategy for $\overline{\DFWMP(\lambda,0)\cap\Reach(U)}$ in the subgame $G[\overline{Y}_{k}]$. The latter strategy requires a memory in $O(\lambda^2 \cdot W \cdot |V|)$ by Proposition~\ref{prop:WMPReach}. Therefore the overall required memory is linear in $O(\lambda^2 \cdot W \cdot |V|^2)$.
\qed\end{proof}

\subsection{Objective $\DFWMP(\lambda,0) \cap \GenBuchi(U_1,\ldots,U_i)$}

Thanks to the algorithm for the objective $\DFWMP \cap \Buchi$ given in the previous section, we here propose an algorithm 
for $\DFWMP \cap \GenBuchi$.

\begin{proposition} \label{prop:WMPGenBuchi}
Let $\G = \Gdev$ be a $1$-weighted game structure, and assume that $\Obj$ is the objective $\DFWMP(\lambda,0) \cap \GenBuchi(U_1,\ldots,U_i)$.
Then $\WinG{1}{\Obj}{\G}$ can be computed in $O(\lambda \cdot i^3 \cdot |V|^2 \cdot (|V| + |E|) \cdot \lceil \log_2(\lambda \cdot W) \rceil)$ time and finite-memory strategies with memory in $O(\lambda^2 \cdot W \cdot i \cdot |V|)$ (resp. in $O(\lambda^2 \cdot W \cdot i^2 \cdot |V|^2)$) are sufficient for player~$1$ (resp. player~$2$).
\end{proposition}

The proof of this proposition uses a classical polynomial reduction of generalized B\"uchi games to B\"uchi games that we recall for the sake of completeness.

\begin{lemma}\label{lem:GenBtoB}
Each unweighted game $(G,\Obj)$ with $\Obj = \GenBuchi(U_1,\ldots,U_i)$ for a family $U_1,\ldots, U_i$ of subsets of $V$ can be polynomially reduced to an unweighted game $(G',\Obj')$ with $i \cdot |V|$ vertices and $i \cdot |E|$ edges, and $\Obj' = \Buchi(U')$ for some $U' \subseteq V'$. 

\noindent A memoryless (resp. (polynomial, exponential) memory) strategy in $G'$ transfers to a linear memory (resp. (polynomial, exponential) memory) strategy in $G$. 
\end{lemma}

\begin{proof}
Let us notice that if player~$1$ wins for the objective $\GenBuchi(U_1,\ldots,U_i)$, as he visits every $U_k$ infinitely often, he visits in particular each $U_k$ in the order $1,\ldots, i$ infinitely often. Therefore from $G$, we construct a new game $G'$ in a way to keep in a vertex $(v,k)$ of $G'$ the current vertex $v$ of $G$ and an integer $k$ indicating that we have visited $U_1, \ldots, U_k$. This integer is updated to $k+1$ if we visit $U_{k+1}$, and in the special case of $k = i$, it is updated to $0$. More formally, we define $V' = V \times \{0,\ldots,i\}$ and 
\[
((v,k),(v',k')) \in E' \text{ if } (v,v') \in E \text{ and } 
k' =\begin{cases}
0 & \text{if } k=i \\
k+1 & \text{if } k < i \text{ and } v' \in U_{k+1} \\
k &  \text{otherwise}
\end{cases}.
\]
We also define $U' = V \times \{0\}$ and $\Obj' = \Buchi(U')$. One can check that player~$1$ has a winning strategy from $v_0$ in $(G,\Obj)$ if and only if he has a winning strategy from $(v_0,0)$ in $(G',\Obj')$.

Finally, due to the way $G'$ is constructed from $G$, any strategy in $G'$ that is memoryless (resp. (polynomial, exponential) memory) transfers to a strategy in $G$ that is linear memory (resp. (polynomial, exponential) memory).
\qed\end{proof}

\begin{proof}[of Proposition~\ref{prop:WMPGenBuchi}]
Let $(\G,\Obj)$ be a 1-weighted game with $\Obj = \DFWMP \cap \GenBuchi(U_1,\ldots,U_i)$. It is easy to adapt the polynomial reduction of Lemma~\ref{lem:GenBtoB} in a way to get a 1-weighted game $(\G',\Obj')$ with $\Obj = \DFWMP \cap \Buchi(U')$ (one just needs to define the weight function $w'$ such that $w'((v,k),(v',k')) = w(v,v')$, see proof of Lemma~\ref{lem:GenBtoB}). By Proposition~\ref{prop:WMPBuchi}, the set $\WinG{1}{\Obj'}{\G'}$ can be computed in $O(\lambda \cdot |V'|^2 \cdot (|V'| + |E'|) \cdot \lceil \log_2(\lambda \cdot W) \rceil)$ time and finite-memory strategies with memory in $O(\lambda^2 \cdot W \cdot |V'|)$ (resp. in $O(\lambda^2 \cdot W \cdot |V'|^2)$) are sufficient for player~$1$ (resp. player~$2$) in $\G'$. Coming back to $\G$, it follows that $\WinG{1}{\Obj}{\G}$ can be computed in $O(\lambda \cdot i^3 \cdot |V|^2 \cdot (|V| + |E|) \cdot \lceil \log_2(\lambda \cdot W) \rceil)$ time and finite-memory strategies with memory in $O(\lambda^2 \cdot W \cdot i \cdot |V|)$ (resp. in $O(\lambda^2 \cdot W \cdot i^2 \cdot |V|^2)$) are sufficient for player~$1$ (resp. player~$2$).
\qed\end{proof}

\subsection{Objective $\DFWMP(\lambda,0) \cap \CoBuchi(U)$}

In this section, we study the last class of $\DFWMP \cap \CoBuchi$ games, and we provide an algorithm 
for computing the winning set of player~$1$ for such games. The study is elaborated; we begin with an example where we give the main ideas of the algorithm.

\begin{example} \label{ex:WMPCoBuchi}

Consider the game $(G,\Obj)$ depicted on Figure~\ref{fig:WMPCoBuchi} with $\Obj = \DFWMP(2,0) \cap \CoBuchi(U)$ and $U = \{v_0,v_2,v_4\}$. Player~$1$ has a winning strategy for the objective $\Obj$ if he can force the play to have a \GoodDec{2} and to eventually stay in $U$. We first notice that any vertex winning for $\DFWMP(2,0) \cap \Safe(U)$ is also winning for $\Obj$, that is $\{v_4\} \subseteq \WinG{1}{\Obj}{G}$. 
Then, we notice that any vertex from which player~$1$ can force the play to 
\begin{itemize}
\item either safely stay in $U$ and have a \GoodDec{2},
\item or have a prefix which has a \GoodDec{2} and ends with a vertex that we know winning for~$\Obj$,
\end{itemize}
is also winning for $\Obj$. Therefore, knowing that $v_4$ is winning for $\Obj$, we get that $v_2$ and $v_3$ are also winning. And then knowing that $v_2, v_3, v_4$ are winning, we finally get that $v_1$ and $v_2$ are winning. Hence $\WinG{1}{\DFWMP(2,0)\cap \CoBuchi(U)}{G} = V$.

\begin{figure}[h]
\centering
  \begin{tikzpicture}[scale=4]
    \everymath{\scriptstyle}
    \draw (0,0) node [rectangle, inner sep=5pt, draw,double] (A) {$v_0$};
    \draw (0.5,0) node [circle, draw] (B) {$v_1$};
    \draw (1,0) node [rectangle, inner sep=5pt, draw,double] (C) {$v_2$};
    \draw (1.5,0) node [circle, draw] (D) {$v_3$};
    \draw (2,0) node [circle, draw,double] (E) {$v_4$};

    \draw[->,>=latex] (A) to node[above,midway] {$-1$} (B);
    \draw[->,>=latex] (B) to node[above,midway] {$1$} (C);
    \draw[->,>=latex] (C) to node[above,midway] {$-1$} (D);
    \draw[->,>=latex] (D) to node[above,midway] {$1$} (E);
    
    \draw[->,>=latex] (A) .. controls +(45:0.4cm) and +(135:0.4cm) .. (A) node[above,midway] {$0$};
    \draw[->,>=latex] (E) .. controls +(45:0.4cm) and +(135:0.4cm) .. (E) node[above,midway] {$0$};
    \draw[->,>=latex] (C) .. controls +(45:0.4cm) and +(135:0.4cm) .. (C) node[above,midway] {$0$};
    
    \end{tikzpicture}
    \caption{A game with the objective $\Obj = \DFWMP(2,0) \cap \CoBuchi(U)$}
\label{fig:WMPCoBuchi}
\end{figure}
\end{example}

To formally give the algorithm and its correctness, we now need to define two specific objectives:

\begin{definition} \label{def:Objg}
Let $\Gdev$ be a $1$-weighted game structure, $U,X,Z \subseteq V$ be three sets of vertices and $\lambda \in \N \setminus \{0\}$ be a window size. We consider the next sets of plays:
\begin{eqnarray*} \label{eq:Objg}
\Objg(U,X,Z) &=& \{ \rho \in \Plays(\G) \mid \exists l \in \{1,\ldots, \lambda\} \mbox{ such that the \window{0} is} \\ 
&&\hspace{1mm} \mbox{\iclosed{\lambda}{l} and either } (\rho_l \in \WinG{1}{\GDReach{\lambda}{X}}{G}) \mbox{ or }  \\
&&\hspace{1mm} (\rho_l \in Z \mbox{ and } \rhofactor{0}{l} \in U^+) \}, \\
\Objf(U,X) &=& (\DFWMP(\lambda,0) \cap \Safe(U)) \cup \GDReach{\lambda}{X}.
\end{eqnarray*}
\end{definition}

\noindent 
Objective $\Objg(U,X,Z)$ asks to inductively-close a window in a vertex in $\GDReach{\lambda}{X}$, or in $Z$ while staying in $U$ (see Figure~\ref{fig:Objg}). Notice that this objective is a variant of objective $\ICWReach{\lambda}{U}$. The second objective $\Objf(U,X)$ asks for a play to be either in $\DFWMP(\lambda,0) \cap \Safe(U)$ or in $\GDReach{\lambda}{X}$.

\begin{figure}
\centering
\begin{tikzpicture}[scale=3]

	\draw (-0.9,0.5) node[] (UU) {Either};
	\draw (0,0.5) node[circle,draw] (ZZ) {$\rho_0$};
	
	\draw (0.3,0.5) node[circle,draw,inner sep = 1pt] (BB) {};
	\draw (0.6,0.5) node[circle,draw,inner sep = 1pt] (CC) {};
	\draw (0.9,0.5) node[circle,draw,inner sep = 1pt] (DD) {};
	\draw (1.2,0.5) node[circle,draw,inner sep = 1pt] (EE) {};
	
	\draw (0,0.32) node[] (AB) {$\in  U$};
	\draw (0.3,0.4) node[] (AC) {$\in U$};
	\draw (0.6,0.4) node[] (AD) {$\in U$};
	\draw (0.9,0.4) node[] (AE) {$\in U$};
	\draw (1.2,0.4) node[] (AF) {$\in U$};
	
	\draw (1.5,0.5) node[circle,draw] (YY) {$\rho_l$};
	\draw (1.8,0.4) node[] (XX) {$\in Z \cap U$};
	\draw (1.9,0.5) node (AA) {$\ldots$};
	\draw (ZZ) to[bend left=20] (YY);
	
	\draw (ZZ) -- (BB);
	\draw (BB) -- (CC);
	\draw (CC) -- (DD);
	\draw (DD) -- (EE);
	\draw (EE) -- (YY);
	
	\draw (YY) -- (AA);

	\draw (-0.9,0) node[] (U) {or};
	\draw (0,0) node[circle,draw] (Z) {$\rho_0$};
	\draw (1.5,0) node[circle,draw] (Y) {$\rho_l$};
	\draw (2,0.14) node[] (X) {$\in \WinG{1}{\GDReach{\lambda}{X}}{G}$};
	\draw (1.9,0) node (A) {$\ldots$};
	\draw (Z) to[bend left=20] (Y);
	\draw (Z) -- (Y);
	\draw (Y) -- (A);

\end{tikzpicture}

\caption{A play that belongs to $\Objg(U,X,Z)$}
\label{fig:Objg}
\end{figure}
 
The algorithm to solve $\DFWMP \cap \CoBuchi$ games roughly works as follows: compute the least fixed point (on $X$) of the winning set of player~$1$ for the objective $\Objf(U,X)$. We will show that this winning set is obtained by computing the greatest fixed point (on $Z$) of the winning set of player~$1$ for the objective $\Objg(U,X,Z)$.

We begin with a lemma showing that the winning set of player~$1$ for the objective $\Objg(U,X,Z)$ can be computed in in time polynomial in the size of the game, $\lambda$ and $\lceil \log_2(W) \rceil$. Notice that the proof of this lemma uses arguments similar to those used for Proposition~\ref{prop:WMPReach}.

\begin{lemma} \label{lem:Objg}
Let $\Obj$ be the objective $\Objg(U,X,Z)$, for some subsets $U,X,Z$ of $V$.
Then $\WinG{1}{\Obj}{\G}$ can be computed in $O(\lambda \cdot (|V| + |E|) \cdot \lceil \log_2(\lambda \cdot W) \rceil)$ time, and finite-memory strategies with memory in $O(\lambda)$ are sufficient for both players.
\end{lemma}

\begin{proof}
From $\G$ and $U$, we construct a new game structure $\G' = (V'_1,V'_2,E',w')$ with a bit added to the vertices of $V$ that indicates whether $V \setminus U$ has been already visited (the bit equals $1$) or not (the bit equals $0$). We denote $V \setminus U$ by $\overline{U}$. More precisely, we have $V'_1 = V_1 \times \{0,1\}$, $V'_2 = V \times \{0,1\}$; given $(v,v') \in E$, an edge $((v,b),(v',b'))$ belongs to $E'$ such that $b'=1$ if ($b=1$ or $v' \in \overline{U}$), and $b' = 0$ otherwise, and $w'((v,b),(v',b')) = w(v,v')$.

For each $v \in V$, we denote by $(v,b_v)$ the vertex of $V'$ such that $b_v = 1$ if $v \in \overline{U}$, and $b_v = 0$ otherwise. 
\begin{description}
\item[$(*)$] Notice that to each play $\rho$ of $\G$ from vertex $v$ corresponds a unique play $\rho'$ of $\G'$ from vertex $(v,b_v)$. Given a play $\rho$ in $G$, if the \window{0} is \iclosed{\lambda}{l} for some $l \in \{1,\ldots, \lambda\}$, then it is also the case for $\rho'$ in $G'$. Moreover, by construction of $G'$, $\rho'_l = (v',0)$ iff $\rhofactor{0}{l} \in U^+$. Finally, a strategy on $\G$ can be easily translated to a strategy on $\G'$. For each of the above statements, the converse is also true. 
\end{description}

The winning set of player~$1$ for the objective $\Obj$ is computed as follows. Let $U' = \{ (v,b) \in V' \mid v \in \WinG{1}{\GDReach{\lambda}{X}}{G} \mbox{ or } (b=0 \mbox{ and } v \in Z) \}$. In $\G'$, we compute the set $X' = \ICWReach{\lambda}{U'}$ and we let $X = \{v \in V \mid (v,b_v) \in X'\}$. Let us show that $X = \WinG{1}{\Obj}{G}$. 
We have that $v \in X$ iff player~$1$ has a winning strategy for $\ICWReach{\lambda}{U'}$ in $G'$ from $(v,b_v)$, i.e. he can ensure that the \window{0} in the play $\rho'$ is \iclosed{\lambda}{l} with $l \in \{1,\ldots,\lambda\}$ and $\rho'_l = (v',b') \in U'$. By $(*)$ and definition of $U'$, it is equivalent to say that the \window{0} in the corresponding play $\rho$ in $G$ is \iclosed{\lambda}{l} and either $\rho_l = v' \in \WinG{1}{\GDReach{\lambda}{X}}{G}$ or ($\rhofactor{0}{l} \in U^+$ and $\rho_l \in Z$). It follows that $v \in X$ iff $v \in \WinG{1}{\Obj}{G}$. Finally, the announced results about the algorithmic complexity and the memory requirements of the strategies follow from Lemma~\ref{lem:ICWReach}.
\qed\end{proof}

We now give an algorithm in time polynomial in the size of the game, $\lambda$ and $\lceil \log_2(W) \rceil$) to compute the winning set of player~$1$ for the objective $\Objf(U,X)$ (see Algorithm~\ref{algo:Objf}). This algorithm uses, through a greatest fixpoint, the algorithm $\SolveObjg(G,\lambda,U,X,Z)$ which returns the winning set of player~$1$ for the objective $\Objg(U,X,Z)$ (see Lemma~\ref{lem:Objg}). 

\begin{algorithm}
\caption{$\SolveObjf$}
\label{algo:Objf}
\begin{algorithmic}[1]
\REQUIRE 1-weighted game structure $G = (V_1,V_2,E,w)$, subsets $U,X \subseteq V$, window size $\lambda \in \N \setminus \{0\}$
\ENSURE $\WinG{1}{\Obj}{G}$ for $\Obj = \Objf(U,X)$
\STATE $k \leftarrow 0$
\STATE $Z_0 \leftarrow V$
\REPEAT
\STATE $T_k \leftarrow \SolveObjg(G,\lambda,U,X,Z_k)$
\STATE $Z_{k+1} \leftarrow X \cup T_k$
\STATE $k \leftarrow k+1$
\UNTIL {$Z_k = Z_{k-1}$}
\RETURN $Z_k$
\end{algorithmic}
\end{algorithm}

\begin{lemma} \label{lem:Objf}
Let $\Obj$ be the objective $\Objf(U,X)$, for some subsets $U,X$ of $V$.
Then $\WinG{1}{\Obj}{\G}$ can be computed in $O(\lambda \cdot |V| \cdot (|V| + |E|) \cdot \lceil \log_2(\lambda \cdot W) \rceil)$ time, and finite-memory strategies with memory in $O(\lambda^2 \cdot W \cdot |V|)$ are sufficient for both players.
\end{lemma}

\begin{proof}
Let $Z^* = \cap_{k \geq 0} Z_k$ and $T^* = \cap_{k \geq 1} T_k$. We have $Z^* = X \cup T^*$ and $T^* = \Objg(U,X,Z^*)$. We will show that $(i)$ if $v_0 \in Z^*$ then $v_0 \in \WinG{1}{\Obj}{G}$ and $(ii)$ if $v_0 \not\in Z^*$ then $v_0 \in \WinG{2}{\overline{\Obj}}{G}$.

$(i)$ Let $v_0 \in Z^*$. If $v_0 \in X$ then obviously $v_0$ is winning for the objective $\GDReach{\lambda}{X}$, and thus winning for $\Obj$. Otherwise, $v_0 \in T^*$, and we consider the next strategy $\sigma_1$ of player~$1$.
\begin{enumerate}
\item Play a winning strategy for the objective $\Objg(U,X,Z^*)$.
\item As soon as this objective is achieved, say in vertex $v$, if $v \in \WinG{1}{\GDReach{\lambda}{X}}{G}$ then play a winning strategy for $\GDReach{\lambda}{X}$, otherwise (notice that $v \in Z^*$) go back to step 1.
\end{enumerate}
Let us show that $\sigma_1$ is a winning strategy of player~$1$ in $(G,\Obj)$. Let $\rho = \OutDev$ with $\sigma_2$ being any strategy of player~$2$. By construction of $\sigma_1$, either there exists $l$ such that $\rhofactor{0}{l}$ has a \GoodDec{\lambda} with $\rho_l \in X$, which shows that $\rho \in \GDReach{\lambda}{X} \subseteq \Obj$; or we infinitely go back to step 1, in which case $\rho$ has a \GoodDec{\lambda} and only visits vertices of $U$. In the latter case, it follows that $\rho \in \DFWMP(\lambda,0) \cap \Safe(U) \subseteq \Obj$. One checks that $\sigma_1$ is a finite-memory strategy with memory in $O(\lambda^2 \cdot W \cdot |V|)$ by Lemmas~\ref{lem:GDReach} and~\ref{lem:Objg}. 

$(ii)$ We are going to prove by induction on $k\geq1$ that if $v_0 \not\in Z_k$ then $v_0 \in \WinG{2}{\overline{\Obj}}{G}$. Let $k = 1$ and suppose that $v_0 \not\in Z_1$. Thus $v_0 \not\in X$ and $v_0 \in \WinG{2}{\overline{\Objg(U,X,Z_0)}}{G}$ with $Z_0 = V$. Notice that a winning strategy of player~$2$ for $\overline{\Objg(U,X,Z_0)}$ forces a play that begin with a window  \fclosed{\lambda}{l \in \{1,\ldots,\lambda\}}, to visit $V \setminus U$ in a position $l' \leq l$ and to visit $\WinG{2}{\overline{\GDReach{\lambda}{X}}}{G}$ at position $l$.
We define the next strategy $\sigma_2$ of player~$2$: play a winning strategy for $\overline{\Objg(U,X,V)}$, and if the \window{0} is \fclosed{\lambda}{l \in \{1,\ldots,\lambda\}}, with $\rho_l \in \WinG{2}{\overline{\GDReach{\lambda}{X}}}{G}$, switch to a winning strategy for $\overline{\GDReach{\lambda}{X}}$. 

Let us show that $\sigma_2$ is winning for the objective $\overline{\Obj}$. Consider $\rho = \OutDev$ with $\sigma_1$ a strategy of player~$1$. If the \window{0} is \bad{\lambda}, then $\rho \not\in \DFWMP(\lambda,0)$ and thus $\rho \in \overline{\Obj}$. Otherwise this window is \fclosed{\lambda}{l \in \{1,\ldots,\lambda\}}. Then by construction of $\sigma_2$, $\rho$ visits $V \setminus U$ and from $\rho_l$, the play belongs to $\overline{\GDReach{\lambda}{X}}$. As $\rho_0 \not\in X$, it follows that $\rho \not\in \Safe(U) \cap \GDReach{\lambda}{X}$, and thus $\rho \in \overline{\Obj}$.

Suppose now that $k \geq 2$, and take $v_0 \not\in Z_k$. Then, $v_0 \not\in X$ and $v_0 \in \WinG{2}{\overline{\Objg(U,X,Z_{k-1})}}{G}$. Notice that a winning strategy of player~$2$ for $\overline{\Objg(U,X,Z_{k-1})}$ forces a play $\rho$ that begin with a window \fclosed{\lambda}{l \in \{1,\ldots,\lambda\}}, to visit $\WinG{2}{\overline{\GDReach{\lambda}{X}}}{G}$ at position $l$ and also to visit $V \setminus Z_{k-1}$ at position $l$ if $\rhofactor{0}{l} \in U^+$. We adapt the definition of $\sigma_2$ given in case $k=1$ as follows: play a winning strategy for $\overline{\Objg(U,X,Z_{k-1})}$, and if the \window{0} is \fclosed{\lambda}{l \in \{1,\ldots,\lambda\}},
\begin{itemize}
\item if the current history $\rhofactor{0}{l}$ is in $U^+$, switch to a winning strategy for $\overline{\Obj}$ (which exists by induction hypothesis as $\rho_l \not\in Z_{k-1}$), 
\item otherwise switch to a winning strategy for $\overline{\GDReach{\lambda}{X}}$.
\end{itemize}

This strategy $\sigma_2$ is winning for $\overline{\Obj}$. Indeed, the arguments are the same as for case $k=1$ except that we have to consider the additional situation where $\rho$ has a \window{0} \fclosed{\lambda}{l \in \{1,\ldots,\lambda\}} such that $\rhofactor{0}{l}$ is in $U^+$. Let $\rho = \rhofactor{0}{l-1}\rho'$. By construction of $\sigma_2$, we have that $\rho' \in \overline{\Obj}$, and therefore $\rho \in \overline{\Obj}$ (since $\rho_0 \not\in X$). 

One checks that the proposed winning strategy $\sigma_2$ is a finite-memory strategy with memory in $O(\lambda^2 \cdot W \cdot |V|)$ by Lemmas~\ref{lem:GDReach} and~\ref{lem:Objg}.

It remains to study the time complexity. Algorithm~\ref{algo:Objf} runs in $O(\lambda \cdot |V|\cdot (|V| + |E|) \cdot \lceil \log_2(\lambda \cdot W) \rceil)$ time since its loop is executed at most $|V|$ times, and each step in this loop is in $O(\lambda \cdot (|V| + |E|) \cdot \lceil \log_2(\lambda \cdot W) \rceil)$ time by Lemma~\ref{lem:Objg}. 
\qed\end{proof}

We are now able to give an algorithm in time polynomial in the size of the game, $\lambda$ and $\lceil \log_2(W) \rceil$ to solve $\DFWMP \cap \CoBuchi$ games. As expected, the algorithm (Algorithm~\ref{algo:WMPinterCoBuchi}) is the least fixed point (on $X$) of the algorithm $\SolveObjf(G,U,X,\lambda)$.

\begin{algorithm}
\caption{$\mathsf{WMPCoBuchi}$}
\label{algo:WMPinterCoBuchi}
\begin{algorithmic}[1]
\REQUIRE 1-weighted game structure $G = (V_1,V_2,E,w)$, subset $U \subseteq V$, window size $\lambda \in \N \setminus \{0\}$
\ENSURE $\WinG{1}{\Obj}{G}$ for $\Obj = \DFWMP(\lambda,0) \cap \CoBuchi(U)$
\STATE $k \leftarrow 0$
\STATE $X_0 \leftarrow \emptyset$
\REPEAT
\STATE $X_{k+1} \leftarrow X_k \cup \SolveObjf(G,U,X_k,\lambda)$
\STATE $k \leftarrow k+1$
\UNTIL {$X_k = X_{k-1}$}
\RETURN $X_k$
\end{algorithmic}
\end{algorithm}

\begin{proposition} \label{prop:WMPcoBuchi}
Let $\G = \Gdev$ be a $1$-weighted game structure, and suppose that $\Obj$ is the objective $\DFWMP(\lambda,0) \cap \CoBuchi(U)$ for some $U \subseteq V$ and a window size $\lambda \in \N \setminus \{0\}$.
Then $\WinG{1}{\Obj}{\G}$ can be computed in $O(\lambda \cdot |V|^2 \cdot (|V| + |E|) \cdot \lceil \log_2(\lambda \cdot W) \rceil)$ time and strategies with memory in $O(\lambda^2 \cdot W \cdot |V|^2)$ (resp. in $O(\lambda^2 \cdot W \cdot |V|)$) are sufficient for player~$1$ (resp. player~$2$).
\end{proposition}

\begin{proof}
Let $X^* = \cup_{k \geq 1}X_k$. We will show that $(i)$ if $v_0 \in X^{*}$ then $v_0 \in \WinG{1}{\Obj}{G}$ and $(ii)$ if $v_0 \not\in X^{*}$ then $v_0 \in \WinG{2}{\overline{\Obj}}{G}$.

We prove $(i)$ by induction on $k$. Let $k = 1$ and suppose that $v_0 \in X_1$, then $v_0 \in \WinG{1}{\Objf(U,\emptyset)}{G}$, i.e. $v_0$ is winning for player~$1$ for the objective $\DFWMP(\lambda,0) \cap \Safe(U)$ (see Definition~\ref{def:Objg}). Therefore, a winning strategy $\sigma_1$ for this objective is clearly winning  for $\Obj$. Suppose now that $k > 1$ and $v_0 \in X_k \setminus X_{k-1}$, then $v_0 \in \WinG{1}{\Objf(U,X_{k-1})}{G}$. Consider the following strategy $\sigma_1$ of player~$1$:
play a winning strategy from $v_0$ for $\Objf(U,X_{k-1})$ and as soon as the current history has a \GoodDec{\lambda} ending in a vertex in $X_{k-1}$, switch to a winning strategy for $\Obj$ (given by the induction hypothesis). This strategy $\sigma_1$ is winning for $\Obj$. Indeed it ensures that either the play $\rho \in \DFWMP(\lambda,0) \cap \Safe(U)$ or it decomposes as $hv\rho'$ such that $hv$ has a \GoodDec{\lambda} ending in $v \in X_{k-1}$ and $v\rho'$ is winning for $\DFWMP \cap \CoBuchi$. It follows that $\rho$ is itself winning for $\DFWMP \cap \CoBuchi$. One checks that $\sigma_1$ is a finite-memory strategy with memory in $O(\lambda^2 \cdot W \cdot |V|^2)$.


$(ii)$ Assume that $v_0 \not\in X^{*}$, then $v_0 \in \WinG{2}{\overline{\Objf(U,X^{*})}}{G}$. by Definition~\ref{def:Objg}, we have \begin{eqnarray} \label{eq:neg} 
\overline{\Objf(U,X^{*})} = (\overline{\DFWMP(\lambda,0)} \cup \overline{\Safe(U)}) \cap \overline{\GDReach{\lambda}{X^*}}
\end{eqnarray}
Consider the following strategy $\sigma_2$ of player~$2$. 
\begin{enumerate}
\item If the current vertex $v$ belongs to $V \setminus X^{*}$, play a winning strategy for $\overline{\Objf(U,X^{*})}$.
\item As soon as the current history starting from $v$ visits $V \setminus U$ and has a \GoodDec{\lambda}, go back to step 1.
\end{enumerate}
Notice that the \GoodDec{\lambda} of second step necessarily ends in a vertex of $V \setminus X^{*}$ by $(\ref{eq:neg})$.

Let $\sigma_1$ be any strategy of player~$1$ and let $\rho = \OutDev$. We will show that $\rho \not\in \Obj$. If $\rho \not\in \DFWMP(\lambda,0)$ then $\rho \not\in \Obj$. Otherwise, $\rho$ has a \GoodDec{\lambda} $(k_i)_{i \geq 0}$ that we can suppose maximal. By $(\ref{eq:neg})$ and definition of $\sigma_2$, $\rho$ has a first visit of $V \setminus U$ in some position $l_1$ such that for the smallest $k_{i_1} \geq l_1$ we switch from step 2. to step 1. Thus $\rho$ has a second visit of $V \setminus U$ in $l_2$ such that for the smallest $k_{i_2} \geq l_2$ we switch from step 2. to step 1., ad infinitum. Therefore, $\rho$ visits $V \setminus U$ infinitely often, that is, $\rho \not\in \CoBuchi(U)$, and thus $\rho \not\in \Obj$. One checks that $\sigma_2$ is a finite-memory strategy with memory in $O(\lambda^2 \cdot W \cdot |V|)$.

Algorithm~\ref{algo:WMPinterCoBuchi} runs in $O(\lambda \cdot |V|^2 \cdot (|V| + |E|) \cdot \lceil \log_2(\lambda \cdot W) \rceil)$ time since its loop is executed at most $|V|$ times, and each step in this loop is in $O(\lambda \cdot |V| \cdot (|V| + |E|) \cdot \lceil \log_2(\lambda \cdot W) \rceil)$ time by Lemma~\ref{lem:Objf}.
\qed\end{proof}

Let us come back to Example~\ref{ex:WMPCoBuchi} to illustrate the two nested fixed points.

\addtocounter{example}{-1}
\begin{example}[continued]
One can check that Algorithm $\mathsf{WMPCoBuchi}$ computes the following sets $X_k$ (based on the computation of sets $Z_{k,l}$ by Algorithm $\SolveObjf$):
\begin{center}
\begin{tabular}{ l l } 
  $X_0 = \emptyset$ &  \\ 
   & $Z_{0,0} = V$ \\ 
   & $Z_{0,1} = \{v_4\}$ \\ 
   & $Z_{0,2} = \{v_4\}$ \\
 \end{tabular} \hspace{0.5cm}
\begin{tabular}{ l l } 
  $X_1 = \{v_4\}$ &  \\ 
   & $Z_{1,0} = V$ \\ 
   & $Z_{1,1} = \{v_2,v_3,v_4\}$ \\ 
   & $Z_{1,2} = \{v_2,v_3,v_4\}$ \\
 \end{tabular} \hspace{0.5cm}
\begin{tabular}{ l l } 
  $X_2 = \{v_2,v_3,v_4\}$ &  \\ 
   & $Z_{2,0} = V$ \\ 
   & $Z_{2,1} = V$ \\ 
   \
 \end{tabular}
\end{center}

and finally $X_3 = X_4 = V$.

\end{example}

\subsection{Proof of Theorem~\ref{thm:2obj}}

We are now able to prove the main theorem of Section~\ref{sec:OneWindow}.

\begin{proof}[of Theorem~\ref{thm:2obj}]
Let $(G, \Obj)$ be a game with objective $\Obj = \Obj_1 \cap \Gamma$  for player~$1$ such that $\Obj_1 = \DFWMP$ and $\Gamma = \cap_{m=2}^{\n} \Obj_m$ such that $\forall m,\Obj_m = \Inf$ (resp. $\Obj_m = \LimInf$, $\Obj_m = \LimSup$, $\{\forall m \Obj_m = \Sup$ and $n$ is fixed $\}$). Let us first provide an algorithm in time polynomial in the size of the game, $\lambda$ and $\lceil \log_2(W) \rceil$) for \prob. We reduce this game to a game $(G', \Obj'_1 \cap \Gamma')$ (with  $\Gamma' = \cap_{m=2}^{\n} \Obj'_m$) such that $\Obj'_1 = \DFWMP$ and $\forall m, \Obj'_m \in \Reg$ thanks to the polynomial reduction of Proposition~\ref{prop:WISLtoWreg}.
Recall that the intersection of safety (resp. co-B\"uchi) objectives is a safety (resp. co-B\"uchi) objective. We thus have to study six cases of games: $\DFWMP \cap \Safe$ game (resp. $\DFWMP \cap \Reach$, $\DFWMP \cap \GenReach$ (with $n$ fixed), $\DFWMP \cap \Buchi$, $\DFWMP \cap \GenBuchi$, $\DFWMP \cap \CoBuchi$ game). The complexity and the memory requirements of each such game have been treated in Propositions~\ref{prop:WMPSafe}-\ref{prop:WMPcoBuchi} respectively (see also Table~\ref{table:poly}). 

By using Proposition~\ref{prop:WISLtoWreg} to come back to $(G, \Obj)$, it follows that the \prob\ for the game $(G, \Obj)$ can be solved in time polynomial in the size of the game, $\lambda$ and $\lceil \log_2(W) \rceil$ and that exponential memory strategies are sufficient for both players. More precisely, the algorithmic complexity and the memory requirements are those of Table~\ref{table:poly} where $|V|$ is replaced by $|V|+|E|$.

The \prob\ is ${\sf P}$-hard because is it is already ${\sf P}$-hard for games $(\G, \Obj_1) = (\G, \Obj_1 \cap \Obj_2)$ with $\Obj_1 = \DFWMP$ and $\Obj_2 = \Inf$ and the weights on the second dimension all equal to $0$ (see Theorem~\ref{thm:WMP}). Moreover, by Theorem~\ref{thm:WMP}, finite-memory strategies are necessary for both players.
\qed\end{proof}

\section{Intersection of objectives in $\ISL$} \label{sec:regular}

The aim of this section is to provide a refinement of Theorem~\ref{thm:general} for games $(\G, \Obj = \cap_{m=1}^{\n} \Obj_m)$ when no objective $\Obj_m$ is a $\DFWMP$ objective. In this case, we get the better complexity of ${\sf PSPACE}$-completeness  (instead of ${\sf EXPTIME}$-completeness) for the \prob ; nevertheless the two players still need exponential memory strategies to win (Theorem~\ref{thm:ISL}). We also study with precision the complexity and the memory requirements in terms of the objectives of $\ISL$ that appear in the intersection $\Obj = \cap_{m=1}^{\n} \Obj_m$ (Corollary~\ref{cor:ISL} below). Notice that the membership to ${\sf PSPACE}$ in Theorem~\ref{thm:ISL} could have been obtained from one result proved in \cite{AlurTM03} (in this paper, the objective is defined by an $\mathsf{LTL}$ formula, and it is proved that deciding the winner is in ${\sf PSPACE}$ when the objective is a Boolean combination of formulas of the form ``eventually $p$'' and ``infinitely often $p$''). We here propose a simple proof adapted to our context, that allows an easy study of the winning strategies as well as the identification of polynomial fragments. The proof roughly works as follows. We polynomially reduce the game $G$ to a game $G'$ by Proposition~\ref{prop:WISLtoWreg}, where the objective $\Obj'$ is the intersection of a generalized-reachability, a generalized-B\"uchi and a co-B\"uchi objective. To solve this new game, we first compute in $\sf PTIME$ the winning set $X_1$ for the generalized-B\"uchi $\cap$ co-B\"uchi objective. Then, we solve in $\sf PSPACE$ the generalized-reachability game in the subgame induced by $X_1$ (see Figure~\ref{fig:3steps}).

%

\begin{theorem} \label{thm:ISL}
Let $(\G, \Obj)$ be an $\n$-weighted game such that $\Obj = \cap_{m=1}^{\n} \Obj_m$ with $\Obj_m \in \ISL$ for all $m$. Then the \prob\ is ${\sf PSPACE}$-complete (with an algorithm in $O(2^n \cdot (|V|+|E|))$ time) and exponential memory strategies are both necessary and  sufficient for both players. 
\end{theorem}

\begin{proof}
Let us first prove that the \prob\ is in ${\sf PSPACE}$ and that both players have exponential memory winning strategies. By Proposition~\ref{prop:WISLtoWreg}, the game $(\G, \Obj)$ is polynomially equivalent to a game $(\G', \Obj')$ such that $\Obj' = \cap_{m=1}^{\n} \Obj'_m$ with $\Obj'_m \in \Reg$ for all $m$. We thus have to  prove that one can decide in ${\sf PSPACE}$ whether player~$1$ has a winning strategy for $\Obj'$ in $\G'$ from a given vertex $v_0$, and that both players have exponential memory winning strategies in $G'$. This proof needs several steps that are illustrated in Figure~\ref{fig:3steps}; to avoid heavy notation, we abusively rename $(\G', \Obj')$ as $(\G, \Obj)$.

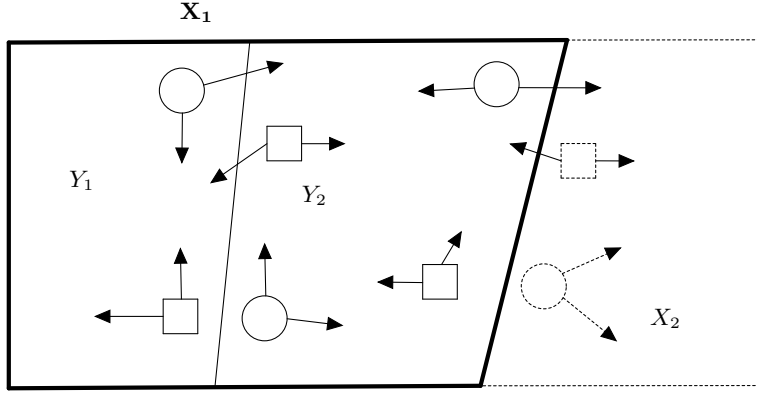
\begin{figure}
\centering
\begin{tikzpicture}[line cap=round,line join=round,>=triangle 45,x=1.0cm,y=1.0cm,scale=2.3]
\clip(1.,-1.2) rectangle (6.,1.3);
\draw [line width=1.6pt] (4.034343502303791,-1.0014394386997696)-- (4.534343502303791,0.9985605613002306);
\draw [line width=1.6pt] (4.534343502303791,0.9985605613002306)-- (1.3076679251120324,0.9841750653629551);
\draw [line width=1.6pt] (1.3076679251120324,0.9841750653629551)-- (1.3076679251120324,-1.0158249346370454);
\draw [line width=1.6pt] (1.3076679251120324,-1.0158249346370454)-- (4.034343502303791,-1.0014394386997696);
\draw [dash pattern=on 1pt off 1pt] (4.534343502303791,0.9985605613002306)-- (5.683519089379931,1.);
\draw [dash pattern=on 1pt off 1pt] (5.683519089379931,1.)-- (5.683519089379931,-1.);
\draw [dash pattern=on 1pt off 1pt] (5.683519089379931,-1.)-- (4.034343502303791,-1.0014394386997696);
\draw [dash pattern=on 1pt off 1pt] (4.034343502303791,-1.0014394386997696)-- (4.534343502303791,0.9985605613002306);
\draw (2.7011940265172916,0.9946550337572669)-- (2.5,-1.);
\draw [dash pattern=on 1pt off 1pt] (4.399803040719828,-0.42609509849238486) circle (0.13cm);
\draw(2.7863657625393405,-0.6111833139272619) circle (0.13cm);
\draw(2.30799749545189,0.7063109109759723) circle (0.13cm);
\draw(4.127186779181149,0.7430092538754098) circle (0.13cm);
\draw [->] (2.4340224864550573,0.7382123473488329) -- (2.8979821450266607,0.8655872851357469);
\draw [->] (2.311167284471002,0.5763495611923758) -- (2.30799749545189,0.28165865742533835);
\draw [->] (2.7931381409686695,-0.4813598380440074) -- (2.7863657625393405,-0.1733036595539279);
\draw [->] (2.9163441288688876,-0.6135548717068768) -- (3.2423813012785176,-0.6532347589280368);
\draw [->] (4.255675099875729,0.7232418199223974) -- (4.7353307472289705,0.7220387722185884);
\draw [->] (3.9991328184509074,0.7205998107476177) -- (3.6707109468619565,0.7024556491927401);
\draw [->,dash pattern=on 1pt off 1pt] (4.506928490865746,-0.3524463515170664) -- (4.850668396341488,-0.20066242068155452);
\draw [->,dash pattern=on 1pt off 1pt] (4.5118544886903384,-0.49200771494562523) -- (4.824455294270462,-0.7458949437589116);
\draw (1.6072672143433187,0.29117552106652417) node[anchor=north west] {$Y_1$};
\draw (2.9523110844786107,0.20486254544286897) node[anchor=north west] {$Y_2$};
\draw (4.966280515697229,-0.5072195034522862) node[anchor=north west] {$X_2$};
\draw (2.247421783552094,1.2628949683935758) node[anchor=north west] {$\mathbf{X_1}$};
\draw (2.2,-0.7)-- (2.4,-0.7);
\draw (2.4,-0.7)-- (2.4,-0.5);
\draw (2.4,-0.5)-- (2.2,-0.5);
\draw (2.2,-0.5)-- (2.2,-0.7);
\draw (2.8,0.3)-- (3.,0.3);
\draw (3.,0.3)-- (3.,0.5);
\draw (3.,0.5)-- (2.8,0.5);
\draw (2.8,0.5)-- (2.8,0.3);
\draw (3.7,-0.5)-- (3.9,-0.5);
\draw (3.9,-0.5)-- (3.9,-0.3);
\draw (3.9,-0.3)-- (3.7,-0.3);
\draw (3.7,-0.3)-- (3.7,-0.5);
\draw [dash pattern=on 1pt off 1pt] (4.5,0.2)-- (4.7,0.2);
\draw [dash pattern=on 1pt off 1pt] (4.7,0.2)-- (4.7,0.4);
\draw [line width=0.4pt,dash pattern=on 1pt off 1pt] (4.7,0.4)-- (4.5,0.4);
\draw [dash pattern=on 1pt off 1pt] (4.5,0.4)-- (4.5,0.2);
\draw [->] (4.5,0.3) -- (4.2,0.4);
\draw [->] (4.7,0.3) -- (4.927540460073136,0.2990304259450714);
\draw [->] (3.8090593036946485,-0.3) -- (3.9266257446539066,-0.10426218614825172);
\draw [->] (3.7,-0.4043659959653071) -- (3.434703004850695,-0.4012721422558529);
\draw [->] (2.305446400899934,-0.5) -- (2.2992586934810255,-0.20017165114133123);
\draw [->] (2.2,-0.6) -- (1.8,-0.6);
\draw [->] (3.,0.3990667405960931) -- (3.2544064991614032,0.3990667405960931);
\draw [->] (2.8,0.3990667405960931) -- (2.4703969705798694,0.16889880559967932);
\end{tikzpicture}
\caption{Illustration of the proof} 
\label{fig:3steps}

\end{figure}

First recall that we can replace safety objectives by co-B\"uchi objectives and that the intersection of co-B\"uchi objectives is a co-B\"uchi objective. Moreover we can assume that each kind of objective in $\{\Reach,\Buchi,\CoBuchi\}$ appears among the $\Obj_m$'s (for instance the absence of a reachability objective can be replaced by the presence of the objective $\Reach(V)$). Therefore one can assume $\Obj$ is the intersection of one generalized reachability objective, one generalized B\"uchi objective and one co-B\"uchi objective:
\begin{eqnarray} \label{eq:ObjReg}
\Obj &=& \GenReach(U_1, \ldots, U_{j-1}) \cap \GenBuchi(U_{j}, \ldots, U_{i-1}) \cap \CoBuchi(U_i)
\end{eqnarray}
with $1 < j < i$.


In a first step, we compute the winning set $X_1 = \WinG{1}{\GenBuchi(U_{j}, \ldots, U_{i-1}) \cap \CoBuchi(U_i)}{G}$ of player~$1$. By Theorem~\ref{thm:reg}, $X_1$ can be computed in polynomial time, memoryless strategies suffice for player~$2$ whereas player~$1$ needs polynomial memory strategies. Moreover, by Corollary~\ref{cor:reg}, $X_1$ is 2-closed and $X_2 = V \ssetminus X_1$ is 1-closed (see Figure~\ref{fig:3steps}). We can thus consider the subgame structure $\G' = \G[X_1]$.

In a second step, we compute the winning set $Y_1 = \WinG{1}{\GenReach(U_1, \ldots, U_{j-1})}{G'}$ of player~$1$. By Theorem~\ref{thm:reg}, deciding whether a vertex is in $Y_1$ is in $\sf PSPACE$ and exponential memory strategies suffice for both players. By Corollary~\ref{cor:reg}, 
$Y_2 = X_1 \ssetminus Y_1$ is 1-closed in $G'$, and $Y_2 \cup X_2$ is 1-closed in $G$.

Let us show that 
\begin{eqnarray}
Y_1 &\subseteq& \WinG{1}{\Obj}{G}. \label{eq:incl1} \\
X_2 \cup Y_2 &\subseteq& \WinG{2}{\overline{\Obj}}{G},  \label{eq:incl2}
\end{eqnarray}
and that both players have exponential memory strategies.
As $X_2$, $Y_2$, and $Y_1$ form a partition of $V$, the two previous inclusions are equalities, thus showing that the \prob\ is in $\sf PSPACE$ (more precisely, if we recall the reduction of Proposition~\ref{prop:WISLtoWreg} used at the beginning of the proof, with an algorithm in $O(2^n \cdot (|V|+|E|))$ time). 

We begin with inclusion (\ref{eq:incl1}). Let us consider the next exponential memory strategy $\sigma_1$ for player~$1$. Given $v \in Y_1$, player~$1$ first uses an exponential memory winning strategy in the game $G'$  for the objective $\GenReach(U_1, \ldots, U_{j-1})$. Then as soon as all the sets $U_1, \ldots, U_{j-1}$ have been visited, he switches to a polynomial memory strategy that is winning for the objective $\GenBuchi(U_{j}, \ldots, U_{i-1}) \cap \CoBuchi(U_i)$ in the game $G$. Let $\sigma_2$ be any strategy of player~$2$ in $\G$ and let $\rho = \Out(v,\sigma_1,\sigma_2)$. As $X_1$ is 2-closed in $G$, $\rho$ cannot leave $X_1$ (hence $\sigma_2$ can be seen as any strategy in  $G'$). Therefore, $\rho$ is winning for the objective $\GenReach(U_1, \ldots, U_{j-1})$ in $G$. When each $U_1, \ldots, U_{j-1}$ has been visited by $\rho$, as it cannot leave $X_1$ and by definition of $\sigma_1$, $\rho$ is also winning in $G$ for the objective $\GenBuchi(U_{j}, \ldots, U_{i-1}) \cap \CoBuchi(U_i)$. This completes the proof of inclusion (\ref{eq:incl1}) with exponential memory winning strategies for player~$1$.

Let us turn to inclusion (\ref{eq:incl2}). $(i)$ Given $v \in X_2$, player~$2$ has a memoryless winning strategy for the objective $\overline{\GenBuchi(U_{j}, \ldots, U_{i-1}) \cap \CoBuchi(U_i)}$ in the game $G$. This strategy is thus winning for the objective $\overline{\Obj}$. $(ii)$ Given $v \in Y_2$, player~$2$ has an exponential winning strategy for the objective $\overline{\GenReach(U_1, \ldots, U_{j-1})}$ in $G'$. As $Y_2 \cup X_2$ is 1-closed in $G$, against any strategy of player~$1$ in $G$, 
\begin{itemize}
\item either the play stays in $Y_2$ and it is then winning for player~$2$ for the objective $\overline{\GenReach(U_1, \ldots, U_{j-1})}$ in $\G$ and thus also for the objective $\overline{\Obj}$, 
\item or the play goes to $X_2$ in which case player~$2$ uses a strategy like in $(i)$, and the play is again winning for player~$2$. 
\end{itemize}
This complete the proof of inclusion (\ref{eq:incl2}) with exponential memory winning strategies for player~$2$.

It remains to prove that the \prob\ is ${\sf PSPACE}$-hard and that exponential memory strategies are necessary for both players. This follows from Proposition~\ref{prop:WRegtoWISL} and the fact that deciding the winner in generalized reachability games is ${\sf PSPACE}$-complete and exponential memory strategies are necessary for both players (Theorem~\ref{thm:reg}).
\qed\end{proof}

In the next corollary, we present several refinements (on the number of occurrences of each kind of measure) of Theorem~\ref{thm:ISL}. When there is at most one $\Sup$, we get a polynomial fragment and in certain cases, players can play memoryless.

\begin{corollary} \label{cor:ISL}
Let $(\G, \Obj)$ be an $\n$-weighted game such that $\Obj = \cap_{m=1}^{\n} \Obj_m$ with $\Obj_m \in \ISL$ for all $m$.  
\begin{itemize}
\item If there is at most one $m$ such that $\Obj_m = \Sup$, then the \prob\ is $\sf P$-complete (with an algorithm in  $O(n^2 \cdot (|V|+|E|) \cdot |E|)$ time), player~$1$ can play with polynomial memory strategies and player $2$ can play memoryless.
\item If there is no $m$ such that $\Obj_m = \Sup$ and there is exactly at most one $m$ such that $\Obj_m =  \LimSup$, then the \prob\ is $\sf P$-complete (with an algorithm in $O((|V|+|E|) \cdot |E|)$ time) and player $1$ can play memoryless.
\item If there is exactly one $m$ such that $\Obj_m = \Sup$ and there is no $m$ such that $\Obj_m \in \{\LimInf,\LimSup\}$, then the \prob\ is $\sf P$-complete (with an algorithm in $O(|V|+|E|)$ time) and player $1$ can play memoryless.  
\end{itemize}
\end{corollary}

The proof of this corollary uses refinements of the proof of Theorem~\ref{thm:ISL}, except for the last item that needs a separate proof. 

\begin{proof}
We begin with the proof of the first two items. As in the proof of Theorem~\ref{thm:ISL}, we replace safety objectives by co-B\"uchi objectives and the intersection of co-B\"uchi objectives by one co-B\"uchi objective. As there is at most one $m$ such that $\Obj_m = \Sup$, we get an intersection like in (\ref{eq:ObjReg}) where $\GenReach(U_1, \ldots, U_{j-1})$ is replaced by $\Reach(U_1)$ or it disappears. With the previous proof of Theorem~\ref{thm:ISL} and by Theorem~\ref{thm:reg}, we can now compute $\WinG{1}{\Obj}{G} = Y_1$ in polynomial time with an algorithm in $O(n^2 \cdot (|V|+|E|) \cdot |E|)$ time. $\sf{P}$-hardness follows from Theorem~\ref{thm:reg} and reduction of Proposition~\ref{prop:WRegtoWISL}.

Let us study the strategies of both players.
\begin{itemize}
\item Suppose that there is at most one $m$ such that $\Obj_m = \Sup$. In the previous proof of inclusion (\ref{eq:incl2}), when $v \in Y_2$, player~$2$ can now play memoryless by Theorem~\ref{thm:reg}. 
\item Suppose that there is no $m$ such that $\Obj_m = \Sup$ and there is exactly at most one $m$ such that $\Obj_m = \LimSup$. Then (\ref{eq:ObjReg}) is replaced by $\Obj =  \Buchi(U_{1}) \cap \CoBuchi(U_2)$. With the proof of Theorem~\ref{thm:ISL} and by Theorem~\ref{thm:reg}, we see that we have an algorithm in $O((|V|+|E|) \cdot |E|)$ time and that player~$1$ can now play memoryless.
\end{itemize}

We now turn to the third item of the corollary: suppose that there is exactly one $m$ such that $\Obj_m = \Sup$ and there is no $m$ such that $\Obj_m \in \{\LimInf,\LimSup\}$. In this case, we cannot replace safety objectives by co-B\"uchi objectives as in the first part of the proof. Nevertheless, we can replace the intersection of safety objectives by one safety objective, and the objective $\Obj$ is now equal to $\Reach(U_1) \cap \Safe(U_2)$. In a first step, we compute the winning set $X_1 = \WinG{1}{\Safe(U_2)}{G}$ of player~$1$ for the objective $\Safe(U_2)$, and let $X_2 = V \setminus X_1$. As $X_1$ is $2$-closed by Corollary~\ref{cor:reg}, we can consider the subgame structure $G' = G[X_1]$. In a second step, we compute the winning set $Y_1 = \WinG{1}{\Reach(U_1)}{G'}$ of player~$1$, and let $Y_2 = X_1 \setminus Y_1$. Let us now explain that with arguments similar to the ones used to prove Theorem~\ref{thm:ISL}, we have that $Y_1 = \WinG{1}{\Obj}{G}$ and that memoryless strategies are sufficient for player~$1$. Indeed, since $X_1$ is $2$-closed and as $X_1 \subseteq U_2$, any memoryless winning strategy of player~$1$ for the objective $\Reach(U_1)$ in $G'$ is also winning for $\Obj$ in $G$. This shows that $Y_1 \subseteq \WinG{1}{\Obj}{G}$. It remains to show that $X_2 \cup Y_2 \subseteq \WinG{2}{\overline{\Obj}}{G}$. If $v \in X_2$, player~$2$ plays a memoryless winning strategy for the objective $\overline{\Safe(U_2)}$ in $\G$, he is thus winning for $\overline{\Obj}$ from $v$. If $v \in Y_2$, player~$2$ plays a memoryless winning strategy for the objective $\overline{\Reach(U_1)}$ in $\G'$, and if the play goes to $X_2$, he plays a memoryless winning strategy for the objective $\overline{\Safe(U_2)}$ in $\G$. We have again that player~$2$ is winning for $\overline{\Obj}$ from $v$. The proposed algorithm is in $O(|V|+|E|)$ time.

Finally, $\sf P$-hardness of the \prob\ follows from Proposition~\ref{prop:WRegtoWISL} and Theorem~\ref{thm:reg}.
\qed\end{proof}

\begin{table}
\begin{center}
\begin{tabular}{|c|c|c|c|c|c|c|c|}
\hline
$\Inf$ & $\Sup$  & $\LimInf$ & $\LimSup$ & Complexity class &    Algorithmic complexity    & Player~$1$ memory & Player~$2$ memory \\
\hline
any    & any        & any          & any            & $\sf PSPACE$-complete & $O(2^n \cdot (|V|+|E|))$  & exponential memory      & exponential memory \\
any    & $\leq 1$ & any          & any            & $\sf P$-complete           &  $O(n^2 \cdot (|V|+|E|) \cdot |E|)$ & polynomial memory      & memoryless \\
any    & 0            & any          & $\leq 1$      & $\sf P$-complete           & $O( (|V|+|E|) \cdot |E|)$ & memoryless        & memoryless \\
any    & $1$       & 0              & 0                & $\sf P$-complete           & $O(|V|+|E|)$ & memoryless        & memoryless \\
\hline
\end{tabular}
\end{center}
\caption{Overview of properties for the intersection of objectives in $\ISL$}
\label{table:ISL}
\end{table}

Table~\ref{table:ISL} summarizes all the possibilities provided by Theorem~\ref{thm:ISL} (first row of the table) and Corollary~\ref{cor:ISL} (next rows of the table). In Example~\ref{ex:limitcases}, we show that the results of Corollary~\ref{cor:ISL} are optimal with respect to the required memory (no memory, finite memory) for the winning strategies.

\begin{example} \label{ex:limitcases}
First, we come back to the game structure $\G$ depicted on Figure~\ref{Ex:Firstexample}, where we only consider the second and the third dimensions.
Assume $\Obj = \Sup(0) \cap \LimSup(0)$. Then, $v_0$ is winning for player~$1$ but memory is required to remember if player~$1$ has visited the edge $(v_1,v_1)$. The same argument holds for $\Obj = \Sup(0) \cap \LimInf(0)$ and $\Obj = \Sup(0) \cap \Sup(0)$. This example with objective $\Obj = \Sup(0) \cap \LimSup(0)$ indicates that player~$1$ cannot win memoryless in a game as in the second row of Table~\ref{table:ISL}. This example with objective $\Obj = \Sup(0) \cap \LimSup(0)$ (resp. $\Obj = \Sup(0) \cap \LimInf(0)$, $\Obj = \Sup(0) \cap \Sup(0)$) also shows that player~$1$ needs memory to win if $[1,0,0]$\footnote{This triple refers to the second, third and fourth columns of Table~\ref{table:ISL}.} in the last row of Table~\ref{table:ISL} is replaced by $[1,0,1]$ (resp. $[1,1,0]$, $[2,0,0]$).

Now, assume that additionally $v_2 \in V_1$ in the previous game structure $\G$ (that is, $\G$ is a one-player game) and let $\Obj = \LimSup(0) \cap \LimSup(0)$. Again, $v_0$ is winning but player~$1$ needs memory since he has to alternate between $v_1$ (and take the self loop) and $v_2$. This shows that in the third row of Table~\ref{table:ISL}, if $[\leq 1]$ is replaced by $[2]$ then player~$1$ needs memory to win.

Finally, consider the game depicted on Figure~\ref{ex:Player2memory}. Let $\Obj = \Sup(0) \cap \Sup(0)$. Vertex $v_0$ is losing for player~$1$ (i.e. winning for player~$2$), but player~$2$ needs memory since he has to know which edge player~$1$ took from $v_0$ to counter him by taking the edge with the same vector of weights from $v_3$. This shows that in the second row of Table~\ref{table:ISL}, if $[\leq 1]$ is replaced by $[2]$ then player~$2$ needs memory to win.

\begin{figure}[h]
\centering
  \begin{tikzpicture}[scale=4]
    \everymath{\scriptstyle}
    \draw (0,0) node [circle, draw] (A) {$v_0$};
    \draw (0.5,0.25) node [circle, draw] (B) {$v_1$};
    \draw (0.5,-0.25) node [circle, draw] (C) {$v_2$};
    \draw (1,0) node [rectangle, inner sep = 5pt, draw] (D) {$v_3$};
    \draw (1.5,0.25) node [circle, draw] (E) {$v_4$};
    \draw (1.5,-0.25) node [circle, draw] (F) {$v_5$};
    
    \draw[->,>=latex] (A) to node[above,midway,sloped] {$(-1,0)$} (B);
    \draw[->,>=latex] (A) to node[below,midway,sloped] {$(0,-1)$} (C);
     \draw[->,>=latex] (B) to node[above,midway,sloped] {$(-1,-1)$} (D);
    \draw[->,>=latex] (C) to node[below,midway,sloped] {$(-1,-1)$} (D);
    \draw[->,>=latex] (D) to node[above,midway,sloped] {$(-1,0)$} (E);
    \draw[->,>=latex] (D) to node[below,midway,sloped] {$(0,-1)$} (F);
    
    \draw[->,>=latex] (E) .. controls +(45:0.4cm) and +(315:0.4cm) .. (E) node[right,midway] {$(-1,-1)$};
    \draw[->,>=latex] (F) .. controls +(45:0.4cm) and +(315:0.4cm) .. (F) node[right,midway] {$(-1,-1)$};
	\path (0,0.2) edge [->,>=latex] (A);    
    
    \end{tikzpicture}
\caption{Example where player $2$ needs memory}
\label{ex:Player2memory}
\end{figure}
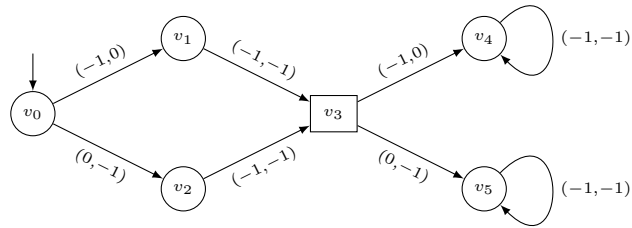
\end{example}

\section{Beyond the intersection of objectives} \label{sec:beyond}

Up to now, we have studied the \prob\ for $\n$-weighted games $(\G,\Obj)$ such that $\Obj$ is the intersection of $n$ objectives among $\WISL$. In this section, we go further by considering any Boolean combination of objectives instead of an intersection. The next theorem is a generalization of Theorem~\ref{thm:general}. It shows that, given a Boolean combination in a DNF or CNF form of objectives in $\WISL$, the complexity and the memory requirements do not blow up: the \prob\ is still $\mathsf{EXPTIME}$-complete and exponential memory strategies are sufficient for both players. Notice that one can assume w.l.o.g. that the dimension of the game is equal to the number of objectives that appear in the Boolean combination (by making copies of components of the weight function).

\begin{theorem}\label{thm:boolWISL}
Let $(\G, \Obj)$ be an $nd$-weighted game such that $\Obj = \cup_{k = 1}^d \cap_{m = 1}^n \Obj_{k,m}$ with $\Obj_{k,m} \in \WISL$ for all $k,m$. Then, the \prob\ is $\mathsf{EXPTIME}$-complete (with an algorithm in $O(n^{d(d+2)} \cdot |V|^{d+1} \cdot |E| \cdot (\lambda^2 \cdot W)^{nd(d+2)})$ time), and exponential memory strategies are both sufficient and necessary for both players. The same result holds when $\Obj = \cap_{k = 1}^d \cup_{m = 1}^n \Obj_{k,m}$. 
\end{theorem}

To prove this theorem, we need to establish the next property, which is a corollary of Proposition~\ref{prop:WISLtoBC}.

\begin{corollary} \label{cor:WISLtoBCc}
Each $\n$-weighted game $(\G,\Obj)$ with $\Obj = \cap_{m=1}^{\n} \overline{\Obj}_m$ such that for all $m$, $\Obj_m \in \WISL$ can be exponentially reduced to a game $(\G',\Obj')$ with $O(|V| \cdot (\lambda^2 \cdot W)^n)$ vertices and $O(|E| \cdot (\lambda^2 \cdot W)^n)$ edges, and $\Obj' = \cap_{m=1}^{\n} \overline{\Obj}_m'$ such that for all $m$, $\Obj'_m \in \{\Buchi,\CoBuchi\}$.
 A memoryless (resp.  finite-memory) strategy in $G'$ transfers to a finite-memory strategy in $G$.
\end{corollary}

The proof immediately follows from the one of Proposition~\ref{prop:WISLtoBC}: one just needs to take the opposite of objectives $\Obj'_m$ of Proposition~\ref{prop:WISLtoBC}. Note that as the opposite of a B\"uchi (resp. co-B\"uchi) objective is a co-B\"uchi (resp. B\"uchi) objective, the game $(G',\Obj')$ of Corollary~\ref{cor:WISLtoBCc} is a generalized-B\"uchi $\cap$ co-B\"uchi game, as in Proposition~\ref{prop:WISLtoBC}.

\begin{proof}[of Theorem~\ref{thm:boolWISL}]
First, by Theorem~\ref{thm:general}, we know that for DNF/CNF Boolean combinations of objectives, the \prob\ is $\mathsf{EXPTIME}$-hard and that both players require exponential memory strategies.
Let us show that the \prob\ is in $\mathsf{EXPTIME}$ and that exponential memory strategies are sufficient for both players.

$(i)$ Suppose first that the objective is in DNF form, i.e. $\Obj = \cup_{m = 1}^d \cup_{k = 1}^n \Obj_{k,m}$. Using the construction of Proposition~\ref{prop:WISLtoBC}, we can reduce the game $(G,\Obj)$ to an unweighted game $(G',\Obj')$ with $|V'| = O(|V| \cdot (\lambda^2 \cdot W)^{nd})$ vertices and $|E'| = O(|E| \cdot (\lambda^2 \cdot W)^{nd})$ edges, such that $\Obj' = \cup_{k=1}^d \cap_{m=1}^n \Obj'_{k,m}$ with $\Obj'_{k,m} \in \{\Buchi,\CoBuchi\}$. Then, $\Obj'$ is the disjunction of $d$ generalized B\"uchi $\cap$ co-B\"uchi objectives. Now, it is easy to adapt the polynomial reduction of generalized B\"uchi games to B\"uchi games to obtain from $(G',\Obj')$ a game $(G'',\Obj'')$ with $(n^d \cdot |V'|)$ vertices and $(n^d \cdot |E'|)$ edges, such that $\Obj''$ is the disjunction of $d$ B\"uchi $\cap$ co-B\"uchi objectives (one just needs to have $d$ counters instead of one and to define properly the new sets for the co-B\"uchi objectives in $G''$).

Therefore, we finally get a game $(G'',\Obj'')$ with $O(n^d \cdot |V| \cdot (\lambda^2 \cdot W)^{nd})$ vertices and $O(n^d \cdot |E| \cdot (\lambda^2 \cdot W)^{nd})$ edges such that $\Obj''$ is a Rabin objective with $d$ pairs. Let us recall that a Rabin game with $d$ pairs can be solved in $O(x^{d+1} \cdot y \cdot dd!)$ time where $x$ is the number of vertices and $y$ is the number of edges, that memoryless strategies (resp. strategies with memory in $O(d!)$) are sufficient for player~$1$ (resp. player~$2$) \cite{PitermanP06}. 
Thus, $G''$ can be solved in $O(n^{d(d+2)} \cdot |V|^{d+1} \cdot |E| \cdot (\lambda^2 \cdot W)^{nd(d+2)})$ time. Finally, one can deduce from 
Proposition~\ref{prop:WISLtoBC} that exponential memory are sufficient for both players.

$(ii)$ Now, suppose that the objective is in CNF form, i.e. $\Obj = \cap_{m = 1}^d \cup_{k = 1}^n \Obj_{k,m}$. Let us take the point of view of player~$2$. We have $\overline{\Obj} = \cup_{m = 1}^d \cap_{k = 1}^n \overline{\Obj}_{k,m}$ with $\Obj_{k,m} \in \WISL$. As done in $(i)$, now using the construction of Corollary~\ref{cor:WISLtoBCc}, we can reduce the game $(G,\Obj)$ to an unweighted game $(G',\Obj')$ such that $\overline{\Obj'}$ is a disjunction of $d$ generalized B\"uchi $\cap$ co-B\"uchi objectives. Using the same arguments as in $(i)$, it suffices to solve (with the point of view of player~$2$) a Rabin objective with $d$ pairs. The result follows as the game is determined.
\qed\end{proof}

\begin{remark}
Whereas the \prob\ for Boolean combinations of $\underline{\mathsf{MP}}$ and $\overline{\mathsf{MP}}$ is undecidable~\cite{Velner15}, it is here decidable for Boolean combinations of objectives in $\WISL$. Furthermore, one can show that this problem is in $\sf EXPSPACE$ using a result of~\cite{AlurTM03} as follows. With constructions of Proposition~\ref{prop:WISLtoBC} and Corollary~\ref{cor:WISLtoBCc}, we build an exponential game $(G',\Obj')$ where $\Obj'$ is a Boolean combination of B\"uchi and co-B\"uchi objectives. Thanks to  \cite{AlurTM03}, games with an objective defined as a Boolean combination of formulas of the form ``infinitely often $p$'' can be solved in $\sf PSPACE$. Since the game $G'$ is exponential, we obtain the desired result.
\end{remark}

\section{Parameterized complexity} \label{sec:para}

We here provide a complete complexity analysis of our results when we fix the number $n$ (resp. $nd$) of dimensions when we consider intersections (resp. DNF/CNF Boolean combinations) of objectives in $\WISL$.

\begin{theorem}
Let $(G, \Obj)$ be a multi-weighted game such that
\begin{itemize}
\item[$(a)$] $\Obj = \cap_{m=1}^n \Obj_m$, with $\Obj_m \in \WISL$,
\item[$(b)$] $\Obj = \cap_{m=1}^n \Obj_m$, with $\Obj_m \in \ISL$,
\item[$(c)$] $\Obj = \cup_{k=1}^d \cap_{m=1}^n \Obj_{k,m}$, with $\Obj_{k,m} \in \WISL$.
\item[$(d)$] $\Obj = \cap_{k=1}^d \cup_{m=1}^n \Obj_{k,m}$, with $\Obj_{k,m} \in \WISL$.
\end{itemize}
If $n$ is fixed (as well as $nd$ in $(c)$ and $(d)$), then the \prob\ is
\begin{itemize}
\item[$(a,c,d)$] $\sf{EXPTIME}$-complete (with a pseudo-polynomial algorithm in $O(|V| \cdot |E| \cdot (\lambda^2 \cdot W)^{2n})$ time for $(a)$ and in $O(|V|^{d+1} \cdot |E| \cdot (\lambda^2 \cdot W)^{nd(d+2)})$ time for $(c,d)$.
\item[$(b)$] $\sf{P}$-complete (with a polynomial algorithm in $O(|V|+|E|)$ time).

\end{itemize}
\end{theorem}

\begin{proof}
For items $(a)$, $(c)$ and $(d)$, the $\sf{EXPTIME}$-hardness follows from Theorem~\ref{thm:WMP}. The algorithms given in Theorems~\ref{thm:general} and \ref{thm:boolWISL} are now pseudo-polynomial since $n$ (as well as $nd$ in items $(c, d)$) is a fixed-parameter. Indeed, the algorithm for $(a)$ is in $O(|V| \cdot |E| \cdot (\lambda^2 \cdot W)^{2n})$ time and the algorithm for $(c, d)$ is in $O(|V|^{d+1} \cdot |E| \cdot (\lambda^2 \cdot W)^{nd(d+2)})$ time.

Thanks to Theorem~\ref{thm:ISL}, we deduce for item $(b)$ a polynomial algorithm in $O(2^n \cdot (|V|+|E|))$ time. $\sf{P}$-hardness follows from Theorem~\ref{thm:reg} and reduction of Proposition~\ref{prop:WRegtoWISL}.
\qed\end{proof}
\section{Conclusion} \label{sec:conclu}

In this paper, we have studied games with Boolean combinations of heterogeneous measures in $\WISL$. The \prob\ is ${\sf EXPTIME}$-complete for DNF/CNF Boolean combinations  and exponential memory strategies are necessary and sufficient for both players. We have focused on games with intersections of such measures and given a detailed study of the complexity and the memory requirements: $\sf{EXPTIME}$-completess and memory requirements still hold when the measures are taken among $\WISL$, and we get $\sf{PSPACE}$-completeness when $\DFWMP$ is not considered. 
In case of intersections of one $\DFWMP$ objective with any number of objectives of one kind among $\ISL$ (this number must be fixed in case of objectives $\Sup$), we get $\sf P$-completeness (for polynomial windows) and pseudo-polynomial strategies for both players. In case of no occurrence of $\DFWMP$ measure, we have proposed several refinements for which we got $\sf{P}$-completeness and lower memory requirements.

Let us conclude with some future work. $(i)$ Knowing that Boolean combination of mean-payoff objectives is undecidable, we have initiated in this paper the study of games with heterogeneous combination of $\omega$-regular measures: $\Inf$, $\Sup$, $\LimInf$, $\LimSup$, and $\DFWMP$. We aim at extending the study to measures that are not omega-regular, such as mean-payoff, discounted-sum, energy. $(ii)$ The \prob\ for arbitrary (instead of DNF/CNF) Boolean combinations of objectives in $\WISL$ is in ${\sf EXPSPACE}$ and ${\sf EXPTIME}$-hard (see Section~\ref{sec:general}). The exact complexity should be settled. $(iii)$ The study of Section~\ref{sec:OneWindow} could be nicely extended with the intersection of one WMP and several other objectives of different kinds. $(iv)$ A window version of parity objectives have been introduced in~\cite{ChatterjeeHH09} and studied for unidimensional games. It could be interesting to investigate multidimensional games mixing this $\omega$-regular objective with the objectives of this paper.

\medskip
\textbf{Acknowledgments. } We would like to thank Mickael Randour for his availability and his valued help.

\bibliographystyle{abbrv}
\bibliography{biblio}

\end{document}